\documentclass[a4paper,twocolumn,9pt,unpublished]{quantumsans}
\pdfoutput=1
\usepackage[english]{babel}
\usepackage[T1]{fontenc}
\usepackage[utf8]{inputenc}
\usepackage{textcomp}

\usepackage{amsmath,amsthm}
\usepackage[dvipsnames]{xcolor}
\usepackage[square,numbers,sort&compress]{natbib}
\usepackage[font={footnotesize,sf}]{caption}
\usepackage{braket}
\usepackage[inline]{enumitem}  

\usepackage{tabularx}
\usepackage{booktabs}
\usepackage{xfrac}
\usepackage[ruled,vlined]{algorithm2e}
\usepackage[scaled=0.8]{FiraMono}

\usepackage{microtype}
\usepackage[bitstream-charter]{mathdesign}
\usepackage{anyfontsize}

\SetMathAlphabet{\mathcal}{normal}{OMS}{cmsy}{m}{n}
\SetMathAlphabet{\mathcal}{bold}  {OMS}{cmsy}{b}{n}

\usepackage{hyperref}
\hypersetup{
    colorlinks = true,
    linkcolor  = blue,
    citecolor  = blue,
    urlcolor   = blue,
    pdftitle   = Many-body contextuality
}
\usepackage{xurl}

\usepackage{tikz}
\usetikzlibrary{decorations.pathreplacing, shapes.geometric, arrows.meta, positioning, shadows, fit, backgrounds}
\graphicspath{{img}}

\usepackage{tikz-cd}

\definecolor{zpauli}{HTML}{1F77B4}

\tikzset{
  boxed node/.style={
    rectangle,
    rounded corners,
    fill=white,
    fill opacity=1,
    inner sep=2pt
  },
  gray lines/.style={
    every path/.style={thick,draw=lightgray,opacity=0.75,line cap=round},
    every node/.style={boxed node}
  }
}

\newcommand{\Zop}[1]{%
  \fill[zpauli, opacity=1]
    ($(#1)+(-2pt,-2pt)$)
    rectangle
    ($(#1)+( 2pt, 2pt)$);
}

\definecolor{mpsgray}{RGB}{228, 228, 228}

\newcommand{\mpstensor}[2]{%
    \draw [thick, fill=mpsgray, rounded corners] 
      ({-0.75+#1}, {-0.75+#2}) 
      rectangle 
      ({0.75+#1}, {0.75+#2});%
}

\usetikzlibrary{math}
\tikzmath{\s = sin(30);} 
\definecolor{xpauli}{HTML}{FF0000}
\definecolor{zpauli}{HTML}{1F77B4}

\newtheorem{theorem}{Theorem}
\newtheorem{proposition}{Proposition}
\newtheorem{lemma}{Lemma}[section]
\newtheorem{corollary}{Corollary}
\newtheorem{Def}{Definition}[section]

\theoremstyle{definition}
\newtheorem{Example}{Example}[section]
\newtheorem{remark}{Remark}[section]

\usepackage{thmtools}
\usepackage{thm-restate}

\DeclareMathOperator{\Image}{Im}    
\DeclareMathOperator{\dist}{dist}   
\DeclareMathOperator{\weight}{wt}   
\DeclareMathOperator{\Tr}{Tr}       
\DeclareMathOperator{\Sp}{Sp}       
\DeclareMathOperator{\supp}{supp}   
\DeclareMathOperator{\coim}{coim}   
\DeclareMathOperator{\sub}{sub}     

\def\U#1{\mathrm{U}({#1})}
\def\Z#1{\mathbb{Z}_{#1}}
\def\F#1{\mathbb{F}_{#1}}
\def\SU#1{\mathrm{SU}(#1)}
\def\GL#1#2{\mathrm{GL}_{#1}(#2)}

\newcommand{\Reals}{\mathbb{R}}
\renewcommand{\vec}[1]{\mathbf{#1}}
\newcommand{\inner}[2]{\langle #1, #2 \rangle}
\newcommand{\BF}[1]{\mathsf{BF}_{#1}}
\newcommand{\ReedMuller}[2]{\mathsf{RM}({#1}, {#2})}
\newcommand{\ketbra}[2]{\ket{#1}\!\bra{#2}}
\newcommand{\tp}{\intercal}

\newcommand{\CZ}{\mathsf{CZ}}
\newcommand{\CCZ}{\mathsf{CCZ}}
\newcommand{\CX}{\mathsf{CX}}

\usepackage{mathtools}
\DeclarePairedDelimiter\abs{\lvert}{\rvert}%
\DeclarePairedDelimiter\norm{\lVert}{\rVert}%
\makeatletter
\let\oldabs\abs
\def\abs{\@ifstar{\oldabs}{\oldabs*}}
\let\oldnorm\norm
\def\norm{\@ifstar{\oldnorm}{\oldnorm*}}
\makeatother

\makeatletter
\newcommand\RedeclareMathOperator{%
  \@ifstar{\def\rmo@s{m}\rmo@redeclare}{\def\rmo@s{o}\rmo@redeclare}%
}
\newcommand\rmo@redeclare[2]{%
  \begingroup \escapechar\m@ne\xdef\@gtempa{{\string#1}}\endgroup
  \expandafter\@ifundefined\@gtempa
     {\@latex@error{\noexpand#1undefined}\@ehc}%
     \relax
  \expandafter\rmo@declmathop\rmo@s{#1}{#2}}
\newcommand\rmo@declmathop[3]{%
  \DeclareRobustCommand{#2}{\qopname\newmcodes@#1{#3}}%
}
\@onlypreamble\RedeclareMathOperator
\makeatother

\RedeclareMathOperator{\Re}{Re}
\RedeclareMathOperator{\Im}{Im}

\begin{document}

\title{Many-body contextuality and self-testing quantum matter via nonlocal games}

\author{Oliver Hart}
\affiliation{Department of Physics and Center for Theory of Quantum Matter, \href{https://ror.org/02ttsq026}{University of Colorado Boulder}, Boulder, Colorado 80309, USA}
\affiliation{\href{https://ror.org/0507j3z22}{Quantinuum}, Terrington House, 13--15 Hills Road, Cambridge CB2 1NL, UK}
\orcid{0000-0002-5391-7483}

\author{David T. Stephen}
\affiliation{Department of Physics and Center for Theory of Quantum Matter, \href{https://ror.org/02ttsq026}{University of Colorado Boulder}, Boulder, Colorado 80309, USA}
\affiliation{Department of Physics, \href{https://ror.org/05dxps055}{California Institute of Technology}, Pasadena, California 91125, USA}
\affiliation{\href{https://ror.org/03ssvsv78}{Quantinuum}, 303 S Technology Ct, Broomfield, CO 80021, USA}
\orcid{0000-0002-3150-0169}

\author{Evan Wickenden}
\affiliation{Department of Physics and Center for Theory of Quantum Matter, \href{https://ror.org/02ttsq026}{University of Colorado Boulder}, Boulder, Colorado 80309, USA}
\orcid{0009-0005-6218-505X}

\author{Rahul Nandkishore}
\affiliation{Department of Physics and Center for Theory of Quantum Matter, \href{https://ror.org/02ttsq026}{University of Colorado Boulder}, Boulder, Colorado 80309, USA}
\orcid{0000-0001-5703-6758}

\maketitle

\begin{abstract}
\onecolumn
\fontsize{9}{10}\selectfont
\sffamily Contextuality is arguably the fundamental property that makes quantum mechanics different from classical physics. It is responsible for quantum computational speedups in both magic-state-injection-based and measurement-based models of computation, and can be directly probed in a many-body setting by multiplayer nonlocal quantum games. Here, we discuss a family of games that can be won with certainty when performing single-site Pauli measurements on a state that is a codeword of a Calderbank-Shor-Steane (CSS) error-correcting quantum code. We show that these games require deterministic computation of a code-dependent Boolean function, and that the classical probability of success is upper bounded by a generalized notion of nonlinearity/nonquadraticity. This success probability quantifies the state's contextuality, and is computed via the function's (generalized) Walsh-Hadamard spectrum. To calculate this, we introduce an efficient, many-body-physics-inspired method that involves identifying the symmetries of an auxiliary hypergraph state. We compute the classical probability of success for several paradigmatic CSS codes and relate it to both classical statistical mechanics models and to strange correlators of symmetry-protected topological states. We also consider CSS \emph{submeasurement} games, which can only be won with certainty by sharing the appropriate codeword up to local isometries. These games therefore enable self-testing, which we illustrate explicitly for the 2D toric code. We also discuss how submeasurement games enable an extensive notion of contextuality in many-body states.
\end{abstract}
\vspace{15pt}

\twocolumn


\section{Introduction}

What makes a state intrinsically quantum? Richard Feynman famously claimed~\cite{FeynmanLectures} that the single-particle double-slit experiment contained the ``\emph{only} mystery'' and that it was ``impossible, \emph{absolutely} impossible, to explain in any classical way,'' and ``has in it the heart of quantum mechanics.'' With half a century of hindsight, we know this statement to have been incomplete, with the phenomenology of the double-slit experiment with a single quantum particle having been experimentally accessed in a completely classical setting~\cite{Couder, Bush}. A more compelling answer invokes {\it entanglement}~\cite{Bell, Horodeckis} as the \emph{sine qua non} of quantumness. However, even highly entangled states can be efficiently described using classical methods. For instance,
\begin{enumerate*}[label=(\roman*)]
    \item the Gottesman-Knill theorem \cite{GottesmanKnill} implies that entanglement is not sufficient for quantum computational speedup, 
    \item for a system obeying the eigenstate thermalization hypothesis (ETH) \cite{Deutsch, Srednicki, Rigol, MBLARCMP}, generic midspectrum states will have volume-law entanglement entropy, and yet expectation values of local measurements will match an essentially classical thermal state, and, perhaps most pertinently,
    \item measuring particles in an entangled Bell state along a \emph{single} axis leads to correlations that can be reproduced entirely classically~\cite{PreskillNotes}.
\end{enumerate*}
A modern diagnostic that is both necessary and sufficient for quantumness is \emph{contextuality}~\cite{kochenSpecker1967, Spekkens2005contextuality, Abramsky_2011, contextuality1, contextuality2}, where the outcomes of one measurement must depend on what other measurements are simultaneously made in a hidden-variable-model description.
Indeed, contextuality is the feature of quantum mechanics that is responsible for violations of Bell inequalities~\cite{BellTests}, and it has been shown that, in certain settings, quantum computational speedups can occur only when contextuality is present in both magic-state- and measurement-based models of quantum computation~\cite{contextuality1,Delfosse2015Wigner,Raussendorf2017contextualityQubits,Anders2009,Hoban2011stronger,RaussendorfContextuality,Abramsky2017,Frembs2018}. Yet, most studies to date have focused on contextuality in a few-body setting. 
As the number of particles grows, so does the possible complexity and richness of correlations in the associated quantum states~\cite{VerstraeteFourQubits}, making the quantification of contextuality challenging both theoretically and numerically. The question of how to practically quantify contextuality in experiment has also been brought into focus by the advent of early non-fault-tolerant quantum devices, which provide direct and controlled access to \emph{many-body} quantum systems, potentially opening the door to new experimental probes of many-body physics.

\emph{Multiplayer nonlocal quantum games}~\cite{DanielStringOrder, Daniel2022Exp, BBSgame, BulchandaniGames, lin2023quantum, DallaTorre, Hart1, hart2024braidingFTW} provide a powerful, operational interpretation of contextuality. {In a quantum game, players receive questions from a referee and, without communicating, respond with answers according to some pre-agreed strategy. Each question-answer pair receives a numerical score, and player cooperatively attempt to maximize this score by refining their pre-agreed strategy.} By sharing suitably entangled quantum states (``resource states'') prior to playing the game, players can use the quantum correlations to beat the performance of the optimal classical strategy. In other words, quantum games \emph{directly} harness contextuality of a resource state to achieve advantage at a certain task (i.e., winning the game), and so directly probe the essential ``quantumness'' of the state. Consequently, they have found applications in randomness generation~\cite{ColbeckThesis}, interactive proof systems~\cite{CleveXOR, Ji2016proofs, MIP*=RE},  quantum key distribution~\cite{Ekert1991cryptography,Barrett2005noSignaling,Vazirani2014fully}, and unconditional proofs of quantum advantage ~\cite{bravyi2020shallow,bravyi2020noisyShallow}. It turns out that quantum games also provide a natural lens through which to study many-body physics on quantum devices. This is thanks in part to deep connections between contextuality and the universal properties of quantum phases of matter, such as spontaneous symmetry breaking and fractional excitations~\cite{DanielStringOrder,Daniel2022Exp,BulchandaniGames,BBSgame,lin2023quantum,Hart1,hart2024braidingFTW}. These connections deepen our understanding of quantum phases of matter and enhance the robustness of quantum games to experimental noise thanks to the inherent stability of topological phases of matter \cite{Hart1}. In this manuscript, we will demonstrate how quantum games may be used to provide a precise quantification of contextuality in the setting of many-body quantum matter.

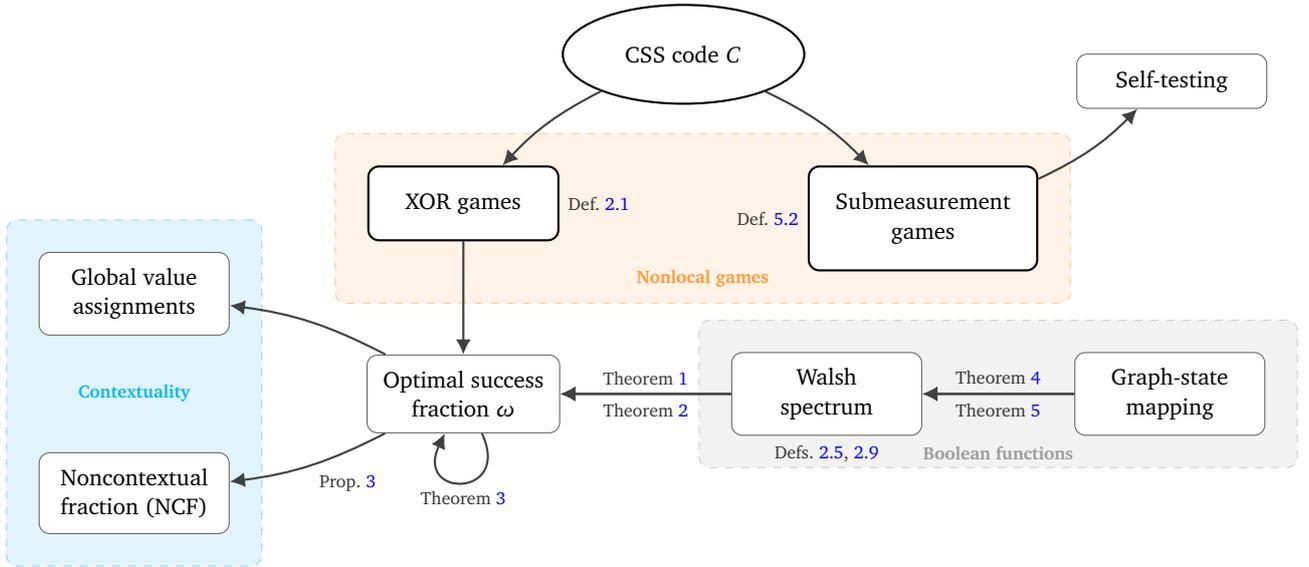
\begin{figure*}[t]
    \centering
    \begin{tikzpicture}[
        node distance = 1.0cm and 2cm, 
        base/.style = {draw, align=center, inner sep=10pt, font=\small},
        root/.style = {base, ellipse, draw=black, thick, fill=white, fill=white},
        game/.style = {base, rounded corners, fill=white, thick, minimum width=2.5cm},
        thm/.style = {base, rectangle, fill=white, dashed, minimum width=2.5cm, inner sep=6pt},
        topic/.style = {base, rectangle, rounded corners, draw=black!60, fill=white, inner sep=6pt, minimum width=2.5cm},
        arrow/.style = {-{Latex[width=2mm,length=2mm]}, thick, darkgray},
        lbl/.style = {font=\scriptsize, color=black!80}
    ]
        \node (css) [root] {CSS code $C$};

        \node (xor) [game, below left = 1cm and 0.5cm of css] {XOR games};
        \node (sub) [game, below right = 1cm and 0.5cm of css] {Submeasurement\\games};

        \node (self) [topic, above right = 0.75cm and 0.5cm of sub] {Self-testing};

        \node (omega) [topic, below = 1.5cm of xor] {Optimal success\\fraction $\omega$};

        \node (walsh) [topic, right = 2.25cm of omega] {Walsh\\spectrum};
        \node [lbl, below] at (walsh.south) {Defs.~\ref{def:Walsh-transform}, \ref{def:generalized-Walsh}};
        \node (graph) [topic, right = of walsh] {Graph-state\\mapping};        

        \node (global) [topic, above left = 0.25cm and 1.8cm of omega] {Global value\\assignments};
        \node (ncf) [topic, below left = 0.25cm and 1.8cm of omega] {Noncontextual\\fraction (NCF)};

        \begin{scope}[on background layer]
            \node [fit=(global) (ncf), 
                   fill=cyan!10, 
                   draw=cyan!40, 
                   dashed, 
                   rounded corners, 
                   inner sep=12pt, 
                   label={[text=cyan!80, font=\scriptsize\bfseries]center:Contextuality}] (ctx_box) {};
        \end{scope}

        \begin{scope}[on background layer]
            \node [fit=(xor) (sub), 
                   fill=orange!10, 
                   draw=orange!40, 
                   dashed, 
                   rounded corners, 
                   inner sep=12pt] (ctx_box) {};
            \node at (ctx_box.south) [anchor=south, yshift=3pt, font=\scriptsize\bfseries, text=orange!80] {Nonlocal games};
        \end{scope}

        \begin{scope}[on background layer]
            \node [fit=(walsh) (graph), 
                   fill=gray!10, 
                   draw=gray!40, 
                   dashed, 
                   rounded corners, 
                   inner sep=12pt] (ctx_box) {};
            \node at (ctx_box.south) [anchor=south, yshift=0pt, font=\scriptsize\bfseries, text=gray!80] {Boolean functions};
        \end{scope}

        \draw [arrow] (css) edge[bend right=10] (xor);
        \draw [arrow] (css) edge[bend left=10] (sub);

        \node [lbl, right] at (xor.east) {Def.\ \ref{def:css-game}};
        \node [lbl, left] at (sub.west) {Def.\ \ref{def:submeasurement-game}};

        \draw [arrow] (xor) -- (omega);
        \draw [arrow] (graph) -- node[lbl, above] {Theorem \ref{thm:walsh-from-symmetries}} (walsh);
        \draw [arrow] (graph) -- node[lbl, below] {Theorem \ref{thm:bell-pairs}} (walsh);

        \draw [arrow] (walsh) -- node[lbl, above] {Theorem \ref{thm:nonlinearity}} (omega);
        \draw [arrow] (walsh) -- node[lbl, below] {Theorem \ref{thm:CSS-success-fraction}} (omega);
        
        \draw [arrow] (omega) edge[bend right=10] (global);
        \draw [arrow] (omega) edge[bend left=10] node[lbl, below right, align=left] {Prop.\ \ref{prop:NCF-bound}} (ncf);

        \draw [arrow] (omega) to [out=295, in=245, looseness=5] node[below, lbl] {Theorem \ref{thm:clifford-equivalence}} (omega);
        \draw [arrow] (sub) edge[bend right=10] (self);
    \end{tikzpicture}
    \caption{Overview of the structure of the technical results of this paper. Broadly, the paper first covers XOR games. Section \ref{sec:XORgames} introduces XOR games, and connects the optimal classical success fraction to properties of Boolean functions (via theorems~\ref{thm:nonlinearity}, \ref{thm:CSS-success-fraction}, and \ref{thm:clifford-equivalence}) and to contextuality of Pauli measurements. Section~\ref{sec:walsh-from-graphs} introduces many-body-physics-inspired techniques for computing the Walsh spectrum via theorems~\ref{thm:walsh-from-symmetries} and \ref{thm:bell-pairs}, while Sec.~\ref{sec:examples} puts these tools to use by computing the success fraction of various CSS XOR games. Finally, Sec.~\ref{sec:submeasurement} introduces submeasurement games and connects them to self-testing.}
    \label{fig:paper-structure}
\end{figure*}

The basic idea behind our work is the following: Suppose there exists a task (game) that can be executed successfully (won) with optimal probability $\omega^*$ given a resource state $\rho$, and single-site Pauli measurement access thereto. Further suppose that without access to quantum resources the probability of success (optimized over all possible classical strategies) can be upper bounded by $\omega \leq \omega^*$. The difference $\omega^* - \omega$ then provides a quantification of the degree of contextuality present in the state with respect to single-site Pauli measurements. Maximizing $\omega^* - \omega$ over all possible tasks (games) then gives us a quantification of the contextuality present in the many-body state $\rho$, and observing a success fraction greater than $\omega$ serves as an experimental witness of contextuality (see also Ref.~\cite{Abramsky2017}). Quantifying contextuality in this way has two desirable features. First, it directly relates to the practical usefulness of contextuality, i.e., the ability of a state to serve as a resource for demonstrating unconditional quantum advantage. Second, it is exactly the kind of quantity that is natural to measure in quantum devices (by actually playing the game), serving as a crisp demonstration of the new ways in which modern quantum devices allow us to probe many-body physics. 

Now, the optimizations over all possible games and over all possible classical strategies may seem like a formidably difficult task. However, we demonstrate that, for certain classes of games, the optimization over classical strategies can be mapped to the evaluation of the degree of nonlinearity of a Boolean function, which can be efficiently accomplished using tools from discrete mathematics (Walsh-Hadamard transforms). Meanwhile, we also demonstrate how to accomplish the optimization over possible games for the special case of resource states that are codewords of Calderbank-Shor-Steane (CSS) codes. To do this, we introduce a mapping from nonlinear Boolean functions to auxiliary hypergraph states~\cite{Rossi_2013}. This mapping makes use of the connection between nonlinear Boolean functions and hypergraphs, and grants access to powerful concepts in many-body theory. By establishing relations to classical statistical mechanics models, we efficiently compute bounds on $\omega$ and make connections to strange correlators in symmetry-protected topological states~\cite{ScaffidiGapless2017,YouStrangeCorrelator2014}. We also provide a polynomial-time algorithm for extracting maximally contextual sets of operators (and hence optimal games that maximize $\omega^* - \omega$) for CSS codewords. 

Finally, we introduce the notion of CSS \emph{submeasurement} games (see Refs.~\cite{Daniel2022Exp, Meyer2023submeasurement} for related work), which can only be won with certainty by sharing the appropriate CSS codeword (up to local isometries), i.e., they enable a so-called \emph{self-test} of the state (see, e.g., Ref.~\cite{Supic2020selftestingof} for a review). Our construction is translation invariant, in contrast to the standard approach for stabilizer/graph states~\cite{Supic2020selftestingof}. We illustrate this with a submeasurement game for the 2D toric code, and discuss how (unlike existing protocols tailored towards the toric code~\cite{Baccari2020}) the translational invariance of submeasurement games enables us to construct a natural and properly extensive quantification of contextuality.

This paper is structured as follows (see also Fig.~\ref{fig:paper-structure}).
In Sec.~\ref{sec:XORgames} we introduce CSS XOR games, which generalize and unify games based on the GHZ state~\cite{ghz1989,GHSZ1990, MerminPolynomials, Brassard2003multiparty,brassard2005recasting,brassard2005pseudotelepathy}, 1D cluster state~\cite{Daniel2022Exp}, and toric codes~\cite{BBSgame,Hart1,hart2024braidingFTW}. Interpreting the success condition as deterministic computation of a certain Boolean function $f$, we show that the success fraction of any classical strategy can be upper bounded by a quantity closely related to the Walsh-Hadamard transform of $f$. We discuss how this enables a practical quantification of contextuality in many-body states. In Sec.~\ref{sec:walsh-from-graphs} we explain how the (generalized) Walsh spectrum -- a key component of the analysis in Sec.~\ref{sec:XORgames} -- can be computed by identifying the $X$ symmetries of an auxiliary hypergraph state, and also how one may use this hypergraph-state formalism to efficiently extract maximally contextual sets of operators for a quantum strategy. In Sec.~\ref{sec:examples} we illustrate the ideas introduced in the previous sections with explicit computations for several paradigmatic states, including the GHZ state, the 1D cluster state, and the 2D toric code. In Sec.~\ref{sec:submeasurement} we introduce the idea of \emph{submeasurement games}, which can only be won given a specific resource state (up to isometries), and which thereby enable self-testing. We introduce a submeasurement game for the 2D toric code and argue that it self-tests the state. We conclude in Sec.~\ref{sec: discussion} with a discussion of some open directions. 

\textbf{Notation.} Throughout the paper, boldface symbols such as $\vec{x}$ denote \emph{row} vectors. The vector space of Boolean functions in $n$ variables is denoted by $\BF{n}\cong \F{2}^{2^n}$ while classical order-$r$ binary Reed-Muller codes with codewords of length $2^n$ are denoted $\ReedMuller{r}{n}$. If not otherwise specified, when referring to the \emph{degree} of a Boolean function, we mean the algebraic degree. Finally, while we distinguish between affine and linear Boolean functions, the word \emph{nonlinearity} (see Def.~\ref{def:nonlinearity}) refers to the distance to the nearest affine Boolean function, as is common.


\section{XOR games}
\label{sec:XORgames}

The family of games introduced in this section corresponds to a subset of so-called XOR games~\cite{CleveXOR}. The games we introduce can be won with probability one, $\omega^* = 1$, when players have access to shared quantum resources, namely codewords of CSS quantum error-correcting codes.
They generalize the GHZ game~\cite{ghz1989,GHSZ1990,MerminPolynomials,Brassard2003multiparty,brassard2005recasting,brassard2005pseudotelepathy}, the cubic Boolean function game~\cite{Daniel2022Exp} (won by the 1D cluster state), and also the toric code games from Refs.~\cite{Hart1,hart2024braidingFTW}. The definition specifically is inspired by the Bell inequalities presented in Ref.~\cite{Guhne2005bell} and the cubic Boolean function game~\cite{Daniel2022Exp}.

\begin{Def}[CSS XOR game]
\label{def:css-game}
Associate a player to every qubit of a CSS code on $N$ qubits defined by the pair of {full-rank} parity-check matrices $H = (H_X, H_Z)$. Player $i$ is handed an ordered pair of bits $(a_i, b_i) \in \F{2}^2$, with $\vec{a} \in I_X \subseteq \Image H_X^\tp$ and $\vec{b} \in I_Z \subseteq \Image H_Z^\tp$, where neither $I_X$ or $I_Z$ is equal to the empty set. The inputs are collectively denoted by the pair $I = (I_X, I_Z)$. Once the players have been handed their bits, any classical communication between players is forbidden. Player $i$ then outputs a bit $y_i(a_i, b_i)$. Collectively, the players attempt to satisfy  
\begin{equation}
    \sum_{i=1}^{N} y_i(a_i, b_i) \equiv \frac12 \sum_{i=1}^N a_i b_i \quad {(\text{mod 2})}
    \, .
    \label{eqn:CSS-game-victory}
\end{equation}
This defines $\mathsf{CSS}(H | I)$. If the inputs are unrestricted, i.e., $I = (\Image H_X^\tp, \Image H_Z^\tp)$, we refer to this game as $\mathsf{CSS}(H)$.
\end{Def}

The pair $(a_i, b_i)$ is referred to as the \emph{question} asked to player $i$ and the output $y_i$ is the player's \emph{answer}. Each pair of bitstrings $(\vec{a}, \vec{b})$, i.e., the collection of all questions, is referred to as a \emph{query}. By definition, elements of $\Image H_X^\tp $ and $\Image H_Z^\tp$ are equal to
\begin{equation}
    \label{eqn:input-bits}
    \vec{a}(\vec{x}) = \vec{x}H_X \quad  \text{and} \quad \vec{b}(\vec{z}) = \vec{z}H_Z,
\end{equation}
respectively, with multiplication over $\F{2}$, for all $\vec{x}^{\tp}$ ($\vec{z}^{\tp}$) in the domain of the parity-check matrix $H_X^{\tp}$ ($H_Z^{\tp}$). A referee computes the bits $a_i(\vec{x})$ and $b_i(\vec{z})$ from Eq.~\eqref{eqn:input-bits} and hands them to player $i$. The win condition~\eqref{eqn:CSS-game-victory} is then written $\bigoplus_i y_i = f_H(\vec{x}, \vec{z})$ for some parity function $f_H$ of the pair $(\vec{x}, \vec{z})$ with codomain $\{0, 1\}$, determined uniquely by the parity-check matrices $H = (H_X, H_Z)$ of the CSS code. The function $f_H$ is referred to as the \emph{target function}.

\begin{remark}
    Physically, each query~\eqref{eqn:input-bits} corresponds to a stabilizer of the input CSS code ($X^\vec{a}Z^{\vec{b}}$) up to a sign, and the players are tasked with outputting this sign.
\end{remark}

The query $(\vec{a}, \vec{b})$ and the parity bit are both determined by the pair $(\vec{x}, \vec{z})$. That the right-hand side of Eq.~\eqref{eqn:CSS-game-victory} is an integer follows from commutation of the stabilizer generators, i.e., $\sum_i a_i b_i \equiv \vec{x}(H_X H_Z^\tp) \vec{z}^\tp \equiv 0$ mod 2. The optimal probability of achieving Eq.~\eqref{eqn:CSS-game-victory}, when queries are drawn uniformly at random, is denoted by $\omega$ and $\omega^*$ when players employ classical and quantum strategies, respectively. To convert the victory condition~\eqref{eqn:CSS-game-victory} to a Boolean function written in algebraic normal form, one makes use of the following Lemma.

\begin{Def}[Hamming weight]
    The \emph{Hamming weight} $\weight(\vec{b})$ of a length-$n$ bitstring $\vec{b}$ is equal to the number of nonzero entries in $\vec{b}$, i.e., $\weight(\vec{b}) = \sum_{i=1}^n b_i$.
\end{Def}

\begin{lemma}[Ref.~\cite{Hoban2011}, Eq.~(4)]
    \label{lem:oplus-to-reals}
    Let $\vec{x} \in \F{2}^n$ be an $n$-bit string. The parity of $\vec{x}$, equal to $\bigoplus_{i} x_i$, is expressed using integer arithmetic by
    \begin{equation}
        \bigoplus_{i=1}^n x_i = \sum_{\vec{b} \in \F{2}^n \setminus \vec{0}} (-2)^{\weight(\vec{b})-1}
        \prod_{i=1}^n x_i^{b_i} .
        \label{eqn:oplus-to-reals}
    \end{equation}
\end{lemma}

For example, $x_1 \oplus x_2 = x_1 + x_2 - 2x_1x_2$. A proof of this Lemma and an explicit expression for the Boolean function $f_H(\vec{x}, \vec{z})$ for generic parity-check matrices are presented in appendix~\ref{app:computed-function}. By direct calculation, the function $f_H$ has algebraic degree at most three.

\begin{Example}[GHZ game~\cite{ghz1989,GHSZ1990,MerminPolynomials,Brassard2003multiparty,brassard2005recasting,brassard2005pseudotelepathy}]
\label{ex:ghz-game}
    Consider a 1D GHZ state on $N$ qubits, which will define an $N$-player game.
    The state is stabilized by, e.g.,
    \begin{equation}
        \label{eqn:GHZ-stabilizers}
        \prod_{i=1}^N X_i, \quad \{ Z_i Z_{i+1} : i \in \mathbb{N}, 1 \leq i < N \} ,
    \end{equation}
    which determine the parity-check matrices $H_X$ and $H_Z$. Respectively associate the bits $x$ and $z_i$ to the stabilizer generators in Eq.~\eqref{eqn:GHZ-stabilizers}. Then, player $i$ is handed $(a_i, b_i) = (x, z_{i-1} \oplus z_{i})$, where we define $z_0=z_N=0$. Using Lemma~\ref{lem:oplus-to-reals}, the win condition~\eqref{eqn:CSS-game-victory} reduces to
    \begin{equation}
        \bigoplus_{i=1}^N y_i = \bigoplus_{i=1}^{N-2}  x z_i z_{i+1} \oplus \bigoplus_{i=1}^{N-1} x z_i .
        \label{eqn:ghz-function}
    \end{equation}
    The standard GHZ-game win conditions result from restricting the inputs such that $x=1$, leaving the bits $\{ z_i \}$ unconstrained. For $N=3$ players, and given this restricted input set, the win condition reduces to computation of the Boolean function \texttt{or} ($\vee$):
    \begin{equation}
        \bigoplus_{i=1}^{3} y_i = z_1 \vee z_2 = z_1z_2 \oplus z_1 \oplus z_2  .
        \label{eqn:ghz-quadratic-function}
    \end{equation}
\end{Example}


\subsection{Quantum strategies}
\label{sec:quantum-strats}

In a \emph{quantum strategy}, prior to being asked questions, the players share a quantum state $\ket{\psi}$. Projective measurement of this state conditioned on their question can inform their answer. As previously mentioned, CSS XOR games possess perfect quantum strategies, $\omega^*=1$ (although they are not the only family of games known to exhibit this property, see, e.g., Refs.~\cite{Brassard1999cost,AVIS_2006} and Refs.~\cite{benedetti2025complementsampling, Benedetti2025unconditional}).
To prove statements about the performance of different quantum strategies, the following lemma about measurement of Pauli operators is useful.

\begin{lemma}
    \label{lem:submeasurement-dist}
    Let $\Lambda$ be the set of all lattice sites and let $I \subseteq \Lambda$. Suppose that the Pauli operator $P_i \in \{ I, X, Y, Z \}$ is measured on site $i$ for all $i \in \Lambda$, producing the measurement outcome $(-1)^{s_i}$. If the system is in the state $\ket{\psi}$ then
    \begin{equation}
        \Pr\left(\bigoplus_{i \in I} s_i = f\right)  =
        \frac12  \left( 1 + (-1)^{f} \braket{\psi |  \prod_{i \in I}  P_i   | \psi} \right)
        \label{eqn:submeasurement-dist}
    \end{equation}
    with $f \in \F{2}$.
\end{lemma}

\begin{proof}
    See the proof of Lemma~\ref{lem:outcome-distribution-general}, given in Appendix~\ref{app:outcome-dist-proof}, which is strictly more general.
\end{proof}

The distribution of $\bigoplus_{i \in I} s_i$ is determined by the expectation value of the Pauli string formed from the product of Pauli measurements restricted to the set of sites $I$. If $\prod_{i \in I}  P_i \ket{\psi} = \pm \ket{\psi}$ then $\bigoplus_{i \in I} s_i$ assumes a definite value with respect to $\ket{\psi}$ determined by the sign.


\subsubsection{Pauli strategy}

\begin{proposition}[Pauli strategy]
    \label{prop:perfect-quantum-strat}
    Let $\ket{\psi}$ be a codeword of the CSS code specified by parity-check matrices $H=(H_X$, $H_Z)$.
    Single-site Pauli measurements $P(a_i, b_i)$ in $\ket{\psi}$ enable a perfect strategy for the game $\mathsf{CSS}(H)$.
\end{proposition}

\begin{proof}
    Let $P(a, b) = i^{ab} X^a Z^b \in \{ I, X, Y, Z \}$ and apply Lemma~\ref{lem:submeasurement-dist} to $\prod_{i\in \Lambda} P_i(a_i, b_i)$. Using Eq.~\eqref{eqn:input-bits}, this operator is equal to the stabilizer $X^\vec{a}Z^{\vec{b}}$ of $\ket{\psi}$ up to the sign $\frac12 \sum_i a_i b_i$ mod 2.
\end{proof}

For a generic quantum resource state $\ket{\psi}$, which may have weight outside of the codespace, and averaging over all inputs to the game with equal weight, the probability of winning, $p_\text{q}$, using the strategy presented in the proof of Prop.~\ref{prop:perfect-quantum-strat} is
\begin{equation}
    p_\text{q}(\rho) = \frac{1}{2}\left( 1 + \Tr[ \rho \Pi_0 ] \right)
    \, ,
    \label{eqn:fidelity-win-probability}
\end{equation}
where $\rho = \ket{\psi}\bra{\psi}$ and the projector $\Pi_0$ is onto the codespace of the CSS code. This result is analogous to estimating the fidelity via random stabilizer measurements~\cite{Flammia2011fidelity,daSilva2011practical}.

The probability of winning when the players implement the Pauli strategy from Prop.~\ref{prop:perfect-quantum-strat} is therefore directly determined by the probability that $\ket{\psi}$ belongs to the codespace. Note that Eq.~\eqref{eqn:fidelity-win-probability} continues to hold when the resource state $\rho$ is impure.


\subsubsection{MERP strategy}

Since only the sum of answers $\bigoplus_i y_i$ is constrained by the win condition, it is also possible for the players to win by sharing an $N$-particle GHZ state using \emph{maximal-entanglement, relative-phase} strategies (MERP)~\cite{Watts2019}. The general strategy mirrors that of the GHZ game.

\begin{proposition}
A CSS XOR game associated to a CSS code on $N$ qubits can be won using a shared $N$-qubit GHZ state.
\end{proposition}

\begin{proof}
    Consider the strategy in which player $i$ measures $X_i$ if $a_ib_i=0$ and measures $Y_i$ if $a_i b_i = 1$. Since $\bigoplus_i a_i b_i = 0$ by commutation of stabilizer generators, they collectively measure a GHZ stabilizer up to a sign. By Lemma~\ref{lem:submeasurement-dist}, the sign is given by $\frac12 \sum_i a_i b_i$ mod 2.
\end{proof}


\subsection{Classical strategies}

Here we connect the properties of the Boolean function that enters the win condition~\eqref{eqn:CSS-game-victory} to the performance of deterministic classical strategies. 

\begin{Def}
    A \emph{deterministic classical strategy} is a collection of functions $y_i : \F{2}\times\F{2} \to \F{2}$ for each $i \in \{1, \dots, N\}$, implemented by player $i$ on $(a_i, b_i)$.
\end{Def}

If the inputs to the game are parametrized by the pair $(\vec{x}, \vec{z})$, players implementing a deterministic classical strategy therefore output the function 
\begin{equation}
    c(\vec{x}, \vec{z}) \coloneq \bigoplus_{i=1}^N y_i(a_i(\vec{x}), b_i(\vec{z})).
\end{equation}
Note that the function $c$ depends implicitly on both the classical deterministic strategy implemented by the players \emph{and} the parity-check matrices, which encode the map from $(\vec{x}, \vec{z})$ to the query $(\vec{a}, \vec{b})$. Suppose that $\dim(\Image H_X^\tp) + \dim(\Image H_Z^\tp) = d$ such that the number of unique queries in the game $\mathsf{CSS}(H)$ is $2^d$. The \emph{success fraction} of $c$, given a target function $f$ is defined by the fraction of inputs on which the functions $c, f \in \smash{\BF{d}}$ agree with one another. It is given by
\begin{equation}
    p(\vec{c} | \vec{f}\,) = \frac{2^d - \weight(\vec{c} \oplus \vec{f})}{2^d}
    ,
    \label{eqn:winning-fraction}
\end{equation}
and satisfies $0 \leq p(\vec{c} | \vec{f}\,) \leq 1$,
where we dropped the explicit dependence of $f_H$ [defined by Eq.~\eqref{eqn:CSS-game-victory}] on the parity-check matrices $H$, and $\vec{c}, \vec{f} \in \smash{\F{2}^{2^d}}$ are vectorized versions of the functions $c, f \in \smash{\BF{d}}$, respectively. If inputs to the game $(\vec{x}, \vec{z})$ are drawn uniformly at random, the success fraction also equals the probability that the deterministic classical strategy $c$ wins the game $\mathsf{CSS}(H)$.

The \emph{optimal success fraction} (optimized over all deterministic classical strategies) defines $\omega$. By convexity, the quantity $\omega$ bounds both deterministic {and} probabilistic classical strategies. In the following, we connect $\omega$ to properties of the Boolean function $f_H$. First, we consider games where either the $X$ or the $Z$ inputs are held fixed, then we consider the case of unrestricted inputs.


\subsubsection{Fixing inputs}

Motivated by the GHZ example, consider an input set $\smash{I_X}$ that contains a single element, i.e., $\vec{a} \in \smash{I_X} = \{ \vec{a}_0 \}$ for some fixed $\vec{a}_0 \in \Image H_X^\tp$ such that, in the associated Pauli strategy (Prop.~\ref{prop:perfect-quantum-strat}), the stabilizer $\prod_{i}X_i^{a_i}$ is then always measured by the players. Optimizing over classical strategies leads to the following result.

\begin{theorem}
    \label{thm:nonlinearity}
     Let $G=\mathsf{CSS}(H | I)$ be the CSS game defined by full-rank parity-check matrices $H = (H_X$, $H_Z)$, where the $X$ inputs are fixed, $I_X = \{ \vec{a}_0 \}$, and $Z$ inputs, $I_Z = \Im H_Z^\tp$ with $\abs{I_Z} = 2^n$, are unconstrained. Then the optimal success fraction satisfies $\omega(G) = 1-2^{-n}N_f$ where $N_f$ is the nonlinearity of $f$, the target function of $G$.
\end{theorem}

The nonlinearity $N_f$ is defined below in Def.~\ref{def:nonlinearity}.
To prove Theorem~\ref{thm:nonlinearity}, we construct the possible deterministic classical strategies, which correspond to functions
\begin{equation}
    c(\vec{z}) = u_0 \oplus \inner{\vec{u}_z}{\vec{z}},
    \label{eqn:classical-strat-fixed}
\end{equation}
parametrized by $u_0 \in \F{2}$, $\vec{u}_z \in \F{2}^n$ with $n = \dim( \Im H_Z^\tp )$, where $\inner{\vec{a}}{\vec{b}} = \vec{a}\vec{b}^\tp$. Equation~\eqref{eqn:classical-strat-fixed} shows that classical strategies therefore correspond to \emph{affine} Boolean functions. Since the parity-check matrix $H_Z$ is assumed to be full rank, it is also surjective, and classical strategies span \emph{all} affine Boolean functions in the $n$ variables $ z_i $.

\begin{remark}
    {If player $i$ answers with $y_i = e_i\oplus s_{i} b_i$, we have $\vec{u}_z^\tp = H_Z \vec{s}^\tp$, which we identify with the error syndrome associated to the state $\vec{s}$. Two states that differ by the addition of an element of $\ker H_Z $ give rise to the same classical strategy, so distinct classical strategies are identified with the quotient space $\coim H_Z $.}
\end{remark}

Different affine functions may have different success fractions~\eqref{eqn:winning-fraction}, and in order to find $\omega$ we are interested in the classical strategies that agree with $f$ on as many inputs as possible. This notion is captured by the nonlinearity $N_f$ of the function $f$, which identifies the distance between $f$ and the nearest affine Boolean function.

\begin{Def}
    \label{def:nonlinearity}
    The \emph{nonlinearity} $N_f$ of a Boolean function $f$ in $n$ variables is equal to the minimum Hamming distance between $f$ and the set of all affine Boolean functions 
    \begin{equation}
        N_f \coloneq \min_{a\in\F{2}, \vec{b} \in \F{2}^n} \dist(f, L_{a,\vec{b}})
    \end{equation}
    where the function $L_{a,\vec{b}}(\vec{x}) = a \oplus \langle \vec{b}, \vec{x} \rangle$ parametrizes all affine Boolean functions on $n$ variables, and $\dist(f, g) = \weight(\vec{f} \oplus \vec{g})$ in terms of the vectorized functions $\vec{f}$, $\vec{g}$.
\end{Def}

Combining Eq.~\eqref{eqn:classical-strat-fixed} and Def.~\ref{def:nonlinearity} allows us to straightforwardly prove Theorem~\ref{thm:nonlinearity}.

\begin{proof}[Proof of Theorem~\ref{thm:nonlinearity}]
    The success fraction of an arbitrary deterministic classical strategy is, by definition, $p(c | f) = 1 - 2^{-n} \dist(c , f)$. Maximizing over all strategies gives
    \begin{equation}
        \omega(G) = \max_c p(c | f) = 1 - 2^{-n} \min_c \dist(c, f).
    \end{equation}
    Since the space of classical strategies corresponds to the space of affine Boolean functions, it immediately follows from Def.~\ref{def:nonlinearity} that $\omega(G) = 1 - 2^{-n} N_f$.
\end{proof}

A Boolean function $f$'s nonlinearity is closely related to its \emph{Walsh spectrum}, obtained by taking the Walsh-Hadamard transform of $f$. This connection will be crucial for the mapping to hypergraph states in Sec.~\ref{sec:walsh-from-graphs}.

\begin{Def}
    \label{def:Walsh-transform}
    The \emph{Walsh-Hadamard (WH) transform} of a Boolean function $f : \F{2}^n \to \F{2}$ is
    \begin{equation}
        \label{eqn:Walsh-transform}
        W_f(\vec{y}) = \sum_{\vec{x} \in \F{2}^n} (-1)^{\langle \vec{x}, \vec{y} \rangle \oplus f(\vec{x})}
        ,
    \end{equation}
    and the \emph{Walsh spectrum} of $f$ is the set of integer-valued coefficients $W_f = \{ W_f(\vec{y}) : \vec{y} \in \F{2}^n \}$.
\end{Def}

An individual Walsh coefficient $W_f(\vec{y})$ encodes the overlap between $f(\vec{x})$ and the linear Boolean function $\inner{\vec{x}}{\vec{y}}$. From the definition of the WH transform, the Walsh coefficients obey~\cite{TokarevaBent}
\begin{subequations}
\begin{align}
    W_f(\vec{y}) &= 2^n - 2 \dist( f, \langle \vec{y}, \, \cdot \, \rangle ) \\
    &= 2 \dist( f, 1 \oplus \langle \vec{y}, \, \cdot \, \rangle ) - 2^n.
\end{align}%
\label{eqn:Walsh-distance-relations}%
\end{subequations}
It therefore follows that the nonlinearity can be expressed in terms of the largest (in magnitude) Walsh coefficient:
\begin{equation}
    N_f = 2^{n-1} - \frac12 \max_{\vec{y} \in \F{2}^n} \abs{W_f(\vec{y})}.
    \label{eqn:nonlinearity-from-walsh}
\end{equation}
Note that, by virtue of Eq.~\eqref{eqn:Walsh-distance-relations}, by taking the modulus, we bound the distance to all \emph{affine} Boolean functions.
It is also clear from Def.~\ref{def:Walsh-transform} that the Walsh spectrum remains unchanged upon adding a linear function to $f$.
Since the Walsh-Hadamard transform~\eqref{eqn:Walsh-transform} is a $\Z{2}$ Fourier transform, the Walsh spectrum satisfies a sum rule (Parseval's theorem)
\begin{equation}
    \sum_{\vec{y}\in\F{2}^n} W_f(\vec{y})^2 = 2^{2n}.
\end{equation}
The maximal Walsh coefficient therefore satisfies a lower bound corresponding to the case where all terms under the sum are equal, $\abs*{W_f(\vec{y})} = 2^{n/2}$. This places an upper bound on the allowable nonlinearity of Boolean functions. Maximally nonlinear functions are known as \emph{bent functions}, which have applications throughout discrete mathematics, communications theory, and theoretical computer science~\cite{mesnager2016bent, TokarevaBent, KaiUweCombinatorics}.

\begin{Def}[Bent functions]
    \label{def:bent}
    A Boolean function $f$ is a \emph{bent function} in $n$ variables (with $n$ even) if any of the following equivalent conditions hold:
    \begin{itemize}
        \item the nonlinearity equals $N_f  = 2^{n-1} - 2^{n/2-1}$,
        \item the Walsh spectrum $W_f$ has constant absolute value.
    \end{itemize}
\end{Def}

We note that Theorem~\ref{thm:nonlinearity} is equivalent to Ref.~\cite[Lemma 1]{RaussendorfContextuality} in the context of measurement-based quantum computation (MBQC). Indeed, the CSS XOR games can be viewed as temporally flat MBQC, which is ensured by the lack of communication between players. In the next section, we generalize to the case of unrestricted inputs, which allows the players to implement strategies with geometrically restricted nonlinearity.

Using the manipulations above that relate the nonlinearity $N_f$ to the Walsh spectrum, we can directly connect the success fraction $\omega(G)$ of the game $G$ from Theorem~\ref{thm:nonlinearity} to the largest Walsh coefficient:
\begin{equation}
    \omega(G) = \frac{1}{2}\left( 
        1 + 2^{-n} \max_{\vec{y} \in \F{2}^n} \abs{W_f(\vec{y})}
    \right) ,
\end{equation}
where bent functions minimize the right-hand side.


\subsubsection{Unrestricted inputs}

When no restrictions are placed on queries, we have $\vec{a} \in I_X = \Image H_X^\tp $ and $\vec{b} \in I_Z = \Image H_Z^\tp $, which are sampled uniformly. In this case, classical strategies correspond to \emph{bi-affine}, as opposed to affine, Boolean functions. Consequently, Theorem~\ref{thm:nonlinearity} must be generalized appropriately to a quantity beyond the nonlinearity. We bound the optimal success fraction both from above and from below and relate it to the target function's nonquadraticity.

\begin{Def}
    \label{def:bi-affine}
    A Boolean function $f : \F{2}^n \times \F{2}^m \to \F{2}$ is \emph{bi-affine} if for all $\vec{z} \in \F{2}^m$ the map $\vec{x} \mapsto f(\vec{x}, \vec{z})$ is affine, and for all $\vec{x} \in \F{2}^n$ the map $\vec{z} \mapsto f(\vec{x}, \vec{z})$ is affine.
\end{Def}

\begin{Def}
    Given a linear subspace of functions $\mathscr{G} \subseteq \BF{d}$, the \emph{$\mathscr{G}$-distance} of $f$, $\nu_f(\mathscr{G})$, of a Boolean function $f : \F{2}^d \to \F{2}$ is defined as
    \begin{equation}
        \nu_f(\mathscr{G}) \coloneq \min_{g \in \mathscr{G}} \dist(g, f).
        \label{eqn:generalized-nonlinearity}
    \end{equation}
\end{Def}

\begin{Def}
    \label{def:generalized-Walsh}
    Given a linear subspace of functions $\mathscr{G} \subseteq \BF{d}$, the \emph{generalized Walsh coefficient} $W_f[g]$ quantifies the overlap between $g \in \mathscr{G}$ and some fixed Boolean function $f : \F{2}^d \to \F{2}\,$,
    \begin{equation}
        W_f[g] \coloneq \sum_{\vec{x} \in \F{2}^d} (-1)^{g(\vec{x}) \oplus f(\vec{x})}
        .
        \label{eqn:generalized-Walsh-def}
    \end{equation}
\end{Def}

It follows that $W_f[g] = 2^d - 2 \dist( g , f )$ quantifies the Hamming distance between $g \in \mathscr{G}$ and $f \in \BF{d}$. One can therefore express the $\mathscr{G}$-distance of $f$ in terms of the generalized Walsh spectrum. If $\mathscr{G}$ contains the constant function $1(\vec{x})$, then $g \in \mathscr{G}$ implies $1\oplus g \in \mathscr{G}$, and
\begin{equation}
    \nu_f(\mathscr{G}) = 2^{d-1} - \frac12 \max_{g \in \mathscr{G} / \langle 1 \rangle} \abs{W_f[g]}.
\end{equation}
with $\langle 1 \rangle$ the subspace formed by constant functions. By taking the modulus, we need only sum over equivalence classes $\mathscr{G} / \langle 1 \rangle$ since $W_f[g] = - W_f[1 \oplus g]$. We are now ready to state Theorem~\ref{thm:CSS-success-fraction}.

\begin{theorem}
    \label{thm:CSS-success-fraction}
     Let $G=\mathsf{CSS}(H)$ be the CSS game defined by full-rank parity-check matrices $H = (H_X, H_Z)$, with no input restrictions, $I_X = \Image H_X^\tp$ and $I_Z = \Image H_Z^\tp$. If $\abs{I_X}\abs{I_Z} = 2^d$, the optimal success fraction satisfies $\omega(G) = 1-2^{-d} \nu_f(\mathscr{C})$ with $\mathscr{C}(H)$ the space of classical deterministic functions. 
\end{theorem}

\begin{proof}
    Making use of Defs.~\ref{eqn:generalized-nonlinearity} and \ref{eqn:generalized-Walsh-def},
    the proof is analogous to that of Theorem~\ref{thm:nonlinearity}.
\end{proof}

We now describe the linear subspace of classical deterministic strategies.
Let $\abs{I_X} = 2^n$ and $\abs{I_Z} = 2^m$ and write $n+m=d$.
Summing up the answers from all players, deterministic classical strategies correspond to functions
\begin{equation}
    c(\vec{x}, \vec{z}) =
    u_0 \oplus \inner{\vec{u}_x}{\vec{x}} \oplus \inner{\vec{u}_z}{\vec{z}} \oplus \inner{\vec{x} u_{xz}}{\vec{z}},
    \label{eqn:unrestricted-classical-strat}
\end{equation}
with $u_0 \in \F{2}$, $\vec{u}_x \in \F{2}^n$, $\vec{u}_z \in \F{2}^m$, and $u_{xz} \in M_{n\times m}(\F{2})$. Since $H_X$ and $H_Z$ are surjective, $\inner{\vec{u}_x}{\vec{x}}$ and $\inner{\vec{u}_z}{\vec{z}}$ span all linear Boolean functions. However, $u_{xz}$ belongs to the space spanned by $H_X \Lambda H_Z^\tp$ where $\Lambda$ is a generic diagonal matrix in $M_N(\F{2})$.

Theorem~\ref{thm:CSS-success-fraction} is rather general and does not make use of the structure of the space $\mathscr{C}(H)$ of classically computable functions. The bi-affine nature of $\mathscr{C}(H)$ does, however, have important implications when bounding the $\mathscr{C}$-distance.
First, note that $\mathscr{L} \subseteq \mathscr{C}(H)$, where $\mathscr{L} \cong \ReedMuller{1}{d}$ is the space of linear functions, which can be used to bound from below on the optimal success fraction for \emph{any} target function $f$.

\begin{remark}
    Since $\mathscr{L} \subseteq \mathscr{C}$ with $\mathscr{C}$ parametrized by Eq.~\eqref{eqn:unrestricted-classical-strat}, the coefficients $\smash{W_f[c]}$ satisfy Parseval's theorem for all $u_{xz}$. This implies that $\nu_f(\mathscr{C}) \leq 2^{d-1} - 2^{d/2-1}$.
\end{remark}

Next, the bi-affine nature of $\mathscr{C}$ also allows us to bound $\nu_f(\mathscr{C})$ from below for a \emph{specific} target function $f$.
The map $\vec{z} \mapsto c(\vec{x}, \vec{z})$ is affine, such that we may write
\begin{equation}
    c(\vec{x}, \vec{z}) = c_0(\vec{x}) \oplus \langle \vec{y}(\vec{x}), \vec{z} \rangle
\end{equation}
with $c_0(\vec{x}) \coloneq u_0 \oplus \inner{\vec{u}_x}{\vec{x}}$ and $\vec{y}(\vec{x}) \coloneq \vec{u}_z \oplus \vec{x} u_{xz}$. Using Def.~\ref{def:generalized-Walsh}, we are able to write the generalized Walsh coefficient in terms of ordinary Walsh coefficients (Def.~\ref{def:Walsh-transform})
\begin{equation}
    W_f[c] \coloneq \sum_{\vec{x}, \vec{z}} (-1)^{c(\vec{x}, \vec{z}) \oplus f_\vec{x}(\vec{z})}
    = \sum_\vec{x} (-1)^{c_0(\vec{x})} W_{f_\vec{x}}(\vec{y}(\vec{x}))
\end{equation}
A simple bound on the magnitude of $W_f[c]$ is obtained by applying the triangle inequality, which allows us to use the techniques we have at our disposal for evaluating ordinary Walsh-Hadamard transforms:
\begin{equation}
    \abs{W_f[c]} \leq \sum_\vec{x} \abs{ W_{f_\vec{x}}(\vec{y}(\vec{x})) }.
    \label{eqn:walsh-triangle-inequality}
\end{equation}
Note that an analogous expression can be obtained by regarding the function $c(\vec{x}, \vec{z})$ as an affine function for each fixed $\vec{z}$. If the code is self-dual, this will lead to an equivalent bound. 

Finally, we relate the optimal success fraction $\omega$ to the \emph{nonquadraticity} (or second-order nonlinearity) of the target function, which is a natural generalization of the nonlinearity, and which also arises in the quantification of magic in a many-body setting~\cite{ManyBodyMagic}.
From Eq.~\eqref{eqn:unrestricted-classical-strat}, the space $\mathscr{C}$ of classical deterministic functions is contained in the space of quadratic Boolean functions. Also, recall that the codewords of the order-$r$ Reed-Muller code $\ReedMuller{r}{n}$ are in correspondence with Boolean functions of algebraic degree at most $r$ in $n$ variables. Let $\chi_f = \min_{g \in \ReedMuller{2}{n}} \dist(g, f)$ be the nonquadraticity of $f$. Since $\mathscr{C} \subset \ReedMuller{2}{d}$, it follows from Def.~\ref{eqn:generalized-nonlinearity} that we have $\nu_f \geq \chi_f$.

\begin{corollary}
    Let $G$, $f$ be the game and the associated target function from Theorem~\ref{thm:CSS-success-fraction}, respectively, and let $\smash{\chi_f}$ be the nonquadraticity of $f$. The optimal success fraction is bounded from above by $\omega(G) \leq 1 - 2^{-d}\chi_f$.
\end{corollary}


\subsection{Code equivalence}

CSS games require that players collectively compute some target function, $f_H$. Here, we establish that 
\begin{enumerate*}[label=(\roman*)]
    \item basis changes naturally realize the notion of \emph{affine equivalence}
    \item equivalent stabilizer codes realize the notion of \emph{extended affine equivalence}.
\end{enumerate*}
All such equivalent functions have the same optimal classical success fraction $\omega$.

\begin{Def}[Equivalent Boolean functions]
    \label{def:ext-affine-equivalence}
    Two Boolean functions $f$ and $g$ in $n$ variables are \emph{affine equivalent} if there exists a matrix $A \in \GL{n}{\F{2}}$ and a bitstring $\vec{b} \in \F{2}^n$ such that $f(\vec{x}) = g(\vec{x}A \oplus \vec{b})$ for all $\vec{x} \in \F{2}^n$.
    Two Boolean functions $f$ and $g$ in $n$ variables are instead \emph{extended affine equivalent} if, in addition to $\vec{b} \in \F{2}^n$ and $A \in \GL{n}{\F{2}}$, there exists an affine function $h$ in $n$ variables such that $f(\vec{x}) = g(\vec{x}A \oplus \vec{b})\oplus h(\vec{x})$ for all $\vec{x} \in \F{2}^n$.
\end{Def}

Consider a basis change of the stabilizer generators: $H'_X = AH_X$ and $H'_Z = B H_Z$ for nonsingular matrices $A, B \in \GL{n}{\F{2}}$. The target function associated to this new basis $H' = (H'_X, H'_Z)$ is 
\begin{equation}
    f_{H'}(\vec{x}, \vec{z}) = f_H(\vec{x}A, \vec{z}B).
\end{equation}
Hence, basis changes modify the target function by a linear transformation on the inputs. The optimal success fraction $\omega(G)$ appearing in Theorems~\ref{thm:nonlinearity} and \ref{thm:CSS-success-fraction} is trivially invariant under such a transformation. Next, we consider the notion of \emph{equivalent stabilizer codes.}

\begin{Def}[Equivalent stabilizer codes]
    Two qubit stabilizer codes are said to be \emph{equivalent} if the codespaces of the two codes are related by a tensor product of elements of the single-qubit Clifford group and a qubit permutation.
\end{Def}

Considering the functions computed by measuring the codewords of equivalent qubit stabilizer codes gives rise to the following theorem.

\begin{theorem}
    \label{thm:clifford-equivalence}
    Let $U = \prod_i C_i$ with each $C_i \in \mathbf{C}_1 / \U{1}$ a single-site unitary gate belonging to the single-site Clifford group, and let $\ket{\psi}$ be a codeword of a CSS code with parity-check matrices $H_X$, $H_Z$. The function that is deterministically computed by applying the Pauli strategy (Prop.~\ref{prop:perfect-quantum-strat}) in the dressed state $U \ket{\psi}$ is related to the function computed by $\ket{\psi}$ by the addition of a classically computable function.
\end{theorem}

\begin{proof}
    See Appendix~\ref{app:equivalent-codes}, where we calculate the function computed by applying the Pauli strategy to $U\ket{\psi}$.
\end{proof}

Intuitively, write each $C_i$ as $C_i = V_i P_i$, where $P \in \{ I, X, Y, Z \}$ is a Pauli operator and $V$ is an element of $\mathbf{C}_1 / \U{1}$ modulo Pauli operators. If $U = P_i$, the target function becomes $f \to f \oplus L$ with $L$ a linear function. On the other hand, $U = V_i$ can give rise to the product $a_ib_i$ on site $i$, which corresponds to the bilinear part of classical strategies~\eqref{eqn:unrestricted-classical-strat}. Hence, under a generic product of single-site Clifford unitaries, $f \to f \oplus c'$ with $c' \in \mathscr{C}$. Using Def.~\ref{def:generalized-Walsh},
\begin{equation}
    W_{f \oplus c'}[c] = \sum_{\vec{x} \in \F{2}^d} (-1)^{c(\vec{x}) \oplus f(\vec{x}) \oplus c'(\vec{x})} = W_{f}[c \oplus c'].
\end{equation}
By the closure property of $\mathscr{C}$, this is simply a permutation of the generalized Walsh coefficients, which leaves the maximum unchanged. Hence, the optimal classical success fraction $\omega$ is again invariant under such a transformation.


\subsection{Quantifying contextuality}
\label{sec:quantifying-contextuality}

So far, we have shown that the performance of classical strategies for CSS games can be related to properties of the target Boolean function entering the win condition~\eqref{eqn:CSS-game-victory}. Here, we connect the performance of strategies involving quantum resources to contextuality and the \emph{contextual fraction}~\cite{Abramsky2017}. The former corresponds to the inability to describe the measurement outcomes obtained in, e.g., the Pauli strategy (Prop.~\ref{prop:perfect-quantum-strat}) via a global probability distribution over deterministic values assignments, whereas the latter quantifies the \emph{extent} to which such a description is possible.

Let $X = \{ O_a \}$ be the set of all possible (nonidentity) single-qubit Pauli operators entering the Pauli strategy of Prop.~\ref{prop:perfect-quantum-strat}, i.e., the union, over all queries, of the sets of all single-site measurements made by players. All $O_a$ are, by assumption, dichotomic, and we represent measurement outcomes via $(-1)^{s_a}$ with each $s_a \in \F{2}$. The measurements made by the players are collected into a set $M$ consisting of subsets $C \subset X$, where each $C \in M$ is referred to as a \emph{measurement context}. Elements $C \in M$ are in one-to-one correspondence with queries, i.e., $O_a \in C$ are the single-site measurements made by the players given the query associated to $C$. For a given context $C$, all $O_a \in C$ commute since they are supported on disjoint sites. The pair $( X, M )$ is called a \emph{measurement scenario}. See Example~\ref{ex:GHZ-measurement-scenario} for an illustration of $( X, M )$ for the GHZ game (Example~\ref{ex:ghz-game}).

For each context, $C \in M$, an \emph{empirical model} $e$ associates a probability distribution $e_C$ to measurement outcomes in $\{ 0, 1 \}^{\abs{C}}$. In order to be physical, such empirical models must satisfy certain criteria, including compatibility of marginals~\cite{Abramsky_2011, Abramsky2017}, and all empirical models arising from quantum mechanics satisfy these criteria. An empirical model $e$ is said to be \emph{contextual} if it cannot be obtained via marginalization of some global probability distribution over value assignments in $\{ 0, 1 \}^{\abs{X}}$, which are functions $\Omega : X \to \{ 0, 1 \}^{\abs{X}}$. 

\begin{remark}
    Value assignments $\Omega : X \to \{0, 1 \}^{\abs{X}}$ correspond to classical deterministic strategies. Hence, $\omega$ is interpreted as the maximum fraction of measurement contexts $C \in M$ whose sign can be correctly described by $\Omega$, modulo a global sign. 
\end{remark}

To see this, consider the constraints that must be satisfied by $\Omega$. As a consequence of Lemma~\ref{lem:submeasurement-dist}, the sum of measurement outcomes associated to the Pauli operators $O_a \in C$ in each context $C \in M$ is fixed:
\begin{equation}
    \bigoplus_{O_a \in C} \Omega(O_a) \stackrel{!}{=} f(C), \quad \forall C \in M
\end{equation}
where $f(C) = 0$ ($1$) if $\prod_{O_a \in C} O_a$ belongs to the stabilizer group $\mathcal{S}$ ($-\mathcal{S}$). These equations are equivalent to the win conditions~\eqref{eqn:CSS-game-victory}, up to a global sign, if we identify player $i$'s answer $y_i(a, b)$ with the value assignment $\Omega(P_i(a, b))$, for $ab \neq 0$, of the Pauli operators on site $i$.

To characterize the extent to which an empirical model $e$ is contextual, one introduces the noncontextual fraction by decomposing $e$ into a convex combination of two empirical models, a noncontextual part $e^{\text{NC}}$ and a contextual part $e^{\text{C}}$:
\begin{equation}
    e = \lambda e^{\text{NC}} + (1 - \lambda) e^{\text{C}} , \quad 0 \leq \lambda \leq 1.
\end{equation}
The noncontextual part, $e^{\text{NC}}$, is a mixture of global value assignments.
Maximizing $\lambda$ over possible valid decompositions of $e$ then gives its noncontextual fraction
\begin{equation}
    \mathsf{NCF}(e) \coloneq \max_{e^{\text{NC}}} \lambda 
    \eqcolon
    1 - \mathsf{CF}(e)
    ,
\end{equation}
with $\mathsf{CF}(e)$ the contextual fraction. The noncontextual fraction may be computed using the linear-programming methods outlined in Refs.~\cite{Abramsky2016,Abramsky2017}.

\begin{figure}
    \centering
    \includegraphics[width=\linewidth]{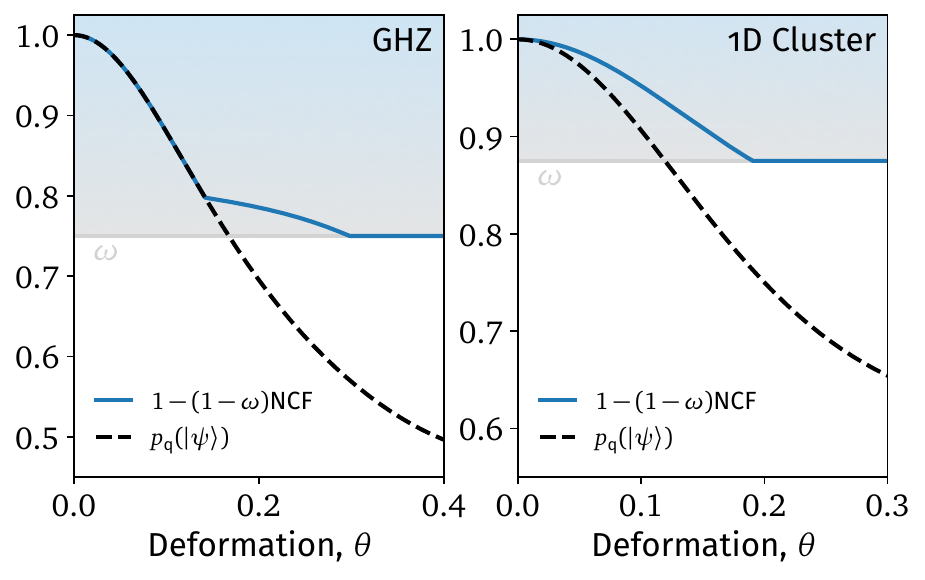}
    \caption{Upper bound~\eqref{eqn:NCF-bound} on the performance of strategies that make use of the empirical model that derives from measurement of a quantum state compared with the performance of the Pauli strategy (Prop.~\ref{prop:perfect-quantum-strat}). States are deformed by the nonunitary operator $A(\theta) = e^{\theta Z} e^{2i\theta Y}$ and the noncontextual fraction ($\mathsf{NCF}$) is computed using the linear-programming method described in Ref.~\cite{Abramsky2017}. \textbf{Left}: The deformed GHZ state $A(\theta) \ket{\text{GHZ}}$ with three qubits, for which $\omega = 3/4$. \textbf{Right}: The deformed 1D cluster state $A(\theta)\ket{C}$ with four qubits, for which $\omega = 7/8$.}
    \label{fig:ncf-vs-exact}
\end{figure}

\begin{Example}[GHZ measurement scenario]
    \label{ex:GHZ-measurement-scenario}
    The quantum strategy for the GHZ game from Example~\ref{ex:ghz-game} can be phrased as a measurement scenario $( X, M )$ described by an empirical model $e$. In this case, 
    \begin{gather*}
        X = \{ X_1, Y_1, X_2, Y_2, X_3, Y_3 \}, \\
        M = \{ \{ X_1, X_2, X_3 \}, \{ X_1, Y_2, Y_3 \}, \{ Y_1, X_2, Y_3 \}, \{ Y_1, Y_2, X_3 \} \}
    \end{gather*}
    The frequencies of various measurement outcomes with respect to the three-qubit GHZ state are captured by an empirical model $e$:
    \noindent
    \begin{center}
    \setlength{\tabcolsep}{4pt}
    \begin{tabularx}{\linewidth}{@{}lllllllll@{}}
    \toprule
                    & 000            & 001            & 010            & 011            & 100            & 101            & 110            & 111            \\ \midrule
$X_1, X_2, X_3$ & $\sfrac{1}{4}$ & $0$            & $0$            & $\sfrac{1}{4}$ & $0$            & $\sfrac{1}{4}$ & $\sfrac{1}{4}$ & $0$            \\
$X_1, Y_2, Y_3$ & $0$            & $\sfrac{1}{4}$ & $\sfrac{1}{4}$ & $0$            & $\sfrac{1}{4}$ & $0$            & $0$            & $\sfrac{1}{4}$ \\
$Y_1, X_2, Y_3$ & $0$            & $\sfrac{1}{4}$ & $\sfrac{1}{4}$ & $0$            & $\sfrac{1}{4}$ & $0$            & $0$            & $\sfrac{1}{4}$ \\
$Y_1, Y_2, X_3$ & $0$            & $\sfrac{1}{4}$ & $\sfrac{1}{4}$ & $0$            & $\sfrac{1}{4}$ & $0$            & $0$            & $\sfrac{1}{4}$ \\ \bottomrule
\end{tabularx}
\end{center}
That is, all parity-even bitstrings occur with equal probability in the top row (measurement context), whereas all parity-odd bitstrings occur with equal probability in all other measurement contexts. This empirical model has $\mathsf{CF}=1$, as illustrated in Fig.~\ref{fig:ncf-vs-exact} for deformation $\theta=0$.
\end{Example} 

With these ingredients in hand, we are ready to bound the performance of quantum-mechanical strategies in terms of the noncontextual fraction. A proof of this proposition can be found in the proofs of Theorems 3 and 4 in Ref.~\cite{Abramsky2017}. We present an abridged version of the proof here for completeness.

\begin{proposition}
    \label{prop:NCF-bound}
    Let $G = \mathsf{CSS}(H)$ be the CSS game defined by full-rank parity-check matrices $H = (H_X, H_Z)$.
    The success fractions of quantum-mechanical strategies for $G$ using an empirical model $e$ as a resource are bounded by
    \begin{equation}
        \omega^*(e) 
        \leq 
        1 - (1-\omega)\mathsf{NCF}(e) 
        =
        \omega + (1 - \omega) \mathsf{CF}(e)
        ,
        \label{eqn:NCF-bound}
    \end{equation}
    where $\mathsf{NCF}(e)$ is the noncontextual fraction of $e$ and $\omega \leq 1$ is the optimal classical success fraction of $G$.
\end{proposition}

\begin{proof}
    By decomposing $e$ into its contextual and noncontextual parts,
    \begin{equation}
        \omega^*(e) = \mathsf{NCF}(e) \omega^*(e^{\text{NC}}) + \mathsf{CF}(e) \omega^*(e^{\text{C}}) .
    \end{equation}
    Then, use $\omega^*(e^{\text{C}}) \leq 1 $ (i.e., assuming the contextual part of $e$ always computes the target function) and $\omega^*(e^{\text{NC}}) \leq \omega = 1-\nu_f / 2^d$ (by Theorem~\ref{thm:CSS-success-fraction}). This gives
    \begin{equation}
        \omega^*(e) \leq \mathsf{NCF}(e) \omega + [1 - \mathsf{NCF}(e)].
    \end{equation}
    Regrouping terms gives Eq.~\eqref{eqn:NCF-bound}.
\end{proof}

Note that if a perfect strategy that uses $e$ exists, so that $\omega^*(e)=1$, while $\omega < 1$, then the model has contextual fraction equal to one, $\mathsf{CF}(e)=1$. More generally, if the performance of the strategy that uses $e$ exceeds the classical success fraction, $\omega^*(e) > \omega$, then we can deduce that the model's contextual fraction is nonzero, $\mathsf{CF}(e) > 0$. This means that succeeding with probability $\omega^*(e) > \omega$ acts as a witness for contextuality of $e$. Finally, if the resource $e$ is entirely \emph{non}contextual, i.e., $\mathsf{NCF}(e)=1$, then it does not enable better performance than the optimal classical strategy and we have $\omega^*(e) \leq \omega$ for such noncontextual models $e$. The upper bound is compared with the performance of the Pauli strategy using a specific one-parameter family of states in Fig.~\ref{fig:ncf-vs-exact}. Note that the upper bound~\eqref{eqn:NCF-bound} is tight for the GHZ state for sufficiently small $\theta$. This is possible since the inequality $\omega^*(e^{\text{NC}}) \leq \omega$ is actually an equality for half of the possible global value assignments (recall that the target function is bent, so all classical strategies perform equally as poorly, up to a global sign). The discontinuous behavior arises from taking the maximum over noncontextual empirical models.


\section{Evaluating Walsh transforms with hypergraph states}
\label{sec:walsh-from-graphs}

We have shown that the optimal classical win probability $\omega$ is intimately connected to the contextuality of quantum states. Both since it can be interpreted as the maximum fraction of measurement contexts whose sign can be correctly described by a global value assignment, and since succeeding with probability greater than $\omega$ signals that the shared resource state must have been contextual. Here, we are concerned with introducing tools that make the evaluation of $\omega$ straightforward. We establish general relationships between properties of Boolean functions, namely their (generalized) Walsh spectrum, and (hyper)graph states. For a quadratic target function, we prove the following theorems:

\begin{theorem}
    \label{thm:walsh-from-symmetries}
    Let $f(\vec{x}) = \bigoplus_{i < j} A_{ij} x_ix_j$ be a homogeneous degree-2 Boolean function in $n$ variables. Let $\ket{C_G}$ be the graph state associated to the adjacency matrix $A \in M_n(\F{2})$. If the group $\mathcal{S}_X$ of $X$ symmetries of $\ket{C_G}$ has dimension $\dim_{\F{2}}(\mathcal{S}_X) = n_X$, the function $f$ has exactly $2^{n-n_X}$ nonzero Walsh coefficients of magnitude $\abs*{W_f}^2=2^{n+n_X}$.
\end{theorem}

\begin{theorem}
    \label{thm:bell-pairs}
    Let $\ket{C_G}$ be the graph state from Theorem~\ref{thm:walsh-from-symmetries}. The state $\ket{C_G}$ can be transformed into a product of $n_x$ qubits in the $\ket{+}$ state and $(n-n_x)/2$ maximally entangled pairs using a circuit of $\CX$ gates. 
\end{theorem}

Theorem~\ref{thm:bell-pairs} is proven by constructing the $\CX$ circuit using an explicit polynomial-time algorithm. It provides a recipe to identify the maximally contextual subsets of operators for generic CSS codes, which in turn allows us to determine the game that minimizes the classical win probability (within the class of CSS XOR games). This is how will be able to``extract'' the maximum contextuality out of a given CSS state. The graph-state mapping is also useful for converting the evaluation of generalized Walsh coefficients into classical partition functions, which is a result that we use extensively in Sec.~\ref{sec:examples}. 


\subsection{Graph states and quadratic Boolean functions}

A graph $G = (V, E)$ is a pair of two sets, $V$ (vertices) and $E \subseteq [V]^2$ (edges), with $[V]^2$ the set of all unordered pairs of vertices. The number of vertices, $\abs{V}$, is the \emph{order} of $G$. The adjacency matrix $A$ associated to $G$ is a $\abs{V} \times \abs{V}$ matrix with $A_{ij} = 1$ if the pair $\{ i, j \} \in E$, and $A_{ij} = 0$ otherwise. Note that we do not assume that the graph is connected, i.e., $G$ may have multiple disconnected components. Given $G$, we can define the associated \emph{graph state}:
\begin{equation}
    \ket{C_G} \coloneq \prod_{\{ i, j \} \in E(G) } \CZ_{ij} \ket{+}^{\otimes \abs{V}}
    ,
    \label{eqn:general-graph-state}
\end{equation}
where $\CZ_{ij}$ is the controlled-$Z$ gate between $i$ and $j$ with the explicit representation 
$\CZ_{ij} = (I + Z_i + Z_j - Z_iZ_j)/2$, 
and $\ket{+}$ is the $+1$ eigenstate of Pauli $X$. The state $\ket{C_G}$ is then stabilized by the operators
\begin{equation}
    S_i = X_i \prod_{j \in N(i)} Z_j
    ,
    \label{eqn:graph-state-stabilizers}
\end{equation}
where $N(i) \subset V$ is the \emph{neighborhood} of vertex $i$, defined by $N(i) = \{ j \in V \setminus \{ i \} : \{ i, j \} \in E \}$. The state $\ket{C_G}$ has amplitudes $\pm1$ in the $Z$ basis. Any such state $\ket{\phi}$ on $n$ qubits can be written as
\begin{equation}
    \ket{\phi} = \frac{1}{2^{n/2}} \sum_{\vec{z} \in \F{2}^n} (-1)^{f_\phi(\vec{z})} \ket{\vec{z}}
\end{equation}
for some Boolean function $f_\phi : \F{2}^n \to \F{2}$. To see that the state $\ket{C_G}$ is of this form, write the state $\ket{C_G}$ in the $Z$ basis and make use of the fact that, on a $Z$-basis eigenstate $\ket{\vec{z}}$, the gate $\CZ_{ij}$ acts as $\CZ_{ij}\ket{\vec{z}} = (-1)^{z_iz_j}\ket{\vec{z}}$:
\begin{equation}
    \ket{C_G} \propto
    \sum_{\vec{z} \in \F{2}^{\abs{V}}} \prod_{\{ i, j \} \in E } (-1)^{z_i z_j}\ket{\vec{z}} 
    = 
    \sum_{\vec{z} \in \F{2}^{\abs{V}}} (-1)^{\, f_G(\vec{z})} \ket{\vec{z}}.
\end{equation}
It follows that the state $\ket{C_G}$ encodes a homogeneous degree-2 Boolean function $f_G$, which can be expressed in ANF in terms of the graph $G$'s adjacency matrix by
\begin{equation}
    f_G(\vec{z}) = \bigoplus_{i < j} A(G)_{ij} \, z_i z_j
    .
    \label{eqn:graph-Boolean-function}
\end{equation}
By Eq.~\eqref{eqn:graph-Boolean-function} a graph $G$ uniquely defines a degree-2 Boolean function, which we call $f_G$. Conversely, a quadratic Boolean function uniquely defines a graph, and hence a graph state, in the following sense. Given a generic quadratic Boolean function $g(\vec{z})$, which may include terms of degree less than two, we can work with the \emph{polar bilinear form} derived from $g$, i.e., $B(\vec{z}, \vec{y}) = g(\vec{z}\oplus\vec{y}) \oplus g(\vec{z}) \oplus g(\vec{y})$, which removes the linear part of $g$ and symmetrizes the associated matrix: Given $g(\vec{z}) = \bigoplus_{i \leq j} g_{ij} z_i z_j$, the polar bilinear form gives rise to the matrix $B$ with entries $B_{ij} = g_{ij} \oplus g_{ji}$, which can be interpreted as an adjacency matrix.
Note that if $h(\vec{z})$ is an affine Boolean function, then $f_G(\vec{z})$ and $f_G(\vec{x})\oplus h(\vec{z})$ have the same Walsh spectrum, up to a sign and a relabeling of Walsh coefficients.
In the following we therefore focus on quadratic Boolean functions with no degree-1 terms. 

Now consider the Walsh spectrum of $f_G$ in~\eqref{eqn:graph-Boolean-function}. Since $(-1)^{zy} = \sqrt{2} \braket{y | H | z}$ with $H$ the Hadamard gate, we have the following Lemma:

\begin{lemma}
    \label{lem:walsh-from-cluster}
    Let $f_G(\vec{z}) = \bigoplus_{i < j} A(G)_{ij} z_i z_j$ be the degree-2 Boolean function in $n$ variables associated to the graph $G$, and let $\ket{C_G}$ be the corresponding $n$-qubit graph state. The Walsh spectrum $\{ W_{f_G}(\vec{y}) : \vec{y} \in F_2^n \}$ of $f_G$ is encoded in the $X$-basis amplitudes of the state $\ket{C_G}$, i.e.,
    \begin{equation}
        \ket{C_G} = 2^{-n} \sum_{\vec{y}\in\F{2}^n} W_{f_G}(\vec{y}) H^{\otimes n} \ket{\vec{y}},
    \end{equation}
    where $\ket{\vec{y}}$ is a $Z$-basis state and $H$ is the Hadamard gate with matrix elements $\sqrt{2}H_{ij} = (-1)^{ij}$ in the $Z$ basis.
\end{lemma}

\begin{proof}
    The state $H^{\otimes n}\ket{\vec{y}}$ is a normalized $X$-basis product state satisfying $X_i H^{\otimes n}\ket{\vec{y}} = (-1)^{y_i} H^{\otimes n}\ket{\vec{y}}$. Using a complete basis of $X$-basis states to insert a resolution of identity, we have $\ket{C_G} = \sum_{\vec{y} \in \F{2}^n} \braket{\vec{y} | H^{\otimes n} | C_G} H^{\otimes n} \ket{\vec{y}}$. We therefore need to relate the matrix element $\braket{\vec{y} | H^{\otimes n} | C_G}$ to the Walsh spectrum.
    Using Def.~\ref{def:Walsh-transform},
    \begin{equation*}
        W_{f_G}(\vec{y}) = \sum_{\vec{z} \in \F{2}^n} (-1)^{\langle \vec{z}, \vec{y} \rangle \oplus f_G(\vec{z})} = 2^n \sum_{\vec{z} \in \F{2}^n} \braket{\vec{y} | H^{\otimes n} | \vec{z}} \braket{\vec{z} | C_G}.
    \end{equation*}
    Summing over $\vec{z} \in \F{2}^n$ gives $W_{f_G}(\vec{y}) = 2^n \braket{\vec{y} | H^{\otimes n} | C_G}$.  
\end{proof}

This allows us to compute the Walsh transform by studying the symmetry properties of the associated graph state $\ket{C_G}$. Let us denote the stabilizer group of $\ket{C_G}$ by $\mathcal{S}$ and a generic element $S_{\vec{a}\vec{b}} \in \mathcal{S}$ by $S_{\vec{a}\vec{b}} = (-1)^{\tau_{\vec{a}\vec{b}}} X^{\vec{a}} Z^{\vec{b}}$, where, e.g., $X^\vec{a}$ is shorthand notation for the product $\prod_{i=1}^{n} X_i^{a_i}$ and $\tau_{\vec{a}\vec{b}} \in \F{2}$ governs the sign obtained by pulling all $X$'s to the left of $Z$'s in products of stabilizer generators~\eqref{eqn:graph-state-stabilizers}. That is, the vectors $\vec{a}$, $\vec{b}$ respectively determine the support of Pauli $X$ and $Z$ operators in each element of the stabilizer group. We denote the order of the graph $G$ by $n$, which implies $\dim_{\F{2}}(\mathcal{S}) = n$. Applying Lemma~\ref{lem:walsh-from-cluster}, we can write Walsh coefficients as $W_{f_G}(\vec{y}) = 2^n \braket{ +_n | Z^{\vec{y}} | C_G }$. Using $S_{\vec{a}\vec{b}}\ket{C_G} = \ket{C_G}$, and anticommutation of Pauli operators, $Z^{\vec{y}} X^{\vec{a}} = (-1)^{\inner{\vec{a}}{\vec{y}}}X^{\vec{a}}Z^{\vec{y}}$, we deduce that the Walsh coefficients satisfy 
\begin{subequations}
\label{eqn:walsh-relations}
\begin{align}
     W_{f_G}(\vec{y}) &= 2^n (-1)^{\tau_{\vec{a}\vec{b}} \oplus \langle \vec{a}, \vec{y} \rangle} \braket{+_n | Z^{\vec{y}\oplus\vec{b}}  | C_G} \\
     &= (-1)^{\tau_{\vec{a}\vec{b}} \oplus \langle \vec{a}, \vec{y} \rangle}  W_{f_G}(\vec{y}\oplus \vec{b}),
\end{align}%
\end{subequations}
with $\ket{+_n} \coloneq \ket{+}^{\otimes n}$. Let $\mathcal{S}_X < \mathcal{S}$ be the proper subgroup of $X$ symmetries of $\ket{C_G}$. That is, we define $\mathcal{S}_X = \{ S_{\vec{a}\vec{b}} \in \mathcal{S} : \vec{b} = \vec{0} \}$, the subset of stabilizers with no $Z$ operators. Applying Eq.~\eqref{eqn:walsh-relations} for $S_{\vec{a}\vec{b}} \in \mathcal{S}_X$, we find
\begin{equation}
    W_{f_G}(\vec{y}) = (-1)^{\tau_{\vec{a}} \oplus \langle \vec{a}, \vec{y} \rangle}  W_{f_G}(\vec{y}) \quad \forall\: \vec{a} \text{~s.t.~} S_{\vec{a}\vec{b}} \in \mathcal{S}_X,
    \label{eqn:Walsh-self-relation}
\end{equation}
with $\tau_\vec{a} \coloneq \tau_{\vec{a}\vec{0}}$. This result implies that $\tau_{\vec{a}} \oplus \langle \vec{a}, \vec{y} \rangle = 0$ and/or $ W_{f_G}(\vec{y}) = 0$. On the other hand, for a nontrivial element of the quotient group $g_\vec{b} \in \mathcal{S}/\mathcal{S}_X$, we have 
\begin{equation}
    \abs{ W_{f_G}(\vec{y}) } = \abs{  W_{f_G}(\vec{y} \oplus \vec{b}) }.
    \label{eqn:Walsh-magnitude-relation}
\end{equation}
Denoting $\dim_{\F{2}}(\mathcal{S}_X) \eqcolon n_X$, by the rank-nullity theorem, the dimension of the space of solutions $\vec{y}$ to the equations $\tau_{\vec{a}} \oplus \langle \vec{a}, \vec{y} \rangle = 0$ for all $\vec{a}$ equals $n - n_X$.
Similarly, $\dim_{\F{2}}(\mathcal{S}/\mathcal{S}_X) = n-n_X$ since $\mathcal{S}$ is isomorphic to $\Z{2}^n$. Hence, defining an equivalence relation $\vec{y}_1 \sim \vec{y}_2$ if there exists a $\vec{b}$ such that $\vec{y}_1 = \vec{y}_2 \oplus \vec{b}$ and $S_{\vec{a}\vec{b}} \in \mathcal{S}$, there exists exactly one equivalence class associated to nonzero Walsh coefficients, all of which have the same magnitude. It follows that the $X$ symmetries of $\ket{C_G}$ are sufficient to determine which Walsh coefficients are nonzero, and their magnitude.

\begin{proof}[Proof of Theorem~\ref{thm:walsh-from-symmetries}]
    From Eq.~\eqref{eqn:Walsh-magnitude-relation}, given a single nonzero Walsh coefficient, we can generate $2^{n-n_X}$ distinct nonzero coefficients of equal magnitude. However, from Eq.~\eqref{eqn:Walsh-self-relation} there are at most $2^{n-n_X}$ nonzero Walsh coefficients. We have therefore constructed all nonzero Walsh coefficients. Parseval's theorem states that
    \begin{equation}
        \sum_{\vec{y}\in\F{2}^n} W_f(\vec{y})^2 = 2^{2n},
        \label{eqn:parseval}
    \end{equation}
    for any $f$,
    from which it follows that the magnitude $\abs{W}$ of the nonzero Walsh coefficients satisfies
    $
        2^{n-n_X} \abs{W}^2 = 2^{2n}
    $,
    and thus $\abs{W}^2 = 2^{n+n_X}$.
\end{proof}

\begin{corollary}
    If the cluster state $\ket{C_G}$ associated to a graph $G$ of order $n$ (with $n$ even) has no $X$ symmetries, the function $f_G$ is bent.
\end{corollary}

\begin{proof}
    If $n_X = 0$, $\abs{W(\vec{y})}^2 = 2^n$ for all $\vec{y}$, i.e., the Walsh spectrum is flat, which is a defining property of bent functions (Def.~\ref{def:bent}).
\end{proof}

\begin{figure}[t]
    \centering
    \includegraphics[width=\linewidth]{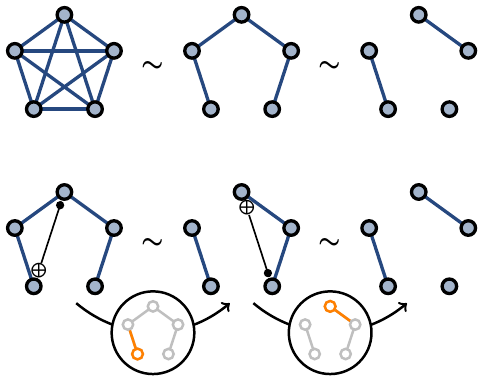}
    \caption{Illustration of graph equivalence under linear transformations. \textbf{Top row}: The left graph, the complete graph $K_{5}$, represents the Boolean function~\eqref{eqn:GHZ-all-to-all} computed in the GHZ game for a specific choice of stabilizer generators with six qubits. There exists a linear transformation that transforms $K_5$ to a linear graph [corresponding to Eq.~\eqref{eqn:ghz-function}]. Finally, there exists a sequence of linear transformations that relates the linear graph to a ladder rung graph with two rungs. \textbf{Bottom row}: Explicit elementary operations that transform the linear graph (left) to the ladder rung graph (right). The source and target qubits that are acted upon in accordance with Eq.~\eqref{eqn:CX-transform} are illustrated by the $\CX_{ij}$ gates. The neighborhood $N(j)$ of the target $j$ that enters Eq.~\eqref{eqn:graph-transform} is illustrated in orange. In each step, the connections between the control, $i$, and the neighbors ($\neq i$) of the target, $j$, are toggled.}
    \label{fig:standard-form}
\end{figure}

Finally, to prove Theorem~\ref{thm:bell-pairs}, we consider how affine equivalence of Boolean functions (Def.~\ref{def:ext-affine-equivalence}) manifests at the level of the associated graph and the corresponding graph state. Recall that two functions $f$ and $g$ are affine equivalent if $g(\vec{z}) = f(\vec{z} T \oplus\vec{b})$ for nonsingular $T$.

\begin{Def}[Standard form]
    \label{def:standard-form}
    A quadratic Boolean function $f(\vec{x})$ in $n$ variables is written in \emph{standard form} if $f(\vec{x}) = x_1x_2 \oplus x_3x_4 \oplus \dots \oplus x_{2k-1}x_{2k} \oplus h(\vec{x})$ for some $2k \leq n$, where $h(\vec{x})$ is a linear Boolean function.
\end{Def}

\begin{theorem}[Dickenson's theorem]
    \label{thm:dickenson}
    There exists a sequence of linear coordinate transformations that brings any quadratic Boolean function into standard form.
\end{theorem}

\begin{proof}
    See, e.g., Ref.~\cite[Appendix A]{ovsienko_real_2016}. Viewed from the perspective of graph states, it also follows from Ref.~\cite[Lemma 18]{bataille2021reduced}. The linear coordinate transformation that rewrites any quadratic Boolean function in standard form is constructed explicitly in Algorithm~\ref{alg:standard-form}. See also Ref.~\cite[Theorem 18]{TokarevaBent2015chap2}.
\end{proof}

A matrix $T \in M_n(\F{2})$ is nonsingular iff it can be connected to the identity matrix via elementary column operations. Consider the matrix $T_{j \to i \oplus j}$ (with $i \neq j$) such that $M T_{j \to i \oplus j}$ sends $M_{kl} \to M_{ki} \oplus M_{kl}$ if $l=j$ and $M_{kl} \to M_{kl}$ otherwise, i.e., adding column $i$ to column $j$ of an arbitrary matrix $M \in M_n(\F{2})$. Hence, when this matrix acts on the row vector $\vec{z}$, we obtain
\begin{equation}
    \vec{z} T_{j \to i\oplus j} = (z_1, \dots, z_i , \dots, z_j \oplus z_i, \dots, z_n).
    \label{eqn:column-operation}
\end{equation}
It follows that an elementary column operation transforms the Boolean function $f_G$ as
\begin{equation}
    f_G(\vec{z} T_{j \to i \oplus j}) = A_{ij} z_i \oplus f_G(\vec{z}) \oplus \bigoplus_{\substack{ k \in N(j) \\ k \neq i }} z_k z_i
    .
    \label{eqn:column-complementation}
\end{equation}
The final term in Eq.~\eqref{eqn:column-complementation} toggles connections between the edges $\{ i, k \}$ for $k \neq i$ in the neighborhood of $j$. This modified function is encoded by the action of $\CX$ on the graph state. Consider
\begin{equation}
    \CX_{ij} \ket{C_G} \propto \sum_{\vec{z} \in \F{2}^n} (-1)^{f_G(\vec{z})} \CX_{ij} \ket{\vec{z}},
\end{equation}
where $\CX_{ij}$ denotes the controlled-$X$ gate with $i$ the control and $j$ the target, and note that $\CX_{ij}\ket{\vec{z}} = \ket{\vec{z}T_{j\to i \oplus j}}$, with the action of $T_{j \to i \oplus j}$ given explicitly in Eq.~\eqref{eqn:column-operation}. Note that the matrix $T_{j\to i \oplus j}$ is self-inverse. Therefore, defining $\vec{y} = \vec{z} T_{j\to i\oplus j}$, we arrive at 
\begin{equation}
    \CX_{ij} \ket{C_G} \propto \sum_{\vec{y} \in \F{2}^n} (-1)^{f_G(\vec{y} T_{j \to i \oplus j})} \ket{\vec{y}}.
    \label{eqn:CX-transform}
\end{equation}
To remove the degree-1 term in Eq.~\eqref{eqn:column-complementation}, we can consider the state
\begin{equation}
    Z_i^{\mathbf{1}(i \in N(j))} \CX_{ij} \ket{C_G},
\end{equation}
which removes the negative sign on the stabilizer generator associated to vertex $i$ if $i$ and $j$ are neighbors, which is the graph state associated to the Boolean function
\begin{equation}
    f_G(\vec{z}) \oplus \bigoplus_{\substack{ k \in N(j) \\ k \neq i }} z_k z_i.
    \label{eqn:graph-transform}
\end{equation}
Thus, if the function $f_G$ contains no degree-1 terms, neither does Eq.~\eqref{eqn:graph-transform}.
Graphically, Eq.~\eqref{eqn:graph-transform} is interpreted as ``toggling'' (i.e., $1 \leftrightarrow 0$) the connections between the control, $i$, and the neighbors ($\neq i$) of the target, $j$. The above discussion is summarized by the commutative diagram
\[
\begin{tikzcd}[row sep=large, column sep=1cm]
f(\vec{x}) = \bigoplus_{i \leq j} A_{ij} x_ix_j
  \arrow[r, leftrightarrow, "{T_{j \to i \oplus j} }"{above}]
  \arrow[d, "{}"']
&
\tilde{f}(\vec{x}) = f(\vec{x} T)
  \arrow[d, "{}"]
\\
\ket{C_{G(f)}}
  \arrow[r, leftrightarrow, "{Z_i^{\mathbf{1}({i \in N(j))}} \CX_{ij}}"{above}]
&
\ket{C_{G(\tilde{f})}}
\end{tikzcd}
\]
where $G(f)$ is the graph obtained from interpreting the degree-2 part of $f$ as an adjacency matrix [see the discussion below Eq.~\eqref{eqn:graph-Boolean-function}].

\begin{proof}[Proof of Theorem~\ref{thm:bell-pairs}]
    By \hyperref[thm:dickenson]{Dickenson's theorem}
    there exists a matrix $T \in \GL{n}{\F{2}}$ such that $f(\vec{x} T)$ is in standard form.
    If the function $f$ has Walsh coefficients of magnitude $\abs{W} = 2^{n-k}$, then $f(\vec{x} T) = x_1x_2 \oplus \dots \oplus x_{2k-1} x_{2k} \oplus h(\vec{x})$ since linear transformations preserve Walsh spectra and the function $x_1x_2 \oplus \dots \oplus x_{2m-1} x_{2m}$ with $2m \leq n$ has Walsh coefficients $\in \{ 0, \pm 2^{n-m} \}$. The graph state corresponding to $f(\vec{x} T)$ is local-Clifford equivalent to $k$ maximally entangled pairs because $\CZ \ket{+}^{\otimes 2} = (I \otimes H) \ket{\Phi}$ with $\ket{\Phi}$ the 2-qubit Bell state. Since the matrix $T$ is invertible, it can always be decomposed into a product of elementary column operations, and each such column operation corresponds to the action of a $\CX$ gate. Hence, a sequence of $\CX$ gates transforms $\ket{C_G}$ into a state that is local-Clifford equivalent to $k$ Bell pairs and $n-2k$ qubits in the $\ket{+}$ state. 
\end{proof}

\begin{figure}
    \centering
    \includegraphics[width=\linewidth]{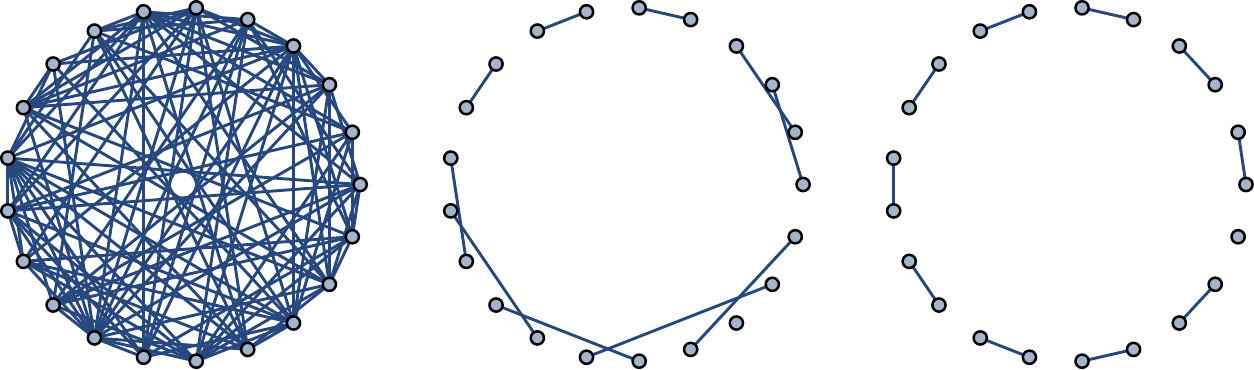}
    \caption{An illustration of Algorithm~\ref{alg:standard-form} for a generic, randomly generated adjacency matrix, constructed by taking a matrix with entries $g_{ij}$ that are iid using the uniform distribution over $\{ 0, 1 \}$, then writing $B_{ij} = g_{ij} \oplus g_{ji}$. \textbf{Left}: The initial graph, in which each vertex is connected to many other vertices. \textbf{Middle}: An intermediate reduced graph in which every vertex is connected to at most one other vertex. \textbf{Right}: The graph in standard form (Def.~\ref{def:standard-form}).}
    \label{fig:generic-reduction}
\end{figure}

\begin{Example}
    \label{ex:ghz-spectrum}
    Consider the GHZ game from Example~\ref{ex:ghz-game} with an \emph{even} number of qubits. Thus, $n=N-1$ is odd. Consider the following stabilizer generators:
    \begin{equation}
        \prod_{i=1}^{N} X_i \, ,
        \quad 
        \{Z_{i} Z_{N} : i \in \mathbb{N}, 1 \leq i < N \}
        \, ,
    \end{equation}
    and respectively associate the bits $x$ and $z_i$ to the stabilizer generators. Fixing $x=1$, the target Boolean function~\eqref{eqn:CSS-game-victory} involves all-to-all coupling
    \begin{equation}
        \label{eqn:GHZ-all-to-all}
        \tilde{f}_{\text{GHZ}}(\vec{z}) = 
        \bigoplus_{\{ i, j \}} z_i z_j \oplus \bigoplus_{i=1}^n z_i
        ,
    \end{equation}
    where the first sum runs over the set of all unordered pairs $\{i, j \}$. This function is affine equivalent to Eq.~\eqref{eqn:ghz-function}. The explicit transformation that implements the equivalence is $\tilde{f}_\text{GHZ}(\vec{z}) = f_\text{GHZ}(\vec{z}T)$ where $f_\text{GHZ}$ is the function in Eq.~\eqref{eqn:ghz-function} with $x=1$ and the matrix $T \in \GL{n}{\F{2}}$ is equal to the upper-triangular matrix $T_{ij} = 1$ if $j \geq i$ and $T_{ij} = 0$ otherwise. The sequence of linear transformations that brings the function $f_{\text{GHZ}}$ into standard form for $n=5$ is shown in the bottom row of Fig.~\ref{fig:standard-form}.
    
    We now examine the Walsh spectrum from the graph-state perspective. The graph $G$ associated to $f_\text{GHZ}$ is the linear graph, such that the auxiliary quantum state $\ket{C_G}$ is the 1D cluster state with open boundary conditions. Removing the degree-1 terms from $f_\text{GHZ}$, the stabilizer generators are $S_i = Z_{i-1}X_iZ_{i+1}$ for $2 \leq i \leq n-1$, as well as $S_1 = X_1Z_2$ and $S_n = Z_{n-1}X_n$ at the boundaries. Since $n$ is assumed odd, both $Z_2$ and $Z_{n-1}$ belong to the \emph{even} sublattice [$Z_i$ belongs to the even (odd) sublattice if $i$ is even (odd)]. Hence, there exists an $X$ symmetry of the form $X_\text{odd} = \prod_{i \text{ odd }}X_{i}$. The subgroup of $X$ symmetries is $\mathcal{S}_X \cong \Z{2}$ and there are thus $2^{n-1}$ nonzero Walsh coefficients, consistent with the standard form in Fig.~\ref{fig:standard-form}.
\end{Example}

In general, the sequence of transformations that brings a generic quadratic Boolean function into standard form can be found using an algorithm analogous to Gaussian elimination over $\F{2}$ applied to the adjacency matrix (defined by the polar bilinear form of the function). An explicit algorithm that performs this is given in Ref.~\cite[Fig.~6]{bataille2021reduced}. For completeness, we also present a slightly modified version of this algorithm in Algorithm~\ref{alg:standard-form}. We remark that the circuit of $\CX$ gates that brings a graph state into standard form necessarily has the effect of mapping the group of $X$-symmetries $\mathcal{S}_X$ to the stabilizer group of $\ket{+}^{\otimes (n-2k)}$. 

Using this algorithm, we can always find a subspace on which the function acts as a bent function. Physically, this means that, having fixed say the $X$ inputs to some specific bitstring (such that the computed function is quadratic), we can naturally find a subspace of $Z$ inputs leading to maximal Bell-inequality violations. We illustrate the algorithmic procedure for a random symmetric matrix over $\F{2}$ in Fig.~\ref{fig:generic-reduction}.


\subsection{Hypergraph states and higher-order Boolean functions}
\label{sec:hypergraph-states}

Any Boolean function can be written in ANF:
\begin{equation}
    f(\vec{z}) = \bigoplus_{I \subseteq \{ 1, \dots, n \} } a_I x^I
    .
    \label{eqn:ANF-to-hypergraph}
\end{equation}
The coefficients $a_I$ entering the ANF of $f$ can be taken to define an associated hypergraph. A $k$-uniform hypergraph $G_k = (V, E)$ is pair of two sets, $V$ (vertices) and $E \subset [V]^k$. Elements of $E$ are hyperedges that connect exactly $k$ pairwise-distinct vertices. We also use the notation $G_{\leq k}$ to denote a general hypergraph that contains hyperedges of any order $\leq k$. In this case, $G_{\leq k} = (V, E)$, where $E \subset \bigcup_{i=1}^k [V]^i$. In Eq.~\eqref{eqn:ANF-to-hypergraph}, the nonzero coefficients $a_I$ determine the hyperedges that are included in the hypergraph, i.e., $E = \{ I : a_I \neq 0 \}$, and $k$ is determined by the algebraic degree of $f$.

To each hypergraph $G_{\leq k}$, we associate a quantum state
\begin{equation}
    \ket{G_{\leq k}} \coloneq \prod_{m=1}^k \prod_{ \{i_1, \dots, i_m\} \in E } \mathsf{C}^{m}\mathsf{Z}_{ \{ i_1, \dots, i_m\} } \ket{+}^{\otimes \abs{V}}
    ,
    \label{eqn:hypergraph-state}
\end{equation}
where $\mathsf{C}^{m}\mathsf{Z}_{ \{ i_1, \dots, i_m\} }$ is diagonal in the $Z$ basis. Given a hyperedge $e \in E$, the action of this operator on a $Z$-basis state $\ket{\vec{z}}$ is 
\begin{equation}
    \mathsf{C}^{\abs{e}}\mathsf{Z}_{ e } \ket{\vec{z}} = (-1)^{\prod_{i \in e} z_i} \ket{\vec{z}}.
\end{equation}
That is, it acts as $-1$ only when all $z_i=1$ for all $i \in e$. The operator $\mathsf{C}^{1}\mathsf{Z}_i = Z_i$, while $\mathsf{C}^{2}\mathsf{Z}_{\{i,j\}} = \mathsf{CZ}_{ij}$. We also define $\mathsf{C}^{0}\mathsf{Z}_i = (-1)^i$ for convenience. Hence, Eq.~\eqref{eqn:hypergraph-state} encodes the function~\eqref{eqn:ANF-to-hypergraph} in its $Z$-basis amplitudes. The state~\eqref{eqn:hypergraph-state} is stabilized by operators
\begin{equation}
    S_i = X_i \prod_{I \in N(i)} \mathsf{C}^{\abs{I}}\mathsf{Z}_I
    \label{eqn:hypergraph-stabilizers}
\end{equation}
where the neighborhood of a vertex $i$ is defined as $N(i) = \{ I \subseteq V \setminus \{ i \} : \{i\} \cup I \in E \}$. If we are considering a hypergraph $G_{\leq k}$, then the product runs over $I$ satisfying $\abs{I} < k$. If $\{ i \} \in E$ then Eq.~\eqref{eqn:hypergraph-stabilizers} includes $I = \emptyset$, giving the sign $\mathsf{C}^{0}\mathsf{Z}_i = (-1)^i$. Note that the stabilizers in Eq.~\eqref{eqn:hypergraph-stabilizers}, $S_i \ket{G_{\leq k}} = \ket{G_{\leq k}}$, are in general no longer elements of the Pauli group.

Now we consider the case relevant to CSS games where $f : \F{2}^d \to \F{2}$ has algebraic degree three. Let $\mathscr{C}$ be the space of deterministic classical strategies of the form~\eqref{eqn:unrestricted-classical-strat}. By definition, denoting $\vec{w} = (\vec{x}, \vec{z})$, the generalized Walsh coefficient $W_f[c]$ is
\begin{equation}
    \sum_{\vec{w}} (-1)^{c(\vec{w}) \oplus f(\vec{w})} 
    = \sum_{\vec{w}} \braket{ \vec{w} | U_c (-1)^{f(\vec{w})} | \vec{w} }
    ,
\end{equation}
where we introduced $U_c$, a diagonal unitary operator with matrix elements $\braket{ \vec{w} | U_c | \vec{w} } = (-1)^{c(\vec{w})}$. Introducing the two hypergraphs $G_{\leq 3}[f]$ (associated to the target function $f$) and $G_{\leq 2}[c]$ (associated to the deterministic classical strategy $c$), we have
\begin{equation}
    2^{-d} W_f[c] = \braket{ +_d | U_c | G_{\leq 3}[f] } = \braket{ G_{\leq 2}[c] | G_{\leq 3}[f] }.
\end{equation}
The generalized Walsh coefficients for CSS games therefore reduce to the overlap between hypergraph states of orders two and three. The regular Walsh spectrum of $f$ is encoded in the overlaps $\braket{ G_{1}(\vec{y}) | G_{\leq 3}[f] }$, where $\ket{G_1(\vec{y})}$ is a 1-uniform hypergraph state, equivalent to $Z^\vec{y} \ket{+_d}$, i.e., a generic state in the $X$ basis, as in Lemma~\ref{lem:walsh-from-cluster}.
 

\section{Bounding classical strategies for prototypical quantum phases}
\label{sec:examples}

In this section we illustrate the power of the theoretical tools introduced in the previous two sections by application to several prototypical quantum states. 


\subsection{Warm up: GHZ state}

First, as a warm-up calculation, we compute the optimal classical win probability $\omega$ for the GHZ game. While the optimal probability in this case is well-known~\cite{Brassard2003multiparty,broadbent2004thesis,brassard2005recasting}, our method for deriving $\omega$ provides new insights and will be useful for later calculations.

Consider the GHZ game from Example~\ref{ex:ghz-game} with $N$ players and introduce the variable $n = N-1$. The target function~\eqref{eqn:ghz-function} with $x=1$ is expressed in ANF as
\begin{equation}
    f(\vec{z}) = \bigoplus_{i=1}^{n-1} z_i z_{i+1} \oplus \bigoplus_{i=1}^n z_i .
    \label{eqn:GHZ-function-fixed}
\end{equation}
The graph $G$ associated to this function is the line graph with $n$ vertices. Therefore, the corresponding auxiliary graph state that encodes Eq.~\eqref{eqn:GHZ-function-fixed} is the 1D cluster state with open boundaries with $n$ qubits. The case of odd $n$ was worked out in Example~\ref{ex:ghz-spectrum}. For even $n$, the cluster-state stabilizers are similarly $S_i = Z_{i-1}X_iZ_{i+1}$ for $2 \leq i \leq n-1$, as well as $S_1 = X_1Z_2$ and $S_n = Z_{n-1}X_n$ at the boundaries. However, since $n$ is even, $X_1$ and $X_n$ belong to opposite sublattices, such that there are no nontrivial $X$ symmetries. It follows from Theorem~\ref{thm:walsh-from-symmetries} that
\begin{equation}
    \max_{\vec{y} \in \F{2}^n} \abs{W_f(\vec{y})} = 2^{\lceil n / 2 \rceil} \implies N_f = 2^{n-1} - 2^{\lceil n / 2 \rceil - 1}.
\end{equation}
Thus, the target function~\eqref{eqn:GHZ-function-fixed} in the GHZ game is bent for $n$ even.
Note that, for even $n$, the target function~\eqref{eqn:GHZ-function-fixed} can be written in standard form $f(\vec{z}) = z_1z_2 \oplus z_3 z_4 \oplus \dots \oplus z_{n-1}z_n$ by working with the stabilizer generators $Z_{2k}Z_{2k+1}$ and $Z_1 Z_2 \cdots Z_{2k}$ for all $k\in \{1, \dots, n/2 \}$.
The optimal success fraction $\omega$ is therefore as small as possible:
\begin{equation}
    \omega = 1-2^{-n}N_f = 
    \frac{1}{2} \left( 1 + 2^{-\lfloor n / 2 \rfloor} \right).
    \label{eqn:GHZ-success-fraction}
\end{equation}
An interesting observation is that, since the Walsh spectrum is flat, $\abs*{W_f(\vec{y})}=\text{const.}$, all deterministic classical strategies perform equally poorly (up to an overall parity bit).
It is also useful to calculate the coefficients exactly via a transfer-matrix method -- \emph{à la} the 1D classical Ising model -- as presented in appendix~\ref{app:walsh-spectrum-ghz}, which corroborates the result~\eqref{eqn:GHZ-success-fraction}.

It is also straightforward to compute the win probability of the GHZ game with \emph{no} input restrictions. Half the time, the target function is trivial (corresponding to measuring commuting $Z$ operators only). Hence, the optimal success fraction increases to $\omega = (3 + 2^{-\lfloor n/2 \rfloor})/4$.


\subsection{1D cluster state}
\label{sec:1D-cluster-bounds}

Consider a 1D cluster state with periodic boundary conditions consisting of $N$ qubits with $N$ even. The associated CSS game uses the $N$ stabilizer generators
\begin{equation}
    Z_{2k-2} Z_{2k-1} Z_{2k} ,
    \:\:
    X_{2k-1}X_{2k}X_{2k+1} ,
    \quad 
    k = 1, \dots, N/2
    \label{eqn:cluster-state-stabilizers}
\end{equation}
associated to the odd and even sublattices, respectively. The bits associated to these stabilizer generators are collected into a single bitstring $\vec{w}$ by identifying $z_k = w_{2k-1}$ and $x_k = w_{2k}$. 
{The players on the even sublattice are handed $a_{2k} = w_{2k}$ and $b_{2k} = w_{2k-1} \oplus w_{2k+1}$, while the odd-sublattice players are handed $a_{2k+1} = w_{2k}\oplus w_{2k+2}$ and $b_{2k+1} = w_{2k+1}$.} 
Note that the stabilizers in Eq.~\eqref{eqn:cluster-state-stabilizers} are unitary equivalent to the canonical set of stabilizers $\{ Z_{i-1} X_i Z_{i+1} \}$ used to define the cluster state under the action of $U_H = \prod_{k} H_{2k-1}$ with $H$ the Hadamard matrix. As a function of $\vec{w}$, the target function is
\begin{equation}
    f_C(\vec{w}) = \bigoplus_{i=1}^N w_i w_{i+1} w_{i+2} \oplus \bigoplus_{i=1}^N w_i w_{i+1}
    ,
\end{equation}
where subscript indices belong to $\mathbb{Z}/N\mathbb{Z}$.
First, observe that fixing, e.g., $w_{2k}=1$ reduces to the GHZ function (with periodic boundary conditions)
\begin{equation}
    f_C(\vec{w}\rvert_{w_{2k}=1}) = \bigoplus_{k=1}^{N/2} w_{2k-1} w_{2k+1} \oplus \bigoplus_{k=1}^{N/2} w_{2k-1}
    .
\end{equation}
While the target function is now identical to Eq.~\eqref{eqn:GHZ-function-fixed}, the questions asked to the players are different. While this does not modify the performance of noncommunicating deterministic classical strategies, it has important consequences for ``strategies'' that correspond to classical circuits consisting of bounded fan-in gates~\cite{Daniel2022Exp}.

We now compute the optimal success fraction of the cluster-state game \emph{without} input restrictions. Deterministic classical strategies correspond to a specific family of quadratic Boolean functions
\begin{equation}
    c(\vec{w}) = u_0 \oplus \bigoplus_{i=1}^N u_i w_i \oplus \bigoplus_{i=1}^N v_i w_i w_{i+1}
\end{equation}
parametrized by $u_0\in \F{2}$, $\vec{u},\vec{v}\in \F{2}^N$ subject to the constraint $\bigoplus_i v_i = 0$. The corresponding generalized Walsh coefficient (Def.~\ref{def:generalized-Walsh}) we denote by $W_C(\vec{u}, \vec{v})$, which, using the results of Sec.~\ref{sec:walsh-from-graphs}, has an interpretation as the overlap of some auxiliary graph state with the state $\ket{+}^{\otimes N}$.
Explicitly, by definition 
\begin{equation}
    W_C(\vec{u}, \vec{v}) = \sum_{\vec{w} \in \F{2}^{N}} (-1)^{\bigoplus_{i=1}^N u_iw_i \oplus v_i w_i w_{i+1} \oplus w_{i} w_{i+1} w_{i+2} }.
    \label{eqn:cluster-state-walsh}
\end{equation}
Note that we drop the parity bit $u_0$ for convenience and that only $\vec{v}$ satisfying $\bigoplus_i v_i = 0$ correspond to implementable classical strategies. Translating Eq.~\eqref{eqn:cluster-state-walsh} into the quantum-state language of Sec.~\ref{sec:walsh-from-graphs},
\begin{equation}
    W_C(\vec{u}, \vec{v}) = 
    2^N 
    \langle + | \prod_i Z_i^{u_i} \CZ_{\{i,i+1\}}^{v_i} \CCZ_{\{i,i+1,i+2\}} | + \rangle.
    \label{eqn:cluster-state-walsh-quantum}
\end{equation}
We can either interpret this expression as the overlap of a hypergraph state $\ket{G_{\leq 3}}$ with the state $\ket{+}^{\otimes N}$ or as the overlap of a 3-uniform hypergraph state with a graph state parametrized by $(\vec{u}, \vec{v})$.
We now give Eq.~\eqref{eqn:cluster-state-walsh} a convenient graphical interpretation. Consider arranging the bits $w_i$ in two rows of size $N/2$ as
\begin{equation*}
    \begin{tikzpicture}[every path/.style={thick,draw=lightgray,opacity=0.75,line cap=round}]
        \foreach \x in {0, 1, ..., 3} {
            \draw ({\x}, 0) --++(0:1);
            \draw ({\x+1}, 0) --++(120:1);
            \draw ($({\x},0)+(60:1)$) --++(1, 0);
        }
        \foreach \x in {0, 1, ..., 4} {
            \draw ({\x}, 0) --++(60:1);
        }
        \node [below] at (2, 0) {$x_k$}; 
        \node [below] at (3, 0) {$x_{k+1}$};
        \node [below] at (1, 0) {$\cdots$};
        \node [below] at (4, 0) {$\cdots$};
        \node [above] at ($(1,0)+(60:1)$) {$z_k$}; 
        \node [above] at ($(2,0)+(60:1)$) {$z_{k+1}$}; 
        \node [above] at ($(0,0)+(60:1)$) {$\cdots$}; 
        \node [above] at ($(3,0)+(60:1)$) {$\cdots$}; 
    \end{tikzpicture}
\end{equation*}
With this spatial arrangement of bits, there exists a $\CCZ$ gate for every triangle (both ``up'' and ``down'' orientations), and a $\CZ$ gate for every edge connecting the upper and lower rows (if the corresponding $v_i$ is nonzero). The $Z$ gates act on vertices. Next, consider fixing the $\vec{x}$ bits on the bottom row. This allows us to write Eq.~\eqref{eqn:cluster-state-walsh-quantum} as
\begin{equation}
    W_C(\vec{u}, \vec{v}) = 2^n\sum_{\vec{x} \in \F{2}^{n}} \braket{ \vec{+}_n | U_{\vec{u},\vec{v}}(\vec{x})  | +_n},
    \label{eqn:generalized-Walsh-fixed}
\end{equation}
where $U_{\vec{u},\vec{v}}(\vec{x})$ is an $\vec{x}$-dependent unitary quantum circuit consisting of $Z$'s and $\CZ$'s and $n=N/2$. Fixing the $\vec{x}$ bits reduces the $\CCZ$ gates to ($Z^{ab} = Z\delta_{ab}$)
\begin{equation*}
    \begin{tikzpicture}[gray lines]
        \draw (0, 0) --++(0:1);
        \draw (1, 0) --++(120:1);
        \draw (0, 0) --++(60:1);
        \node [text=zpauli] at (60:1) {$Z^{ab}$};
        \node at (0, 0) {$a$};
        \node at (1, 0) {$b$};
    \end{tikzpicture}
    \qquad
    \begin{tikzpicture}[gray lines]
        \draw (0.5, 0) --++(120:1);
        \draw (0.5, 0) --++(60:1);
        \draw ($(0.5,0)+(120:1)$) --++(1, 0);
        \node at (0.5, 0) {$0$};
    \end{tikzpicture}
    \qquad
    \begin{tikzpicture}[gray lines]
        \draw (0.5, 0) --++(120:1);
        \draw (0.5, 0) --++(60:1);
        \draw ($(0.5,0)+(120:1)$) --++(1, 0);
        \node at (0.5, 0) {$1$};
        \coordinate (L) at ($(0.5,0)+(120:1)$);
        \coordinate (R) at ($(0.5,0)+(60:1)$);
        \coordinate (C) at ($(L)!0.5!(R)$);
        \draw [red, ultra thick, opacity=1] (L) to (R);
        \fill [red, draw=none, fill opacity=1] (L) circle [radius=2pt];
        \fill [red, draw=none, fill opacity=1] (R) circle [radius=2pt];
        \node [above, text=red, fill=none, font=\small] at (C) {$\textsf{CZ}$};
    \end{tikzpicture}
\end{equation*}
Each summand in Eq.~\eqref{eqn:generalized-Walsh-fixed} therefore reduces to the overlap of some graph state with the state $\ket{+}^{\otimes n}$. The locations of the $\CZ$ gates are determined entirely by $\vec{x}$, while the locations of the $Z$ gates depend also on $\vec{u}$, $\vec{v}$. For $\vec{u} = \vec{v} = \vec{0}$, the sum includes terms such as
\begin{equation*}
    \begin{tikzpicture}[gray lines]
        \foreach \x in {0, 1, ..., 3} {
            \draw ($({\x},0)+(60:1)$) --++(1, 0);
        }
        \foreach \x in {0, 1, ..., 4} {
            \draw ({\x}, 0) --++(60:1);
            \draw ({\x}, 0) --++(0:1);
            \draw ({\x+1}, 0) --++(120:1);
        }
        \node at (0, 0) {$0$};
        \node at (1, 0) {$1$};
        \node at (2, 0) {$0$};
        \node at (3, 0) {$1$};
        \node at (4, 0) {$1$};
        \node at (5, 0) {$0$};
        \coordinate (L1) at ($(0,0)+(60:1)$);
        \coordinate (R1) at ($(1,0)+(60:1)$);
        \coordinate (L2) at ($(2,0)+(60:1)$);
        \coordinate (R2) at ($(3,0)+(60:1)$);
        \coordinate (L3) at ($(3,0)+(60:1)$);
        \coordinate (R3) at ($(4,0)+(60:1)$);

        \foreach \n in {L1, L2, R1, R2, R3} {
            \fill [red, draw=none, fill opacity=1] (\n) circle [radius=2pt];
        }
        \foreach \L/\R in {L1/R1, L2/R2, L3/R3} {
        \draw[red,ultra thick,opacity=1]
            (\L) -- (\R)
            node[midway,above,text=red,fill=none,font=\small]{\textsf{CZ}};
        }
        \Zop{L3}
        \node [above,text=zpauli,fill=none,font=\small,inner sep=3pt] at (L3) {$Z$};
    \end{tikzpicture}
\end{equation*}
Changing $\vec{u}$ and $\vec{v}$ will modify the pattern of $Z$'s inserted on the top row.
Each summand can, in principle, be evaluated exactly using the methods presented in Sec.~\ref{sec:walsh-from-graphs}. We instead provide simple upper and lower bounds on the largest generalized Walsh coefficient~\eqref{eqn:cluster-state-walsh}.

\paragraph{Lower bound} A lower bound on $\omega$ is obtained by evaluating $W_C(\vec{u}, \vec{v})$ exactly for a specific choice of $\vec{u}$, $\vec{v}$. Consider the strategy in which $\vec{u} = \vec{v} = \vec{0}$. In this case, we are tasked with evaluating the sum
\begin{equation}
    W_C(\vec{0}, \vec{0}) = \sum_{\vec{w} \in \F{2}^{N}} (-1)^{\bigoplus_{i=1}^N w_{i} w_{i+1} w_{i+2} }.
\end{equation}
This expression can be calculated straightforwardly by representing the sum as a tensor network:
\begin{equation*}
    \begin{tikzpicture}[baseline={(0, 0.5)}, x=2.4ex, y=2.4ex]
        \clip (-1.5, -2.5) rectangle ({9.5+4}, 5);
        
        \foreach \x in {-4, 0, 4, 8, 12} {
            \begin{scope}[shift={(\x, 0)}]
                \draw[thick, domain=0:4, smooth, variable=\x, black] plot ({\x}, {1.5*(1-cos(360*\x/8))});
                \draw[line width=4, domain=0:4, smooth, variable=\x, white] plot ({\x}, {1.5*(1+cos(360*\x/8))});
                \draw[thick, domain=0:4, smooth, variable=\x, black] plot ({\x}, {1.5*(1+cos(360*\x/8))});
            \end{scope}
        }
        \foreach \x in {0, 4, 8, 12} {
            \draw [thick] ({\x}, 0) --({\x}, 3);
            \fill ({\x}, 0) circle [radius=1.5pt];
            \fill ({\x}, 3) circle [radius=1.5pt];
            \mpstensor{\x}{3}
            \node [font=\footnotesize] at ({\x}, 3) {$\mathsf{CCZ}$};
        }
        \draw [lightgray, rounded corners=5pt] ({2+1}, -0.4) rectangle ({6+1}, {3+1});
        \draw[
            thick,
            decorate,
            decoration={brace,mirror,amplitude=6pt}
        ]
            (-1.5, -0.75)
              --
            (13.5, -0.75)
            node[midway,below=5pt,font=\small] {$N$ times};
    \end{tikzpicture}
    \: = \Tr(T^N),
    \label{eqn:cluster-state-T}
\end{equation*}
where the 3-index tensor represented by the rounded square is the $\CCZ$ gate, and the black filled circle represents the 3-index Kronecker delta tensor $\delta_{ijk}$ in the $Z$ basis (equal to one if $i=j=k$ and zero otherwise). The gray box defines the transfer matrix $T$. This transfer matrix has a characteristic polynomial equal to
\begin{equation}
    (y^3 -2y - 2)y = 0,
    \label{eqn:characteristic-cubic}
\end{equation}
and the trace evaluates to $\Tr(T^{N}) = \sum_{i=1}^3 y_i^N$, where $y_i \in \mathbb{C}$ are the three nonzero solutions to the characteristic equation~\eqref{eqn:characteristic-cubic}. Denoting the root with the largest magnitude by $\lambda$, we have $\Tr(T^N) = \lambda^N[1+o(1)]$, where
\begin{equation}
    \lambda = \sum_{\sigma=\pm 1 }\frac{1}{3} \left[3 (9+\sigma \sqrt{57})
    \right]^{1/3} \approx 1.7693 .
    \label{eqn:cluster-state-lowerbound}
\end{equation}
We have verified numerically that $\vec{u} = \vec{v} = \vec{0}$ in fact corresponds to the optimal classical strategy for system sizes amenable to exact numerical enumeration of deterministic classical strategies.  

\paragraph{Upper bound} A simple upper bound on $\omega$ is obtained by applying the triangle inequality to Eq.~\eqref{eqn:generalized-Walsh-fixed}, as in Eq.~\eqref{eqn:walsh-triangle-inequality}, and bringing the maximum under the sum:
\begin{equation}
    \max_{\vec{u}, \vec{v}} \abs{W_C(\vec{u}, \vec{v})} 
    \leq 2^n\sum_{\vec{x} \in \F{2}^{n}} 
    \max_{\vec{u}, \vec{v}} \abs{  \braket{ \vec{+}_n | U_{\vec{u},\vec{v}}(\vec{x})  | +_n} }.
    \label{eqn:cluster-walsh-after-triangle}
\end{equation}
The maximization of each term under the sum can be achieved using the following Lemma.

\begin{restatable}{lemma}{zremove}
    \label{lem:z-removal}
    Let $\ket{C_\ell} = \prod_{i=1}^{\ell - 1} \CZ_{i,i+1}\ket{+_\ell}$ be the 1D cluster state with open boundaries on $\ell \geq 2$ qubits. For all $\vec{y} \in \F{2}^\ell$ the matrix element obeys the bound
    \begin{equation}
        \braket{ +_\ell | Z^{\vec{y}} | C_\ell  } \leq \braket{ +_\ell | C_\ell  }.
        \label{eqn:z-removal-obc}
    \end{equation}
    With periodic boundary conditions, define $\ket{C'_\ell} = \CZ_{\ell,1}\ket{C_\ell}$ for $\ell > 2$.
    Then, for even $\ell \geq 4$,
    \begin{equation}
        \braket{ +_\ell |  Z^{\vec{y}} | C'_{\ell} } \leq \braket{ +_\ell |  C'_{\ell} } .
        \label{eqn:z-removal-pbc}
    \end{equation}
\end{restatable}

\begin{proof}
    See appendix~\ref{app:walsh-spectrum-ghz}.
\end{proof}

Since the only effect of varying $(\vec{u} , \vec{v})$ for fixed $\vec{x}$ is to produce different patterns of $Z$ operators, each summand in Eq.~\eqref{eqn:cluster-walsh-after-triangle} can individually be upper bounded by the removal of all $Z$ operators using Lemma~\ref{lem:z-removal}. Note that we need $n$ to be even (i.e., $N$ to be an integer multiple of four) in order for Lemma~\ref{lem:z-removal} to be applicable to \emph{all} terms under the summation. This leads to
\begin{equation}
    \max_{\vec{u}, \vec{v}} \abs{W_C(\vec{u}, \vec{v})} 
    \leq
    2^n\sum_{\vec{x} \in \F{2}^{n}} 
    \braket{ \vec{+}_n | \prod_{k=1}^n \CZ_{k,k+1}^{x_k} | +_n} 
    \label{eqn:cluster-walsh-CZ-bound}
\end{equation}
The right-hand side can be evaluated by representing it as a matrix product state (MPS) with ``physical'' indices $\{ x_k \}$. Summation over $\vec{x}$ is then implemented by contracting the physical legs with the state $\ket{+}$:
\begin{equation*}    
    2^{n} \times \:
    \begin{tikzpicture}[baseline={(0, 0.4)}, x=2.4ex, y=2.4ex]
        \draw [thick] (-2, 2) -- (10, 2);
        \draw [thick] (-4, 2) -- (-4.5, 2) -- (-4.5, 2.25) -- (12.5, 2.25) -- (12.5, 2) -- (12, 2);
        \node at (-3, 2) {$\ldots$};
        \node at (11, 2) {$\ldots$};
        \foreach \x in {0, 4, 8} {
            \draw [thick] ({\x}, 0) --({\x}, 2);
            \mpstensor{\x}{2}
            \node at ({\x}, 2) {$A$};
        }
        \node [below] at (0, 0) {$+$};
        \node [below] at (4, 0) {$+$};
        \node [below] at (8, 0) {$+$};
    \end{tikzpicture}
    \label{eqn:sum-as-mps-overlap}
\end{equation*}
where the MPS tensor $A$ has components
\begin{equation*}
    \frac{1}{\sqrt{2}} A^0_{ij} = 
    \begin{tikzpicture}[baseline={(0, -0.25)}, x=2.4ex, y=2.4ex]
        \draw [thick, rounded corners=5pt] (-1, -0.75) to (-1, -0.5) to (-1, 0.5) to (-2, 0.5);
        \node [left] at (-2, 0.5) {$i$};
        \draw [thick, rounded corners=5pt] (1, -0.75) to (1, -0.5) to (1, 0.5) to (2, 0.5);
        \node [right] at (2, 0.5) {$j$};
        \node [below] at (-1, -0.75) {$+$};
        \node [below] at (1, -0.75) {$+$};
    \end{tikzpicture}
    \quad
    \frac{1}{\sqrt{2}} A^1_{ij} = 
    \begin{tikzpicture}[baseline={(0, -0.25)}, x=2.4ex, y=2.4ex]
        \draw [thick, rounded corners=5pt] (-1, -0.75) to (-1, -0.5) to (-1, 0.5) to (-2, 0.5);
        \node [left] at (-2, 0.5) {$i$};
        \draw [thick, rounded corners=5pt] (1, -0.75) to (1, -0.5) to (1, 0.5) to (2, 0.5);
        \node [right] at (2, 0.5) {$j$};
        \node [below] at (-1, -0.75) {$+$};
        \node [below] at (1, -0.75) {$+$};

        \draw [line width=1.5, red] (-1, -0.5) to (1, -0.5);
        \fill [red] (1, -0.5) circle [radius=1.5pt];
        \fill [red] (-1, -0.5) circle [radius=1.5pt];
        \node [above] at (0, -0.5) {\color{red}$\CZ$};
    \end{tikzpicture}
    \label{eqn:MPS-tensors-sum}
\end{equation*}
The leading behavior of Eq.~\eqref{eqn:cluster-walsh-CZ-bound} for large $n$ is obtained by finding the leading eigenvalue of the associated transfer matrix.
Defining $\sqrt{2}T = A^0 + A^1$, we find that
\begin{equation}
    T = \begin{pmatrix}
        1 & 1 \\
        1 & 0
    \end{pmatrix},
\end{equation}
whose leading eigenvalue is equal to the golden ratio, $\varphi$, leading to $\Tr( T^n ) \asymp \varphi^n$. 
Recall however that $n=N/2$. This means that the upper bound for the largest generalized Walsh coefficient grows as 
\begin{equation}
    \lim_{N \to \infty} \left( \max_{\vec{u}, \vec{v}} \abs{W_C(\vec{u}, \vec{v})} \right)^{1/N} \leq \sqrt{2\varphi} \approx 1.7989.
\end{equation}
The naive upper bound that combines the triangle inequality~\eqref{eqn:cluster-walsh-after-triangle} with Lemma~\ref{lem:z-removal} therefore gives a result that is within 2\% of the lower bound~\eqref{eqn:cluster-state-lowerbound}.


\subsection{2D toric code}
\label{sec:2D-toric-bounds}

Consider a 2D toric-code state with periodic boundary conditions.
We use the notation $X_e$, $Z_e$ for the Pauli $X$ and $Z$ operators on edge $e$. The toric code Hamiltonian~\cite{Kitaev2003} on an arbitrary cellulation of the torus is
\begin{equation}
    H_\text{TC} = - \sum_{v} A_v -  \sum_p B_p
    \, ,
    \label{eqn:H-toric-code}
\end{equation}
where $A_v = \prod_{e \in \delta v} Z_e$ and $B_p = \prod_{e \in \partial p} X_e$ are centered on the vertices ($v$) and plaquettes ($p$) of the lattice $\mathcal{L}$, respectively, and $\partial$ ($\delta$) is the boundary (coboundary) operator. The classical bits $\vec{x}$ are associated to $B_p$ operators and the bits $\vec{z}$ to $A_v$ operators. The question asked to the player associated to edge $e$ involves the bits 
\begin{equation}
    a_e \coloneq \bigoplus_{p \in \delta e} x_p \, ,
    \quad
    b_e \coloneq \bigoplus_{v \in \partial e} z_v
    \, .
\end{equation}
Physically, the bits $a_e$ ($b_e$) may be regarded as the domain-wall configuration of classical Ising spins $x_p$ ($z_v$) living on plaquettes (vertices), where the domain walls live on edges of the primary (dual) lattice. Since not all $B_p$ are independent, $\prod_p B_p = I$ (and similarly for $A_v$), the map from $(\vec{x}, \vec{z})$ to the pair $(\vec{a}, \vec{b})$ is neither injective nor surjective: There are exactly two $\vec{x}$ configurations for every $\vec{a}$ configuration and similarly for the map $\vec{z} \mapsto \vec{b}$. The two Ising spin configurations associated to the same domain-wall configuration are, e.g., $\vec{x} \to \vec{x} \oplus \vec{1}$. While the corresponding parity-check matrices are not full rank, this does not affect the success fraction of the optimal deterministic strategy.

Consider the function $\chi_H \coloneq (-1)^{f_H}$ with codomain $\{ +1, -1 \}$ with $f_H$ the target function. The function $\chi_H$ allows us to represent $f_H$ by a product of $\CZ$ and $\CCZ$ gates. Using the general form of the target function presented in appendix~\ref{app:computed-function}, there exists a $\CZ$ gate for every pair $(x_p, z_v)$ for which there exists an $e$ such that $v \in \partial e$ and $p \in \delta e$. Similarly, there exists a $\CCZ$ gate for every set $\{ x_p, z_v, z_{v'} \}$ for which there exists an $e$ such that $v, v' \in \partial e$ with $v\neq v'$ and $p \in \delta e$, and analogously for $\{ x_p, x_{p'}, z_v \}$. The corresponding function is more succinctly represented by introducing the so-called \emph{medial graph} $M(\mathcal{L})$ and its dual, $M^*(\mathcal{L})$. Given a graph $\mathcal{L}$, the edges of $\mathcal{L}$ correspond to vertices of $M(\mathcal{L})$. Two vertices of the medial graph are connected by an edge if the associated edges of the original graph are incident to the same vertex and face of the original graph. For example, the medial graph of the square (honeycomb) lattice is a rotated square (kagome) lattice. The dual medial graph leads to a cellulation consisting of four-sided plaquettes, as depicted in Fig.~\ref{fig:generic-TC-function}.

Equipped with the dual medial graph, the target function for a generic toric-code game on a lattice $\mathcal{L}$ can be written succinctly as
\begin{equation}
    f_{\text{TC}}(\vec{w}) =
    \bigoplus_{f} \bigoplus_{ \{ i, j, k \} \in f } w_i w_j w_k \oplus
    \bigoplus_{\langle ij \rangle} w_i w_j 
    ,
    \label{eqn:toric-code-function}
\end{equation}
where the first sum is over all sets of three indices belonging each face $f \in F[M^*(\mathcal{L})]$ of the dual medial graph, and $\langle ij \rangle \in E[M^*(\mathcal{L})]$ denotes neighboring sites.
Vertices, edges, and faces of graph $G$ are denoted by $V[G]$, $E[G]$, and $F[G]$, respectively. This function is also depicted graphically in Fig.~\ref{fig:generic-TC-function} for a generic plane graph.

\begin{figure}[t]
    \centering
    \includegraphics[width=\linewidth]{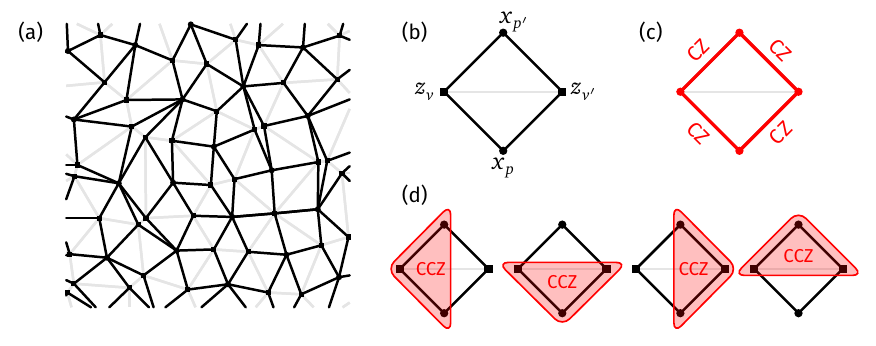}
    \caption{Graphical illustration of the function computed by the toric-code game on a lattice corresponding to a generic tesselation. \textbf{Left}: Generic tesselation $\mathcal{L}$ constructed from a Voronoi diagram (gray) and the associated dual medial graph $M^*(\mathcal{L})$ (black). \textbf{Right}: Locations of $\CZ$ and $\CCZ$ gates given the dual medial graph. For every edge of $M^*(\mathcal{L})$, there exists a $\CZ$ gate. Given a face of $M^*(\mathcal{L})$, there exists a $\CCZ$ gate for every triple of sites.}
    \label{fig:generic-TC-function}
\end{figure}


\subsubsection{Unrestricted inputs}

Deterministic classical strategies again correspond to a specific family of quadratic Boolean functions
\begin{equation}
    c(\vec{w}) = u_0 \oplus \bigoplus_{i} u_i w_i \oplus \bigoplus_{\langle i,j \rangle } v_{ij} w_i w_j
    ,
\end{equation}
parametrized by $u_0\in \F{2}$, $\vec{u}\in \F{2}^N$, and $\vec{v}\in \F{2}^{2N}$, where $i,j$ run over vertices of the dual medial graph. 
Note that not all entries $v_{ij}$ are independent. Specifically, the vector $\vec{v}$ is derived from the deterministic strategies of players on the edges of the original lattice, which are in one-to-one correspondence with faces $f$ of the dual medial graph. The entries $v_{ij}$ satisfy
\begin{equation}
    v_{ij} = \bigoplus_{f \in \delta \langle ij \rangle} v_f,
\end{equation}
where edges of the original lattice are indexed by $f$. We may therefore interpret valid $v_{ij}$ configurations (i.e., corresponding to implementable deterministic classical strategies) as domain walls that form closed loops, which live on the edges of the dual medial graph.
Combining these ingredients, the relevant generalized Walsh coefficient is given by
\begin{equation}
    W_{\text{TC}}(\vec{u}, \vec{v}) =
    \sum_{\vec{w} \in \F{2}^{N}} (-1)^{ \bigoplus_i u_iw_i + \bigoplus_{\langle ij \rangle} (1 \oplus v_{ij}) w_{i} w_{j} + \sum_{\triangle} w_{i} w_{j} w_{k} },
\end{equation}
where the sum over triangles is shorthand for the sum over triples in Eq.~\eqref{eqn:toric-code-function}.
As in Eq.~\eqref{eqn:cluster-state-walsh-quantum}, this expression can also be interpreted as a quantum matrix element where degree-$n$ terms correspond to unitary $\mathsf{C}^{n}\mathsf{Z}$ gates.

We now give the generalized Walsh coefficient a convenient graphical interpretation. Place bits $x_p$ on plaquette centers and bits $z_v$ on vertices. With this spatial bit arrangement, consider fixing the bits $\vec{x}$ on plaquettes.
This gives
\begin{equation}
    W_{\text{TC}}(\vec{u}, \vec{v}) = 2^n\sum_{\vec{x} \in \F{2}^{N-n}} \braket{ \vec{+}_n | U_{\vec{u},\vec{v}}(\vec{x})  | +_n},
    \label{eqn:generalized-Walsh-fixed-TC}
\end{equation}
for some $\vec{x}$-dependent unitary circuit consisting of $Z$'s and $\CZ$'s with $n=N_v$ the number of vertices. For the remainder of this discussion we focus our attention on the honeycomb lattice for specificity. Then, the $\CCZ$ gates, which depend only on $\vec{x}$, reduce to
\begin{equation*}
    \begin{tikzpicture}[x=2.4ex,y=2.4ex]
        \foreach \x in {0, 1, 2} {
            \draw [thick, lightgray, opacity=0.75] (0, 0) to ({90+120*\x}:{4*\s});
            \node [font=\small, gray] at ({30+120*\x}:1) {0};
        }
    \end{tikzpicture}
    \qquad
    \begin{tikzpicture}[x=2.4ex,y=2.4ex]
        \foreach \x in {0, 1, 2} {
            \draw [thick, lightgray, opacity=0.75] (0, 0) to ({90+120*\x}:{4*\s});
        }
        \node [font=\small, gray] at ({30+120*0}:1) {0};
        \node [font=\small, gray] at ({30+120*1}:1) {0};
        \node [font=\small, gray] at ({30+120*2}:1) {1};

        \draw [line width=1.25, xpauli, line cap=round] (0, 0) to ({90+120*1}:{4*\s});
        \draw [line width=1.25, xpauli, line cap=round] (0, 0) to ({90+120*2}:{4*\s});
        \fill [xpauli] ({90+120*1}:{4*\s}) circle [radius=1.5pt];
        \fill [xpauli] ({90+120*2}:{4*\s}) circle [radius=1.5pt];
        \fill [zpauli] (-2pt, -2pt) rectangle (2pt, 2pt);
    \end{tikzpicture}
    \qquad
    \begin{tikzpicture}[x=2.4ex,y=2.4ex]
        \foreach \x in {0, 1, 2} {
            \draw [thick, lightgray, opacity=0.75] (0, 0) to ({90+120*\x}:{4*\s});
        }
        \node [font=\small, gray] at ({30+120*0}:1) {1};
        \node [font=\small, gray] at ({30+120*1}:1) {1};
        \node [font=\small, gray] at ({30+120*2}:1) {0};

        \draw [line width=1.25, xpauli, line cap=round] (0, 0) to ({90+120*1}:{4*\s});
        \draw [line width=1.25, xpauli, line cap=round] (0, 0) to ({90+120*2}:{4*\s});
        \fill [xpauli] ({90+120*1}:{4*\s}) circle [radius=1.5pt];
        \fill [xpauli] ({90+120*2}:{4*\s}) circle [radius=1.5pt];
        \fill [zpauli] (-2pt, -2pt) rectangle (2pt, 2pt);
    \end{tikzpicture}
    \qquad
    \begin{tikzpicture}[x=2.4ex,y=2.4ex]
        \foreach \x in {0, 1, 2} {
            \draw [thick, lightgray, opacity=0.75] (0, 0) to ({90+120*\x}:{4*\s});
            \node [font=\small, gray] at ({30+120*\x}:1) {1};
        }
    \end{tikzpicture}
\end{equation*}
As before, the red lines denote the application of a $\CZ$ gate, while the blue square denotes the application of $Z$ operator. Thus, given some configuration $\vec{x}$ on plaquettes, the $\CZ$ gates live on the domain walls between
the bits on plaquette centers, and therefore form closed loops. 
An example is shown in Fig.~\ref{fig:loop-model-bound}. The $\CZ$ and $Z$ gates in Eq.~\eqref{eqn:generalized-Walsh-fixed-TC} give rise to additional $Z$ operators, the precise locations of which depend on $\vec{u}$, $\vec{v}$, in addition to $\vec{x}$. Each term in Eq.~\eqref{eqn:generalized-Walsh-fixed-TC} now corresponds to a graph state and can be evaluated using the methods of Sec.~\ref{sec:walsh-from-graphs}.

\begin{figure}[t]
    \centering
    \includegraphics[width=\linewidth]{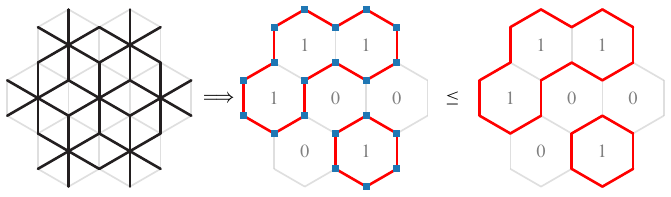}
    \caption{Mapping to a nonintersecting loop model on the honeycomb lattice. Left: the lattice whose vertices enter Eq.~\eqref{eqn:toric-code-function}. $\CZ$ gates act on every edge, and $\CCZ$ gates act on all triples of spins belonging to a face. If we fix the spins on the plaquettes of the original honeycomb lattice, then we obtain $\CZ$ ($Z$) gates on all edges (vertices) that host a domain wall. A bound on the optimal $p_\text{cl}$ may be obtained by removing all $Z$'s.}
    \label{fig:loop-model-bound}
\end{figure}

\paragraph{Upper bound} 
We proceed analogously to the 1D cluster state example. An upper bound is obtained by applying the triangle inequality to Eq.~\eqref{eqn:generalized-Walsh-fixed-TC}. The maximization of each term can then be performed by applying Lemma~\ref{lem:z-removal}. Note that closed loops on the honeycomb lattice always consist of an even number of edges, so the bound~\eqref{eqn:z-removal-pbc} can always be applied. Given a configuration $\vec{x}$ on plaquettes, removal of all $Z$'s leads us to consider
\begin{equation}
    \max_{\vec{u}, \vec{v}} \abs{W_{\text{TC}}(\vec{u}, \vec{v})} 
    \leq 2^n \sum_{\vec{x} \in \F{2}^{N-n}}
    \braket{+_{\ell_1} | C'_{\ell_1}} \braket{+_{\ell_2} | C'_{\ell_2}} \cdots
\end{equation}
where $\ell_i$ is the number of edges that participate in the $i$th closed loop and $\ket{C'_\ell}$ is the cyclic cluster state on $\ell$ qubits. The example configuration in Fig.~\ref{fig:loop-model-bound} would give two decoupled cluster states of lengths 6 and 14. The sum over $\vec{x}$ generates all domain-wall configurations.
Since $\braket{ +_\ell |  C'_{\ell} } = 2^{1-\ell/2}$, we can write
\begin{equation}
    \max_{\vec{u}, \vec{v}} \abs{W_{\text{TC}}(\vec{u}, \vec{v})} 
    \leq
    2^{1+N_v/2} \sum_{\{ \text{dw} \} } t^{N_v-\ell[\text{dw}]} n^{N[\text{dw}]}
    \label{eqn:partition-function-bound}
\end{equation}
where $\{ \text{dw} \}$ is the set of all domain-wall configurations on the honeycomb lattice, $\ell[\text{dw}]$ is the number of edges occupied by domain walls, and $N[\text{dw}]$ is equal to the number of loops. The parameters $(t, n)$ are equal to $(2^{1/2}, 2)$. The right-hand side of Eq.~\eqref{eqn:partition-function-bound} can be interpreted as the partition function of a classical, nonintersecting loop model on the honeycomb lattice, and the specific parameters $(t, n)$ place the model on a critical line, $t^2 = 2 + \sqrt{2-n}$, along which it is exactly solvable by Bethe ansatz~\cite{Baxter1986Loops,Baxter_1987}.
We have
\begin{equation}
    Z_\text{loop} \coloneq \sum_{\{ \text{dw} \} } t^{N_v-\ell[\text{dw}]} n^{N[\text{dw}]}, 
    \quad 
    \lim_{N_v \to \infty} Z_\text{loop}^{1/N_v} = W.
    \label{eqn:classical-loop-partition-function}
\end{equation}
The asymptotic growth of the partition is therefore controlled by $W$, which is expressed by
\begin{equation}
    \ln W = \frac12 \int_{-\infty}^{\infty} dx \frac{\sinh \left(\frac{\pi }{2}-\lambda \right) x \sinh \lambda  x}{x \sinh \frac{\pi  x}{2} (2 \cosh \lambda  x-1)} 
\end{equation}
The parameter $\lambda$ determines $(t, n) = (2\cos\lambda, -2\cos4\lambda)$.
Given the parameters above, we have $\lambda = \pi/4$, which allows the integral to be performed exactly.
As we show in Appendix~\ref{app:toric-code-integral}, the integral is equal to $4\ln W = I(0)$ with
\begin{align}
    I(a) &\coloneq \int_{-\infty}^{\infty} dx \,
    \frac{x}{a^2+x^2} \frac{\tanh 2\pi x }{ 2 \cosh 2 \pi x -1} \\
    &= \psi(a+\tfrac14)+\psi(a+\tfrac34) - \psi(a+\tfrac16) - \psi(a+\tfrac56) \notag
\end{align}
with $\psi(x)$ the digamma function. Plugging this result back into the bound~\eqref{eqn:partition-function-bound}, we arrive at
\begin{equation}
    \lim_{N \to \infty} \left( \max_{\vec{u}, \vec{v}} \abs{W_{\text{TC}}(\vec{u}, \vec{v})} \right)^{1/N} \leq \sqrt{3}.
\end{equation}

In the context of quantum error correction, there exists another statistical-mechanics mapping to an Ising model when computing the threshold of the toric code in the presence of iid Pauli noise~\cite{Dennis2002}. It is interesting that a different classical partition function~\eqref{eqn:classical-loop-partition-function} controls the contextuality of single-site Pauli measurements in toric-code states. It is also worthy of note that the same partition function arises in the context of strange correlators in symmetry-protected topological (SPT) states~\cite{ScaffidiGapless2017,YouStrangeCorrelator2014}.
The correspondence with strange correlators can be traced back to the hypergraph-state representation of Walsh coefficients (Sec.~\ref{sec:hypergraph-states}). Generalized Walsh coefficients correspond to the overlap of a hypergraph state, which can be interpreted as an SPT wave function, with lower-order graph states. Hence, $Z_{\Psi_{\text{SPT}}} = \braket{\Psi_{\text{Trivial}} | \Psi_{\text{SPT}}}$ is related to the generalized Walsh spectrum. 


\subsubsection{Restricting inputs}

Consider now the toric code state on an $L \times L$ square lattice with degrees of freedom on the edges with $L$ even.
Since the lattice is bipartite, it is possible to fix $\vec{a}=\vec{1}$ by choosing an antiferromagnetic configuration for the $\vec{x}$ inputs (recall that $\vec{a}$ may be interpreted as the domain-wall configuration associated to $\vec{x}$). From Eq.~\eqref{eqn:toric-code-function}, the target function in this case is
\begin{equation}
    f_\text{TC}(\vec{z}) = \bigoplus_{\langle v v' \rangle} z_{v} z_{v'}
    ,
    \label{eqn:restricted-TC-function}
\end{equation}
which is a natural generalization of the target function in the GHZ game to two dimensions. Consider the Walsh spectrum of Eq.~\eqref{eqn:restricted-TC-function}.
Applying Theorem~\ref{thm:walsh-from-symmetries}, the auxiliary graph state is now the 2D cluster state, which possesses the following $X$ symmetries:
\begin{equation}
    \prod_{i} X_{i, k+i} ,\: \prod_{i} X_{i, k-i},
    \quad 
    \prod_{i+j \text{ even}} X_{i,j}, \: \prod_{i+j \text{ odd}} X_{i,j}.
\end{equation}
The first two are subsystem symmetries along diagonal lines (indexed by the offset $k$), while the latter two are global symmetries associated to the even and odd sublattices. Note that $\prod_{k,i}X_{i,k+i}= \prod_{k,i}X_{i,k-i} =\prod_{i,j}X_{i,j}$, such that there are only $2L$ \emph{independent} $X$ symmetries.
The nonzero Walsh coefficients therefore have magnitude $\abs*{W}^2=2^{L^2 + 2L}$.

This result may alternatively be understood via a mapping to the partition function of a plaquette Ising model.
Using Lemma~\ref{lem:walsh-from-cluster}, we write 
\begin{equation}
    W_{f_\text{TC}}(\vec{0}) = 2^{L^2}
    \begin{tikzpicture}[scale=1,baseline={(0, 1.1)},x=2.4ex,y=2.4ex]
        \pgfmathsetmacro{\dx}{2};
        \foreach \j in {0,1,2,3} {
            \draw [thick] (0, \dx*\j) -- ++(-\dx/3, 0);
            \draw [thick] (3*\dx, \dx*\j) -- ++(\dx/3, 0);
            \draw [thick] (\dx*\j, 0) -- ++(0, -\dx/3);
            \draw [thick] (\dx*\j, 3*\dx) -- ++(0, +\dx/3);
        }
        \foreach \i in {0,1,2,3} {
            \foreach \j in {0,1,2,3} {
                \ifnum \i<3
                    \draw [thick] (\dx*\i,\dx*\j) --++ (\dx,0);
                    \node[fill=white, inner sep=1pt] at ({\dx*\i+\dx*0.5},\dx*\j) {\small $H$};
                \fi
                \ifnum \j<3
                    \draw [thick] (\dx*\i,\dx*\j) --++ (0,\dx);
                    \node[fill=white, inner sep=1pt] at (\dx*\i,{\dx*\j+\dx*0.5}) {\small $H$};
                \fi
                \draw[thick, fill=black] (\dx*\i,\dx*\j,0) circle (1.25pt);
            }
        }
    \end{tikzpicture}
    \label{eqn:TC-fixed-walsh}
\end{equation}
where the filled black circles represent the delta tensor in the $Z$ basis, $\delta_{ijkl}=\delta_{ij}\delta_{jk}\delta_{kl}$.
Suppose now that we fix the bits on the even sublattice. To evaluate Eq.~\eqref{eqn:TC-fixed-walsh}, it suffices to consider an elementary unit
\begin{equation}
    \begin{tikzpicture}[scale=1,baseline={(0, -0.2)},x=2.4ex,y=2.4ex]
        \pgfmathsetmacro{\dx}{2};
        
        \draw [thick] (-\dx,0) --++ (2*\dx,0);
        \node[fill=white, inner sep=1pt] at ({-\dx/2}, 0) {\small $H$};
        \node[fill=white, inner sep=1pt] at ({\dx/2}, 0) {\small $H$};
            
        \draw [thick] (0,-\dx) --++ (0,2*\dx);
        \node[fill=white, inner sep=1pt] at (0, \dx/2) {\small $H$};
        \node[fill=white, inner sep=1pt] at (0, -\dx/2) {\small $H$};

        \draw[thick, fill=black] (0, 0) circle (1.25pt);

        \node [above] at (0, \dx)  {$i$};
        \node [right] at (\dx, 0)  {$j$};
        \node [below] at (0, -\dx) {$k$};
        \node [left] at (-\dx, 0)  {$l$};
    \end{tikzpicture}
    = \frac{1}{2}
    \braket{ + | Z^{i+j+k+l} | + },
    \label{eqn:TC-fixed-constraint}
\end{equation}
which vanishes unless $i \oplus j \oplus k \oplus l = 0$. Thus, summing over the bits on the even sublattice, Eq.~\eqref{eqn:TC-fixed-constraint} implements a local constraint on bit configurations. Let $\vec{b}$ be the bitstring associated to $\vec{z}$ bits on the even sublattice.
\begin{align}
    W_{f_\text{TC}}(\vec{0}) &= 2^{L^2/2} \sum_{ \vec{b} } \prod_{v} \delta\left( \bigoplus_{v' \in N(v)} b_{v'} \right) \\
    &\eqcolon 2^{L^2/2} Z_\square(L^2/2)
    ,
\end{align}
where $N(v)$ denotes the neighborhood of vertex $v$ on the larger rotated square lattice formed by even vertices of the original square lattice. The object $Z_\square(M)$ is the zero-temperature partition function of the plaquette Ising model, also known as the Xu-Moore model~\cite{XuMoore1, XuMoore2}, on a square lattice with periodic boundaries and $M$ sites, which counts the number of ground states of the model. With $L^2/2$ sites, the number of ground states is equal to $2^L$, such that the nonzero Walsh coefficients have magnitude $\abs{W} = 2^{L^2/2 + L}$ (recall that all nonzero Walsh coefficients have the same magnitude).


\section{Beyond XOR games}
\label{sec:submeasurement}


\subsection{Pauli submeasurement games}
\label{sec:submeasurement-game}

As noted in Sec.~\ref{sec:quantum-strats}, the games considered thus far possess an alternative strategy in which players share a GHZ state and employ a modified Pauli strategy. Here, we consider a slight generalization of the win conditions that prevent GHZ states from providing a perfect quantum strategy. The game is based on the concept of \emph{Pauli submeasurements} (see Refs.~\cite{Meyer2023submeasurement, Daniel2022Exp}). Intuitively, while XOR games constrain only the sum of all players' answers, there can be finer structure when the players collectively measure stabilizers that factorize into products of stabilizers supported on different regions.

\begin{Def}[Pauli substrings]
    Given a set of sites $\Lambda$, the set of \emph{Pauli substrings} of a Pauli string $P = \prod_{i \in \Lambda} P_i$, with each $P_i \in \{ I, X, Y, Z \}$, is given by 
    \begin{equation}
        \sub(P) = \left\{ 
            \prod_{i \in K} P_i : K \subseteq \Lambda 
        \right\}.
    \end{equation}
\end{Def}

\begin{Def}[CSS submeasurement game]
    \label{def:submeasurement-game}
    Associate a player to every qubit of a CSS code on $N$ qubits defined by the pair of {full-rank} parity-check matrices $H = (H_X, H_Z)$ with stabilizer group $\mathcal{S}$. Player $i$ is handed an ordered pair of bits $(a_i, b_i) \in \F{2}^2$, with $\vec{a} \in I_X \subseteq \Image H_X^\tp$ and $\vec{b} \in I_Z \subseteq \Image H_Z^\tp$, where neither $I_X$ or $I_Z$ is equal to the empty set. The inputs are collectively denoted by the pair $I = (I_X, I_Z)$. Once the players have been handed their bits, any classical communication is forbidden. Every player outputs a bit, $y_i(a_i, b_i)$. 
    Collectively, the players win the game if, for all $P$ such that
    $P \in \sub(\prod_{i \in \Lambda} P_i(a_i, b_i))$ and $P \in \pm \mathcal{S}$,
    \begin{equation}
        \sum_{i\in \Lambda_P} y_i \equiv \frac12 \sum_{i \in \Lambda_P} a_i b_i 
        \quad {(\text{mod 2})}
        \, ,
        \label{eqn:submeasurement-game-victory}
    \end{equation}
    where $\Lambda_P = \supp(P)$ is the set of sites on which $P$ acts nontrivially. This defines the game $\mathsf{SUB}(H | I)$. If the inputs are unrestricted, we refer to this game as $\mathsf{SUB}(H)$.
\end{Def}

By definition, there exists a constraint for each Pauli substring that belongs to the stabilizer group (up to a sign). As a simple example, consider the four-qubit GHZ state and the Pauli string $ZZZZ$ (shorthand for $Z \otimes Z \otimes Z \otimes Z$). Pauli substrings are found by replacing operators with identities. The Pauli substrings $ZZII$, $IZZI$, $IIZZ$ belong to $\mathcal{S}$, whereas, e.g., $ZIII$ does not. Conversely, given the Pauli string $YYXX \in -\mathcal{S}$, there are no nontrivial substrings (i.e., $\notin \{ IIII, YYXX \}$) that belong to $\pm \mathcal{S}$. These examples also highlight that not all constraints~\eqref{eqn:submeasurement-game-victory} are necessarily independent.

Note that the definition could be further generalized such that the original Pauli string $\prod_i P_i$ does not belong to the stabilizer group, but includes nontrivial Pauli substrings that do. We do not require this more general definition for our purposes.


\subsection{Quantum strategy}

To prove that the Pauli strategy from Prop.~\ref{prop:perfect-quantum-strat} remains a perfect quantum strategy, we need a generalization of Lemma~\ref{lem:submeasurement-dist} that incorporates multiple constraints.

\begin{restatable}{lemma}{measDistribution}
    \label{lem:outcome-distribution-general}
    Let $\Lambda$ be the set of all lattice sites and consider a set of subsets of $\Lambda$, $C = \{ I_1, I_2, \dots \}$ with each $I_k \subseteq \Lambda$ such that the sets do not overlap $\bigcap_{k=1}^{\abs{C}} I_k = \emptyset$.
    Suppose that the operator $O_i$, satisfying $O_i^2 = I$, is measured on site $i$ for all $i \in \Lambda$, producing the measurement outcome $(-1)^{s_i}$. If the system is in the state $\ket{\psi}$ then
    \begin{equation}
        \Pr\left(\bigwedge_{k=1}^{\abs{C}} \bigoplus_{i \in I_k} s_i = f_k\right)  =
        \braket{\psi | \prod_{k=1}^{\abs{C}} 
        \frac{1}{2} 
        \left( 1 + (-1)^{f_k} \prod_{i \in I_k}  O_i \right)   | \psi}
        \label{eqn:general-outcome-dist-nonoverlap}
    \end{equation}
    with each $f_k \in \F{2}$. If the sets are no longer forbidden from intersecting, the most general expression is
    \begin{equation}
        \Pr =
        \frac{1}{2^{\abs{C}}} \sum_{\{r_c \}} (-1)^{\bigoplus_{c \in C} r_c f_c} \langle \psi | \prod_{i \in \Lambda} O_i^{\bigoplus_{c \ni i} r_c} | \psi \rangle
        \label{eqn:general-outcome-dist}
    \end{equation}
    with $c \ni i$ the constraints $c \in C$ that contain site $i$.
\end{restatable}

\begin{proof}
    See Appendix~\ref{app:outcome-dist-proof}.
\end{proof}

Working with the most general expression Eq.~\eqref{eqn:general-outcome-dist}, consider the Pauli strategy in which the operator $P_i(a_i, b_i)$ is measured on site $i$, where the set $C$ is determined by Eq.~\eqref{eqn:submeasurement-game-victory}.
We have 
\begin{equation}
    \prod_{i \in \Lambda} P_i(a_i, b_i)^{\bigoplus_{c \ni i} r_c} = (-1)^{\frac12 \sum_c r_c \sum_{i \in c} a_i b_i} \prod_{i \in \Lambda} (X^{a_i} Z^{b_i} )^{\bigoplus_{c \ni i} r_c}
\end{equation}
where the operator on the right-hand side belongs to the stabilizer group by definition of the constraints.
Hence, the probability of winning for a specific query is
\begin{equation}
    \Pr\left( \bigwedge_{c \in C} \bigoplus_{i \in c} s_i = f_c \right) = \frac{1}{2^{\abs{C}}} \sum_{\{r_c \}} (-1)^{\sum_{c \in C} r_c (f_c + \frac12 \sum_{i \in c} a_i b_i ) } 
\end{equation}
which evaluates to a product of Kronecker deltas enforcing that $\bigoplus_{i \in c} s_i \equiv \frac12 \sum_{i \in c} a_i b_i$ (mod 2) for all $c \in C$. Hence, the strategy in which players answer with their measurement outcome $y_i = s_i$ is perfect.


\subsection{Rigidity of submeasurement games}

Here, we prove that submeasurement games can be won with certainty \emph{only} by the corresponding CSS codeword, up to local isometries, for the specific case of the 2D toric code on the square lattice.

Consider $N$ players attempting to win the submeasurement game $\mathsf{SUB}(H)$ for parity-check matrices $H$ defining the 2D toric code (see Sec.~\ref{sec:2D-toric-bounds}). In a general quantum strategy, the players share a quantum state $\ket{\psi}$ and the player on edge $e$ measures the Hermitian operators
\begin{equation}
    \begin{array}{cc}
         E_e & \text{if } \: (a_e, b_e) = (0, 0)  \\
         A_e & \text{if } \: (a_e, b_e) = (0, 1)  \\
         B_e & \text{if } \: (a_e, b_e) = (1, 0)  \\
         C_e & \text{if } \: (a_e, b_e) = (1, 1)
    \end{array}
\end{equation}
thereby producing an outcome $\in \{+1, -1 \}$. We assume that these operators are $(\pm 1)$-valued, satisfying $O_e^2 = I$ for all operators $O_e \in \{ E_e, A_e, B_e, C_e \}$, and are thus unitary. 
Note that for every stabilizer generator $S \in \mathcal{S}$ there exists a corresponding query $(\vec{a}, \vec{0})$ or $(\vec{0}, \vec{b})$. Since, by assumption, the strategy succeeds with certainty, it must win with $\Pr=1$ for every query.
Thus, for $Z$-type and $X$-type stabilizer generators, we respectively have that
\begin{align}
    \braket{\psi | \prod_{e \in \partial^\dagger v} A_e | \psi} &= 1 \implies \prod_{e \in \partial^\dagger v} A_e \ket{\psi} = \ket{\psi} \\
    \braket{\psi | \prod_{e \in \partial p} B_e | \psi} &= 1 \implies \, \prod_{e \in \partial p} B_e \ket{\psi} = \ket{\psi}
    \label{eqn:self-test-stabilizer-constraint}
\end{align}
where the second equalities follow from the fact that $\norm{O_e \ket{\psi}} = 1$. Note that success at the game $\mathsf{CSS}(H)$ would have constrained the sum of \emph{all} measurement outcomes rather than only the outcomes contributing to the support of the stabilizer generators.

\emph{Anticommutation}.---We now prove anticommutativity of the operators, $\{ A_e, C_e \} \ket{\psi} = \{ B_e, C_e \} \ket{\psi}=0$. Consider the following set of operators:
\begin{subequations}
\begin{align}
    S_{12} &\coloneqq
    \begin{tikzpicture}[x=2.4ex,y=2.4ex,baseline={(0,0.1)}]
    \draw [line width=1.25, zpauli, line join=round] ({-1+1.5}, {-1+1.5}) rectangle ({1+1.5}, {1+1.5});
    \draw [line width=1.25, xpauli, line join=round] ({-1+0.5}, {-1+0.5}) rectangle ({1+0.5}, {1+0.5});
    \draw [line width=1, fill=white] (+1.5, 0.5) circle [radius=2pt];
    \draw [line width=1, fill=white] (0.5, 1.5) circle [radius=2pt];
    \end{tikzpicture}
    \qquad\qquad
    S_{23} \coloneqq
    \begin{tikzpicture}[x=2.4ex,y=2.4ex,baseline={(0,0.1)}]
    \draw [line width=1.25, zpauli, line join=round] ({-1+1.5}, {-1-0.5}) rectangle ({1+1.5}, {1-0.5});
    \draw [line width=1.25, xpauli, line join=round] ({-1+0.5}, {-1+0.5}) rectangle ({1+0.5}, {1+0.5});
    \draw [line width=1, fill=white] (1.5, 0.5) circle [radius=2pt];
    \draw [line width=1, fill=white] (0.5, -0.5) circle [radius=2pt];
    \end{tikzpicture} \\[0.1cm]
    S_{31} &\coloneqq
    \begin{tikzpicture}[x=2.4ex,y=2.4ex,baseline={(0,0.1)}]
    \draw [line width=1.25, zpauli, line join=round] ({-1+1.5}, {-1-0.5}) rectangle ({1+1.5}, {1+1.5});
    \draw [line width=1.25, xpauli, line join=round] ({-1+0.5}, {-1+0.5}) rectangle ({1+0.5}, {1+0.5});
    \draw [line width=1, fill=white] (0.5, 1.5) circle [radius=2pt];
    \draw [line width=1, fill=white] (0.5, -0.5) circle [radius=2pt];
    \end{tikzpicture}
    \qquad\qquad
    S_{0} \coloneqq
    \begin{tikzpicture}[x=2.4ex,y=2.4ex,baseline={(0,0.1)}]
    \draw [line width=1.25, xpauli, line join=round] ({-1+0.5}, {-1+0.5}) rectangle ({1+0.5}, {1+0.5});
    \end{tikzpicture}
\end{align}%
\label{eqn:GHZ-embedding}%
\end{subequations}
where red edges live on the primary lattice and denote $A_e$, blue edges live on the dual (square) lattice and denote $B_e$, and white circles represent $C_e$. The win conditions ensure that these operators obey
\begin{equation}
    (S_{12} + S_{23} + S_{31} - S_0)^2 \ket{\psi} = 16\ket{\psi} 
    \label{eqn:self-testing-constraint}
\end{equation}
since they satisfy $S_{12} \ket{\psi} = S_{23} \ket{\psi} = S_{31} \ket{\psi} = -\ket{\psi}$, whereas $S_0 \ket{\psi} = \ket{\psi}$. As in Eq.~\eqref{eqn:self-test-stabilizer-constraint}, the operators $S_{ab}$ are supported locally as a result of Eq.~\eqref{eqn:general-outcome-dist-nonoverlap}.
Expanding the square in Eq.~\eqref{eqn:self-testing-constraint}, and using unitarity (which follows from $O_e^2=I$ on every edge) we arrive at
\begin{subequations}
\begin{gather}
    (T_{23} + T_{31} + T_{12}) \ket{\psi} = 12 \ket{\psi} \label{eqn:self-testing-constraint-expanded} \\
    T_{ab} \coloneq \{S_{ca}, S_{bc}\} - \{S_{ab}, S_0\} \quad c\neq a,b
\end{gather}%
\end{subequations}
Consider the operators that constitute each $T_{ab}$. Labeling sites $1, \dots, 4$ clockwise around the edges of the red plaquette in Eq.~\eqref{eqn:GHZ-embedding}, starting at the top, 
\begin{align*}
    T_{23} &= 
    I \otimes i[B_2, C_2] \otimes i[B_3, C_3] \otimes I \otimes U \\
    T_{31} &= 
    i[B_1, C_1] \otimes I \otimes i[B_3, C_3] \otimes I \otimes U' \\
    T_{12} &= 
    i[B_1, C_1] \otimes i[B_2, C_2] \otimes I \otimes I \otimes U'' 
\end{align*}
where the final factor in the tensor product corresponds to some product of $A_e$ operators and is thus a Hermitian operator with eigenvalues $\pm 1$. Furthermore, each commutator of the form $i[B, C]$ is a Hermitian operator with eigenvalues between $-2$ and $2$ (inclusive of boundaries) since $\norm{B_e} = \norm{C_e} = 1$. By submultiplicativity of the operator norm,
\begin{equation}
    \norm{T_{ab}\ket{\psi}} \leq \norm{i[B_a, C_a]} \cdot \norm{i [B_b, C_b]} \cdot \norm{U} \leq 4
\end{equation}
Next, applying the triangle inequality to the left-hand side of Eq.~\eqref{eqn:self-testing-constraint-expanded},
\begin{equation}
    \norm{T_{12} \ket{\psi}} + \norm{T_{23} \ket{\psi}} + \norm{T_{31} \ket{\psi}} \leq 12
\end{equation}
However, by Eq.~\eqref{eqn:self-testing-constraint-expanded} the inequality is actually an equality. Thus, the states $T_{ab}\ket{\psi}$ have norm 4 and are all parallel to one another. It follows that for each operator $T_{ab}$ we have an equation of the form
\begin{equation}
    I \otimes i[B_2, C_2] \otimes i[B_3, C_3] \otimes I \otimes U  \ket{\psi} = 4 \ket{\psi}.
    \label{eqn:self-testing-23-constraint}
\end{equation}
Define the two states $\ket{\phi} = i[B_2, C_2] \ket{\psi}$ and $\ket{\chi} = i[B_3, C_3] U \ket{\psi}$. Note that we have abused notation by dropping identities in the tensor product. We can now apply the Cauchy-Schwarz inequality
\begin{equation}
    \abs{\braket{\phi | \chi}} \leq \norm{\ket{\phi}} \cdot \norm{\ket{\chi}} \leq 4
\end{equation}
since $\norm{i[B_a, C_a]} \leq 2$ and $\norm{U} \leq 1$ such that $\norm{\ket{\phi}} \leq 2$ and $\norm{\ket{\chi}} \leq 2$. However, the inequality is again actually an equality by virtue of Eq.~\eqref{eqn:self-testing-23-constraint}. This implies that one of $\ket{\phi}$ and $\ket{\chi}$ is a scalar multiple of the other and $\norm{\ket{\phi}} = \norm{\ket{\chi}} = 2$ so $\ket{\phi} = \ket{\chi}$. Using this relation allows us to trade $i[B_2, C_2]$ for $i[B_3, C_3]U$ when acting on $\ket{\psi}$, which finally leads to the constraint 
\begin{equation}
    (i[B_2, C_2])^2 \ket{\psi} = 4 \ket{\psi}
\end{equation}
with analogous arguments giving rise to the same relation for $a = 1,3$.
Expanding out the commutator and using the self-inverse property of $B_a$ and $C_a$, this equation can be rearranged to give
\begin{equation}
    (B_aC_aB_aC_a + C_aB_aC_aB_a + 2) \ket{\psi} = 0,
\end{equation}
which, again using the self-inverse property, can be written in terms of the anticommutator
\begin{equation}
    \{ B_a, C_a \}^2 \ket{\psi} = 0 \implies \{ B_a, C_a \} \ket{\psi} = 0.
    \label{eqn:anticomm-constraint}
\end{equation}
A set of operators equivalent to Eq.~\eqref{eqn:GHZ-embedding} up to rotations and translations can be used to show that $\{ B_e, C_e \} \ket{\psi}=0$ and $\{ A_e , C_e\}\ket{\psi} = 0$ for all edges $e$. We have therefore shown that $B_e, C_e$ anticommute on $\ket{\psi}$, as do the operators $A_e , C_e$.

\emph{Relating products of operators}.---Next, we derive a relationship between the operators $A_e$, $B_e$, and $C_e$ when acting on $\ket{\psi}$. Consider for some fixed plaquette $p$ and neighboring vertex $v$ the operators
\begin{equation}
S_{12} \coloneqq
\begin{tikzpicture}[x=2.4ex,y=2.4ex,baseline={(0,0.3)}]
    \draw [line width=1.25, zpauli, line join=round] ({-1+1.5}, {-1+1.5}) rectangle ({1+1.5}, {1+1.5});
    \draw [line width=1.25, xpauli, line join=round] ({-1+0.5}, {-1+0.5}) rectangle ({1+0.5}, {1+0.5});
    \draw [line width=1, fill=white] (+1.5, 0.5) circle [radius=2pt];
    \draw [line width=1, fill=white] (0.5, 1.5) circle [radius=2pt];
\end{tikzpicture}
\qquad 
B_p \coloneqq
\begin{tikzpicture}[x=2.4ex,y=2.4ex,baseline={(0,0.3)}]
    \draw [line width=1.25, xpauli, line join=round] ({-1+0.5}, {-1+0.5}) rectangle ({1+0.5}, {1+0.5});
\end{tikzpicture}
\qquad 
A_v \coloneqq
\begin{tikzpicture}[x=2.4ex,y=2.4ex,baseline={(0,0.3)}]
    \draw [line width=1.25, zpauli, line join=round] ({-1+1.5}, {-1+1.5}) rectangle ({1+1.5}, {1+1.5});
\end{tikzpicture}
\label{eqn:XYZ-ops}
\end{equation}
which satisfy $S_{12}\ket{\psi} = -\ket{\psi}$ and $B_p \ket{\psi} = A_v\ket{\psi} = \ket{\psi}$ as a result of the victory conditions~\eqref{eqn:general-outcome-dist-nonoverlap}.
Thus, writing $S_{12}$ in terms of its constituent operators,
\begin{equation}
    \prod_{e \in \partial^\dagger v \setminus \partial p } \!\! A_e \prod_{e \in \partial p \setminus \partial^\dagger v} \!\! B_e  \prod_{e \in \partial^\dagger v \cap \partial p } \!\! C_e \ket{\psi} =
    - B_p A_v \ket{\psi}.
\end{equation}
By the self-inverse property, we can remove all $A_e$ and $B_e$ from the left-hand side, giving
\begin{equation}
    \prod_{e \in \partial^\dagger v \cap \partial p } C_e \ket{\psi} = \prod_{e \in \partial^\dagger v \cap \partial p } i B_e A_e  \ket{\psi}.
\end{equation}
An analogous sequence of steps can be followed for any two simply connected regions, one on the primary lattice and one on the dual lattice, whose boundaries intersect twice. Since, given an arbitrary pair of edges $(e, e')$, we can write down such regions whose boundaries intersect at $(e, e')$, we must have
\begin{equation}
    (iB_{e}C_{e}) \otimes (i B_{e'}C_{e'}) \ket{\psi} = A_e \otimes A_{e'} \ket{\psi},
    \label{eqn:XYZ}
\end{equation}
for all pairs $(e, e')$.

\emph{Constructing the factorization}.---Consider the constraint $\{ B_e, C_e \} \ket{\psi} = 0$. Performing a Schmidt decomposition of the state $\ket{\psi}$ between $e$ and its complement $\bar{e}$,  $\ket{\psi} = \sum_k \lambda_k \ket{\varphi_k}_e \otimes \ket{\chi_k}_{\bar{e}}$, we must have that $\Pi_e \{ B_e, C_e \} \Pi_e = 0$, where $\Pi_e$ is the local projector onto basis states with nonzero $\lambda_k$, i.e., $\Pi_e = \sum_{k : \lambda_k > 0} \ketbra{\varphi_k}{\varphi_k}$.
Next, introduce the projected operators $x_e$, $y_e$, and $z_e$ on every edge
\begin{equation}
    x_e \coloneq \Pi_e B_e \Pi_e, \quad y_e \coloneq \Pi_e C_e \Pi_e, \quad z_e \coloneq \frac{1}{2i}[x_e, y_e] , 
    \label{eqn:Pauli-op-embedding}
\end{equation}
Since $[\Pi_e, B_e]\ket{\psi}=[\Pi_e, C_e]\ket{\psi}=0$, we have $x_e^2 \ket{\psi} = y_e^2\ket{\psi} = \ket{\psi}$ and thus $x_e^2=y_e^2=\Pi_e$. Manipulations analogous to Eqs.~\eqref{eqn:self-testing-constraint}--\eqref{eqn:anticomm-constraint} can then be used to show that $\{x_e, y_e \} = 0$. Since the operators in Eq.~\eqref{eqn:Pauli-op-embedding} square to $I$ on the subspace restricted by $\{ \Pi_e \}$, pairwise anticommute and satisfy the $\mathfrak{su}(2)$ commutation relations, the restricted Hilbert space on edge $e$ decomposes into invariant subspaces labeled by representations of $\SU{2}$:
\begin{equation}
    \mathcal{H}_e = V_{j_1} \oplus V_{j_2} \oplus \dots
\end{equation}
with $V_j$ the spin-$j$ irrep. The anticommutation property ensures that only the spin-$1/2$ irrep can appear, from which it follows that $\mathcal{H}_e = \mathbb{C}^{d_e} \otimes V_{1/2}$ for some integer $d_e \geq 1$ specifying the multiplicity of the spin-$1/2$ irrep. We can therefore find a basis in the subspace restricted by $\{ \Pi_e \}$ in which
\begin{equation}
    x_e = I_{d_e} \otimes \sigma^x_e, \quad y_e = I_{d_e} \otimes \sigma^y_e, \quad z_e = I_{d_e} \otimes \sigma^z_e, 
    \label{eqn:spin-1/2-reps}
\end{equation}
With $\sigma_e^x$, $\sigma_e^y$, $\sigma_e^z$ the three Pauli matrices.
Introducing $D \coloneq \prod_e d_e$, we find that, in terms of these operators, the state must satisfy
\begin{equation}
    I_D \otimes \bigotimes_{e \in \partial p} \sigma^x_e W^\dagger \ket{\psi} = W^\dagger \ket{\psi}
    \label{eqn:self-test-X-constraint}
\end{equation}
for all $p$ for some edge-local unitary $W$ (i.e., factorizable over edges). To obtain the constraint involving $\sigma^z$, consider $A_v \ket{\psi} = \ket{\psi}$. Using Eq.~\eqref{eqn:XYZ}, this gives
\begin{equation}
    \prod_{e \in \partial^\dagger v} i B_e C_e \ket{\psi} = \ket{\psi} \implies 
    \prod_{e \in \partial^\dagger v} z_e \ket{\psi} = \ket{\psi}
\end{equation}
where we used $iy_e x_e = z_e$ from the definition in Eq.~\eqref{eqn:Pauli-op-embedding}. Using the decomposition~\eqref{eqn:spin-1/2-reps}, we are able to write down the final set of constraints that $\ket{\psi}$ must satisfy. For every vertex $v$,
\begin{equation}
    I_D \otimes \bigotimes_{e \in \partial^\dagger v} \sigma^z_e W^\dagger \ket{\psi} = W^\dagger \ket{\psi} 
    \label{eqn:self-test-Z-constraint}
\end{equation}
The solution to Eqs.~\eqref{eqn:self-test-X-constraint} and \eqref{eqn:self-test-Z-constraint} is given by $\ket{\psi} = W(\ket{\text{junk}}\otimes \ket{\text{TC}})$, where $\ket{\text{TC}}$ is a codeword of the toric code and $\ket{\text{junk}}$ is a generic state in the Hilbert space of dimension $D$ with nonvanishing Schmidt coefficients on every edge. Hence, players share a toric-code state up to an operator $W^{\dag}$ which is site local, i.e., expressible as a product of single-qubit operators. 

Note that we only used a small number of queries in order to self-test the toric code state. Namely, those corresponding to stabilizer generators [e.g., Eq.~\eqref{eqn:XYZ-ops}] and those corresponding to the diagrams show in Eq.~\eqref{eqn:GHZ-embedding}. We have therefore shown that the submeasurement game for the toric code naturally self-tests the state in a translation-invariant way. We leave the question of whether the self-testing result holds for states beyond the toric code to future work.


\subsection{Extensive contextuality quantification}
\label{sec: ExtensiveContextuality}

One of our original motivations was to use games to provide an operational quantification of contextuality for many-body quantum states. It would be desirable for any such quantification to be \emph{extensive}, i.e., scaling proportionally with the number of qubits involved. Submeasurement games allow us to construct properly extensive quantifications of contextuality, as we shall now illustrate, using the toric code game discussed above.

Consider an $L\times L$ patch of toric code. Carve it up into $\ell \times \ell$ blocks, with $\ell \gg 1$, and consider restricting inputs such that an independent toric code XOR game is played on each block.
The submeasurement game victory condition~\eqref{eqn:submeasurement-game-victory} then requires that \emph{every} XOR game is won simultaneously. It then immediately follows that the classical victory probability for the submeasurement game $\omega_L$ is written in terms of the win probability of the block $\omega_\ell$ as $\log\omega_L = (L/\ell)^2 \log\omega_\ell$, whereas with access to quantum resources, we still have a quantum victory probability $\omega^*=1$. We can then straightforwardly define the quantity
\begin{equation}
    \log(\omega^*/\omega) \propto (L/\ell)^2, 
    \label{eq: extensivecontextuality}
\end{equation}
which, using the results of Sec.~\ref{sec:quantifying-contextuality}, serves as an extensive quantification of contextuality, being related to the number of contexts whose measurement outcomes can be correctly described by a hidden-variable model per unit area. We expect that a similar result holds for the toric code submeasurement game with no input restrictions. Note that the translational invariance of the submeasurement game (Def.~\ref{def:submeasurement-game}) means that the possible inputs to each block are identical, which is crucial for the interpretation of Eq.~\eqref{eq: extensivecontextuality} as a density. This contrasts with previous similar constructions of games or Bell inequalities~\cite{Baccari2020,selftest1,selftest2}, which are not translationally invariant even if the underlying CSS state is.


\section{Discussion}
\label{sec: discussion}

In this manuscript we have discussed classes of multiplayer quantum games that can be won with certainty, given single-site Pauli measurement access to a resource state that is a CSS codeword. We have related success at these games to the evaluation of a code-dependent Boolean function, and have shown that the success fraction for optimal classical strategies can be bounded by the degree of nonlinearity of the Boolean function, which in turn may be quantified by its Walsh-Hadamard transform, and efficiently evaluated by identifying the $X$ symmetries of an auxiliary hypergraph state. We have argued that this provides a route to the quantification of the contextuality of (single-site Pauli measurements on) a CSS codeword. We have further illustrated these results with explicit computations for a number of paradigmatic CSS codewords, including GHZ states, 1D cluster states, and 2D toric code states. We believe we have introduced a powerful new perspective on both quantum games and quantum contextuality. 

We have also introduced the notion of \emph{submeasurement games}, where the players are required to satisfy multiple overlapping constraints. The additional structure inherent in submeasurement games affords considerable advantages. For one, it opens the door to self-testing \cite{selftest1, selftest2} via the design of games for which perfect strategies can be achieved if and only if the resource state equals a specific CSS codeword state (up to local isometries). We have illustrated this with a game that can \emph{only} be won with certainty using a resource state that is equivalent to a toric code state up to local isometries. A submeasurement game of this type can be used as a ``fingerprint'' for toric code states. It appears likely that similar ``fingerprinting'' games can be designed for arbitrary CSS codewords, and developing a general framework for identifying the appropriate game for any given codeword would be a worthwhile direction for future work. Separately, we have also shown how submeasurement games enable a properly \emph{extensive} quantification of contextuality, see Sec.~\ref{sec: ExtensiveContextuality}.

One could ask whether the ideas introduced herein may be extended beyond CSS codeword states. One path to doing so may involve incorporating error correction. That is, an imperfect CSS codeword may be viewed as a parent codeword plus errors, which are removed under the action of some decoder if the error rate is sufficiently low. The results introduced herein can then be applied to the error-corrected state. However, how to do this in practice is presently unclear. As a stepping stone, it may be fruitful to consider states with less structure including subsystem codes~\cite{Poulin2005stabilizer}, or by going beyond qubit stabilizer codes altogether.

While the existing literature on quantum games and self-testing has largely focused on games \emph{without} communication, one could also ask about games where a limited amount of communication between players is permitted (see Refs.~\cite{Murta2023selftesting,meyer2024selftesting}) based, e.g., on spatial locality. While it is still possible to obtain quantum advantage with limited classical communication~\cite{Daniel2022Exp} (see also Refs.~\cite{bravyi2020shallow,bravyi2020noisyShallow,Watts2019exponential,LeGall2019averagecase,Grier2020interactive,Daniel2022Exp,caha2023colossal} for the connection between quantum games and unconditional quantum advantage in shallow quantum circuits), it would be desirable to develop a framework analogous to the one presented herein, both to develop intuition and to increase the separation between optimal quantum and classical strategies. We leave this and other related questions, such as the inclusion of noise~\cite{fu2025gamesselftests}, to future work.


\textbf{Notes.}
The posting of this preprint to the arXiv was coordinated with a simultaneous submission by Zhao \emph{et al.}~\cite{ZhaoInPrep}, which was developed independently from the work presented here. They also evaluate the optimal success fraction for the cyclic cluster state and the toric code, and, where our results overlap, they agree.
The code used to prepare the figures in this manuscript is available in Ref.~\cite{github}.


\begin{acknowledgements}
We thank Jan Behrands and Shival Dasu for helpful feedback on the manuscript. DTS acknowledges support from the Simons Collaboration on Ultra-Quantum Matter, which is a grant from the Simons Foundation (651440).
\end{acknowledgements}   


\appendix


\section{Function computed by CSS game}
\label{app:computed-function}

\begin{proof}[Proof of Lemma~\ref{lem:oplus-to-reals}]
    The sum $\bigoplus_{i=1}^n x_i$ evaluates to 0 if $\vec{x}$ contains an even number of ones and evaluates to 1 otherwise. In the following, all arithmetic is expressed over $\Reals$. Hence,
    \begin{equation}
        \bigoplus_{i=1}^n x_i = \frac12 \left[ 1-(-1)^{\sum_{i=1}^n x_i} \right] . 
    \end{equation}
    Using $(-1)^{x_i} = 1-2x_i$ and expanding out the product 
    \begin{equation}
        \bigoplus_{i=1}^n x_i = \frac12 \left[ 1 - \sum_{I \subseteq \{ 1, \dots, n \} } (-2)^{\abs{I}}\prod_{i \in I} x_i \right] .
    \end{equation}
    The sum over all subsets $I \subseteq \{1, \dots, n \}$ can alternatively be written as a sum over all bitstrings $\vec{b}\in \F{2}^n$, where the nonzero entries of $\vec{b}$ correspond to the elements of $\{ 1, \dots, n \}$ that are included in $I$. This gives
    \begin{equation}
        \bigoplus_{i=1}^n x_i = \frac12 \left[ 1 - \sum_{\vec{b} \in \F{2}^n } (-2)^{\weight(\vec{b})}\prod_{i =1}^n x_i^{b_i} \right] .
    \end{equation}
    The contribution from $\vec{b}=\vec{0}$ then cancels with $2^{-1}$. Note that all $\vec{b}\in\F{2}^n\setminus \vec{0}$ then obey $\weight(\vec{b}) \geq 1$ such that all corresponding coefficients are elements of $\Z{}$. Thus, we arrive at the final result
    \begin{equation}
        \bigoplus_{i=1}^n x_i =
        \sum_{\vec{b} \in \F{2}^n \setminus \vec{0}} (-2)^{\weight(\vec{b})-1} \prod_{i=1}^n x_i^{b_i} 
        .
    \end{equation}
    If we consider the value of the right-hand side mod 2, the expression reduces to $\sum_{i=1}^n x_i$, as expected.
\end{proof}

To construct an explicit expression for the Boolean function $f_H(\vec{x}, \vec{z})$ entering the win condition~\eqref{eqn:CSS-game-victory}, we are interested the value of
\begin{equation}
    \frac12 \sum_{i\in\Lambda} a_i b_i = \frac12 \sum_{i\in\Lambda} [\vec{x}H_X]^{\oplus}_i [H_Z^\tp \vec{z}^\tp]^{\oplus}_i
    \label{eqn:function-mod-4}
\end{equation}
modulo 2, expressed in algebraic normal form (ANF). Recall that $\Lambda$ is the set of all lattice sites and we use the notation $[\vec{v}A]_i^\oplus$ to denote matrix-vector multiplication over $\F{2}$ between the matrix $A$ and the row vector $\vec{v}$, expressed as an element of $\Reals$. We thus wish to evaluate
\begin{equation}
    \sum_{i\in\Lambda} [\vec{x}H_X]^{\oplus}_i [H_Z^\tp \vec{z}^\tp]^{\oplus}_i \mod 4 .
\end{equation}
Applying Lemma~\ref{lem:oplus-to-reals} we have that
\begin{equation}
    \bigoplus_{i=1}^n x_i \equiv_{(\text{mod 4})}
        \sum_{\substack{ \vec{b} \in \F{2}^n \setminus \vec{0} \\ \weight(\vec{b}) \leq 2  } } (-2)^{\weight(\vec{b})-1} \prod_{i=1}^n x_i^{b_i}  ,
\end{equation}
since all higher-degree terms evaluate to zero modulo 4. More generally, evaluating Lemma~\ref{lem:oplus-to-reals} modulo $2^k$ will restrict the summation over bitstrings to those satisfying $\weight(\vec{b}) \leq k$. We therefore have
\begin{equation}
    [\vec{x}H_X]^{\oplus}_i \equiv_{(\text{mod 4})}  \sum_\alpha x_\alpha [H_X]_{\alpha i} + 2 \sum_{\alpha < \gamma} x_\alpha x_\gamma [H_X]_{\alpha i} [H_X]_{\gamma i} 
\end{equation}
where the second summation is over all unordered pairs,
with an analogous expression involving $H_Z$. Inserting this result into Eq.~\eqref{eqn:function-mod-4} and summing over all $i \in \Lambda$, we have
\begin{equation}
    f_H(\vec{x}, \vec{z}) = 
    \bigoplus_{I \subseteq G : \abs{I}=2,3} b_I  w^I
\end{equation}
where $\vec{w} = (\vec{x}, \vec{z})$, and $I \subseteq G$ correspond to subsets of stabilizer generators. Only terms of algebraic degree two or three contribute to $f_H$. The coefficients $b_I$ are \emph{only nonzero when $I$ contains at least one $X$-type and one $Z$-type stabilizer generator}, in which case they equal
\begin{align}
    b_{\abs{I}=2} &\equiv_{(\text{mod 2})} \weight \left[ H_{I_1} \wedge H_{I_2} \right] / 2 \label{eqn:deg(2)-coeff} \\
    b_{\abs{I}=3} &\equiv_{(\text{mod 2})} \weight \left[ H_{I_1} \wedge H_{I_2} \wedge H_{I_3} \right] \label{eqn:deg(3)-coeff}
\end{align}
where $H_m = \vec{e}_m (H_X, H_Z)$ with $\vec{e}_m$ the basis vector corresponding to stabilizer generator $m$, and $\wedge$ denotes ``and.'' Note that the degree-2 term can alternatively be written
\begin{equation}
    \bigoplus_{\alpha, \beta} b_{\alpha\beta} x_\alpha z_\beta
\end{equation}
where $b_{\alpha\beta} \equiv \frac12 \sum_{i \in \Lambda} [H_X]_{\alpha i} [H_Z^\tp]_{i \beta}$ (mod 2).


\section{Code equivalence}
\label{app:equivalent-codes}

\begin{lemma}[e.g., Ref.~\cite{Bravyi2021HadamardFree}, Lemma 3]
    \label{lem:C1-representative-elements} 
    Any element $C$ of the Clifford group $\mathbf{C}_1 / \U{1}$ can be written as $C = V P$, where $P \in \{ I, X, Y, Z \}$ is a Pauli matrix and $V \in \{ I, H, S, HS, SH, HSH \}$, where $H$, $S$ are the Hadamard and phase gates, respectively.
\end{lemma}

To prove Theorem~\ref{thm:clifford-equivalence}, we need to know how an arbitrary Pauli matrix $P'$ transforms under conjugation, $C P' C^\dagger$, for any element $C \in \mathbf{C}_1 / \U{1}$. Since $C = V P$ from Lemma~\ref{lem:C1-representative-elements}, we can consider conjugation by $P$ and $V$ separately. Given $\vec{u} = (u_0, u_1)$, writing $P(\vec{u}) \coloneq i^{u_0 u_1} X^{u_0} Z^{u_1}$, anticommutation of Pauli matrices gives rise to: 
\begin{equation}
    P({\vec{u}}) P({\vec{v}}) P({\vec{u}}) = (-1)^{[\vec{u}, \vec{v}]} P({\vec{v}}) 
    \label{eqn:pauli-commutation-symplectic}
\end{equation}
with $[\vec{u}, \vec{v}] = u_0v_1 \oplus v_0u_1$ the symplectic inner product. The adjoint action of $V$, on the other hand, leads to $P({\vec{s}}) \mapsto P({\vec{s}C_V})$ (up to a sign) with $C_V \in \Sp(2, \F{2})$, the group of 2-dimensional symplectic matrices over $\F{2}$:
\begin{equation}
    V P({\vec{s}}) V^\dagger =(-1)^{\lambda_V(\vec{s})} P({\vec{s} C_V}) ,
    \label{eqn:V-adjoint-action}
\end{equation}
where the sign is given by $\lambda_S(\vec{s}) = \lambda_{H}(\vec{s}) = s_0 s_1$ and
\begin{equation}
    C_H =
    \begin{pmatrix}
        0 & 1 \\
        1 & 0
    \end{pmatrix}
    ,
    \qquad
    C_S =
    \begin{pmatrix}
        1 & 1 \\
        0 & 1
    \end{pmatrix}
    .
    \label{eqn:Sp2-generators}
\end{equation}
These expressions follow immediately from the adjoint action of $H$ and $S$ on Pauli matrices. By direct calculation, $H(X, Y, Z)H = (Z, -Y, X)$ and $S(X, Y, Z)S^{\dagger} = (Y, -X, Z)$. Note the minus signs for $Y$, which are responsible for the factor $(-1)^{\lambda_V}$ in Eq.~\eqref{eqn:V-adjoint-action}. The matrices $C_H$ and $C_S$ generate $\Sp(2, \F{2})$.

\begin{proof}[Proof of Theorem~\ref{thm:clifford-equivalence}]
    We proceed by explicitly computing the function realized by measuring single-site Pauli operators in the state $U\ket{\psi}$ for any codeword $\ket{\psi}$. Using Lemma~\ref{lem:submeasurement-dist}, measurement of single-site Pauli operators $P_i$ on every lattice site $i \in \Lambda$ results in a deterministic outcome if
    \begin{equation}
        \prod_{i \in \Lambda} P_i = \pm S,
        \label{eqn:deterministic-sign}
    \end{equation}
    for some element $S \in \mathcal{S}$ of the code's stabilizer group $\mathcal{S}$. Since we are interested in the function computed by \emph{any} codeword $\ket{\psi}$, we do not consider logical operators on the right-hand side of Eq.~\eqref{eqn:deterministic-sign}. Hence, given a basis, we decompose a generic element of the stabilizer group into its constituent Pauli operators, up to a sign:
    \begin{equation}
        \prod_{\alpha} {B'_\alpha}^{x_\alpha} \prod_{\beta} {A'_\beta}^{z_\beta} =  (-1)^{g(\vec{x}, \vec{z})}  \prod_{i \in \Lambda} P_i[a'_i(\vec{x}), b'_i(\vec{z})]
        \label{eqn:function-computation-sufficient}
    \end{equation}
    where $B'_\alpha = U B_\alpha U^\dagger$ with $B_\alpha$ an $X$-like stabilizer generator, and likewise for each dressed $Z$-like generator $A_\beta$. The Boolean function computed is $g(\vec{x}, \vec{z})$.
    Our goal is to determine the function $g$ and the bits $a_i'$, $b_i'$.
    By definition of the parity-check matrices, the stabilizer generators can be decomposed as
    \begin{equation}
        B_\alpha = \prod_{i\in \Lambda} X_i^{[H_X]_{\alpha i}},
        \quad
        A_\beta = \prod_{i\in \Lambda} Z_i^{[H_Z]_{\beta i}}
    \end{equation}
    Therefore, introducing the bitstrings $\vec{a}(\vec{x}) = \vec{x}H_X$ and $\vec{b}(\vec{z}) = \vec{z}H_Z$, as in Eq.~\eqref{eqn:input-bits}, the left-hand side of Eq.~\eqref{eqn:function-computation-sufficient} decomposes over lattice sites as
    \begin{equation}
        \prod_{\alpha} {B'_\alpha}^{x_\alpha} \prod_{\beta} {A'_\beta}^{z_\beta}
        =
        \prod_{i \in \Lambda} C_i X_i^{a_i} Z_i^{b_i} C_i^\dagger .
    \end{equation}
    Introducing $\vec{s}_i = (a_i, b_i) \in \F{2}^2$, we can now use Eqs.\ \eqref{eqn:pauli-commutation-symplectic}--\eqref{eqn:Sp2-generators} to determine the sign accumulated under conjugation by $U$. Write each $C_i = V_i P(\vec{u}_i)$. Then
    \begin{equation}
        \prod_{i \in \Lambda} C_i X_i^{a_i} Z_i^{b_i} C_i^\dagger = 
        (-1)^{g(\vec{x}, \vec{z})} \prod_{i \in \Lambda} P_i (a_i', b_i')
    \end{equation}
    where we defined $(a_i', b_i') = \vec{s}' = \vec{s} C_{V_i}$ and the function
    \begin{equation}
        g(\vec{x}, \vec{z}) = f_H(\vec{x}, \vec{z}) \oplus \bigoplus_{i \in \Lambda} \left( [\vec{u}_i, \vec{s}_i] \oplus \lambda_{V_i}(\vec{s}_i) \right)
        \label{eqn:equivalent-function}
    \end{equation}
    where $f_H$ is the target function~\eqref{eqn:CSS-game-victory}.
    Thus, measurement of $P_i(a'_i, b'_i)$ $\forall i \in \Lambda$ computes $g(\vec{x}, \vec{z})$.
    The $U$-dependent term on the right-hand side of Eq.~\eqref{eqn:equivalent-function} corresponds to a generic deterministic classical strategy. The term $[\vec{u}_i, \vec{s}_i]$ is the most general linear function of the bits $a_i$, $b_i$, while the term $\lambda_V$ can give rise to the product $a_i b_i$.
\end{proof}


\section{GHZ game Walsh spectrum}
\label{app:walsh-spectrum-ghz}


\subsection{Open boundaries}

\begin{proposition}
    \label{prop:ghz-bent}
    The GHZ function~\eqref{eqn:ghz-function} acting on $n$ variables (with $n$ even) is bent. For general $n \geq 1$, the GHZ function has nonlinearity $N_f = 2^{n-1} - 2^{ \lceil n/2 \rceil -1}$.
\end{proposition}

\begin{proof}
The Walsh transform (Def.~\ref{def:Walsh-transform}) of the GHZ function on $n$ variables with generic open boundary conditions $z_0 = l \in \F{2}$ and $z_{n+1} = r \in \F{2}$ is given by
\begin{equation}
    [\mathcal{Z}_n(\vec{x})]_{lr} \coloneq \sum_{z_i \in \mathbb{F}_2} (-1)^{  \sum_{i=0}^{n}  z_i z_{i+1} + \sum_{i=1}^{n} (1+x_i) z_i  }
    ,
    \label{eqn:walsh-partition-function}
\end{equation}
where only $\{ z_1, \dots, z_n \}$ are summed over, i.e., excluding the fixed boundary conditions. The left-hand side of Eq.~\eqref{eqn:walsh-partition-function} can be regarded as the matrix elements of a $2\times 2$ matrix. The full matrix $\mathcal{Z}_n(\vec{x})$ can be evaluated by writing it as a product of $2 \times 2$ transfer matrices, which evaluate to Hadamard $H$ and Pauli $Z$ matrices:
\begin{subequations}
\begin{align}
    \mathcal{Z}_n(\vec{x}) &= 2^{N/2} [H Z^{1+x_1} H Z^{1+x_2} H \cdots H Z^{1+x_{n}} H] \\
    &= \pm X^{\sum_{k} (1+x_{2k+1}) } Z^{\sum_{k} (1+x_{2k}) } (2^{N/2} H^N) \label{eqn:partition-function-final}
\end{align}%
\end{subequations}
where the overall sign arises from commuting all $X$'s to the left of $Z$'s.
Note that the $2 \times 2$ Hadamard matrix is normalized such that $H^2 = I$ and that $N = n+1$. Taking the $(l,r)=(0,0)$ matrix element of this expression allows us to obtain the desired Walsh spectrum. Depending on the parity of $N$, the matrix $2^{N/2}H^N$ evaluates to
\begin{equation}
    2^{N/2}H^N = 2^{\lfloor N/2 \rfloor} \times 
    \begin{cases}
        I &\text{$N$ even}, \\
        \sqrt{2}H &\text{$N$ odd}.
    \end{cases}
    \label{eqn:Hadamard-powers}
\end{equation}
Since the $X$ and $Z$ Pauli matrices in Eq.~\eqref{eqn:partition-function-final} only permute matrix elements or change their sign, respectively, we deduce that the Walsh spectrum takes the values 
\begin{equation}
    W_\mathsf{GHZ}(\vec{x}) \in 
    \begin{cases} 
        \left\{ 0, \pm 2^{ N/2} \right\} &\text{$N$ even}, \\
        \left\{ \pm 2^{\lfloor N/2 \rfloor}\right\}
        & \text{$N$ odd}.
    \end{cases}
    \label{eqn:GHZ-walsh-spectrum}
\end{equation}
Hence, the Walsh spectrum is flat for $N$ odd (i.e., $n$ even), and it follows that the function is bent. The magnitude of the maximal Walsh coefficient is equal to $2^{\lfloor N/2 \rfloor} = 2^{\lceil n/2 \rceil}$. From Eq.~\eqref{eqn:nonlinearity-from-walsh} we deduce that $N_f = 2^{n-1} - 2^{\lceil n/2 \rceil - 1}$.
\end{proof}

These results corroborate the understanding of the Walsh spectrum from the perspective of cluster-state symmetries presented in Sec.~\ref{sec:walsh-from-graphs}. When $n$ is even, the associated 1D cluster state has no $X$ symmetries and thus by Theorem~\ref{thm:walsh-from-symmetries} the largest Walsh coefficient has magnitude $\abs{W} = 2^{n/2}$. When $n$ is odd, there is one nontrivial $X$ symmetry and from Theorem~\ref{thm:walsh-from-symmetries} we have $\abs{W} = 2^{(n+1)/2}$. Combining these expressions, the largest Walsh coefficient is given by $\abs{W} = 2^{\lceil n/2 \rceil}$, in agreement with Eq.~\eqref{eqn:GHZ-walsh-spectrum}.


\subsection{Periodic boundaries}

We also compute the Walsh spectrum of the GHZ function with periodic boundary conditions. The associated Walsh coefficient is given by
\begin{equation}
    \sum_{z_i \in \mathbb{F}_2} (-1)^{  \sum_{i=1}^{n}  z_i z_{i+1} + \sum_{i=1}^{n} (1+x_i) z_i  } = \Tr[\mathcal{Z}'_{n}(\vec{x})].
\end{equation}
where we introduced the product of transfer matrices
\begin{subequations}
\begin{align}
    \mathcal{Z}'_{n}(\vec{x}) &= 2^{n/2} H Z^{1+x_1} H Z^{1+x_2} \cdots H Z^{1+x_n} \\
    &= \pm X^{\sum_{k} (1+x_{2k+1}) } Z^{\sum_{k} (1+x_{2k}) } (2^{n/2} H^n). \label{eqn:partition-function-final-pbc}
\end{align}%
\end{subequations}
We now use the fact that $\Tr[X^a Z^b] = 2\delta_a \delta_b$, and that $\Tr[X^a Z^b \sqrt{2}H] = 2\delta_{a \oplus b}$ . The corresponding Walsh spectrum therefore always contains zeros; its entries are
\begin{equation}
    W_{\mathsf{GHZ}'}(\vec{x}) \in \left\{ 0, \pm 2^{\lfloor n/2 \rfloor + 1}  \right\}.
\end{equation}


\subsection{Matrix element bounds}

\zremove*

\begin{proof}
    From Theorem~\ref{thm:walsh-from-symmetries}, we know that all nonzero Walsh coefficients have the same magnitude. We can therefore prove the Lemma by showing that $\braket{+_\ell | C_\ell}$ is positive for all $\ell \geq 2$ and that $\braket{+_\ell  | C'_\ell}$ is positive for all even $\ell \geq 4$. 

    First, consider Eq.~\eqref{eqn:z-removal-obc}. From Eq.~\eqref{eqn:partition-function-final} we have that $\mathcal{Z}_\ell(\vec{0}) = 2^{(\ell+1)/2}H^{\ell + 1}$, which is evaluated according to Eq.~\eqref{eqn:Hadamard-powers}. Since $\braket{0 | H | 0} > 0$ and $\braket{0 | I | 0} > 0$, we have $\braket{+_\ell  | C_\ell} > 0$, as required.

    For Eq.~\eqref{eqn:z-removal-pbc}, the overlap $\braket{+\ell  | C'_\ell}$ is calculated using Eq.~\eqref{eqn:partition-function-final-pbc}, which gives $\mathcal{Z}'_\ell(\vec{0}) = 2^{\ell/2}H^\ell$. Now, for $\ell$ even, we have $H^\ell = I$ and $\Tr I > 0$. It follows that $\braket{+_\ell  | C'_\ell} > 0$, as required.
\end{proof}


\section{The algorithm \texttt{StandardForm}}

The algorithm that brings a quadratic Boolean function into standard form is presented as Algorithm~\ref{alg:standard-form}.
As an input, this algorithm takes the matrix defined by the polar bilinear form associated to the function. As an output, it gives the reduced polar bilinear form and the the linear coordinate transformation that produces it. Note that the rank of the polar bilinear form, which can immediately be read off from its reduced form, is always even, consistent with Def.~\ref{def:standard-form}. The algorithm makes use of the matrices (with $i < j$ for concreteness)
\begin{align}
    X_{[ij]} \vec{x}^\tp &= (x_1, \dots, x_i \oplus x_j, \dots, x_j, \dots, x_n)^\tp \\
    S_{[ij]} \vec{x}^\tp &= (x_1, \dots, x_j, \dots, x_i, \dots, x_n)^\tp
\end{align}
which add row $j$ to row $i$, and interchange rows $i$ and $j$, respectively.

\begin{algorithm}[t]
\caption{Adjacency matrix to standard form}
\label{alg:standard-form}
\KwIn{$\mathbf{B}$, an adjacency matrix}
\KwOut{$(\mathbf{B}', \mathbf{A})$ where \\
\quad $\mathbf{B}'$, a reduced matrix congruent to $\mathbf{B}$, and\\
\quad $\mathbf{A} \in \mathrm{GL}_n(\mathbb{F}_2)$ satisfying $\mathbf{B}' = \mathbf{A} \mathbf{B} \mathbf{A}^\tp$}

$\mathbf{B}' \leftarrow \mathbf{B}; \quad \mathbf{A} \leftarrow \mathbf{I}; \quad j \leftarrow 0; \quad e \leftarrow n-1$\;
\While{$j < e$}{
    \tcp{ Move empty rows/columns to the end }
    \If{$\#\{ i \mid B'_{ij} = 1 \} = 0$}{
        $\mathbf{B}' \leftarrow S_{[ej]} \mathbf{B}' S_{[j e]}$; 
        \quad $\mathbf{A} \leftarrow \mathbf{A} \, S_{[je]}$\;
        $e \leftarrow e-1$\;
        \textbf{continue}\;
    }
    $k \leftarrow \min\{ i \mid B'_{ij} = 1 \}$; \quad $p \leftarrow j+1$\;
    $\mathbf{B}' \leftarrow S_{[kp]} \mathbf{B}' S_{[p k]}$;
    \quad $\mathbf{A} \leftarrow \mathbf{A} \, S_{[p k]}$\;

    \tcp{ Remove remaining 1s in row/column $j$ }
    \For{$r \leftarrow p+1$ \KwTo $e$}{
        \If{$B'_{rj} = 1$}{
            $\mathbf{B}' \leftarrow X_{[rp]} \mathbf{B}' X_{[pr]}$;
            \quad $\mathbf{A} \leftarrow \mathbf{A} \, X_{[pr]}$\;
        }
    }

    \tcp{ Remove remaining 1s in row/column $p$ }
    \For{$c \leftarrow j+1$ \KwTo $e$}{
        \If{$B'_{pc} = 1$}{
            $\mathbf{B}' \leftarrow X_{[cj]} \mathbf{B}' X_{[j c]}$;
            \quad $\mathbf{A} \leftarrow \mathbf{A} \, X_{[jc]}$\;
        }
    }
    $j \leftarrow j+2$\;
}
\Return{$(\mathbf{B}', \mathbf{A}$})\;
\end{algorithm}


\section{Toric code integral}
\label{app:toric-code-integral}

In this Appendix we analytically evaluate the integral that bounds the performance of classical strategies for the toric code CSS game. We proceed using integration methods developed in, e.g., Ref~\cite{blagouchine2014rediscovery}. The integral we wish to evaluate is
\begin{subequations}
\begin{align}
    \ln W &= \frac12 \int_{-\infty}^{\infty} \frac{dx}{x} \frac{\sinh^2 \frac{\pi }{4}x }{ \sinh \frac{\pi  x}{2} (2 \cosh \frac{\pi }{4}x -1)} \\
    &= \frac14 \int_{-\infty}^{\infty} \frac{dx}{x} \frac{\tanh \frac{\pi }{4}x }{ 2 \cosh \frac{\pi }{4}x -1} \\
    &= \frac14 \int_{-\infty}^{\infty} \frac{dx}{x} \frac{\tanh 2\pi x }{ 2 \cosh 2 \pi x -1}
\end{align}%
\end{subequations}
In the second line we used $\sinh2x = 2\sinh x \cosh x$, and in the third line we redefined $ x \to 8 x$. To facilitate the computation of $\ln W$, we consider a one-parameter family of integrals. Consider the two generalized integrals $I(a)$, $J(a)$ parametrized by $a \in \Reals$:
\begin{subequations}
\begin{align}
    I(a) &= \int_{-\infty}^{\infty} dx \,\frac{x}{a^2+x^2} \frac{\tanh 2\pi x }{ 2 \cosh 2 \pi x -1} \\
    J(a) &= \int_{-\infty}^{\infty} dx \, \ln(a - ix) \frac{\tanh 2\pi x }{ 2 \cosh 2 \pi x -1}
\end{align}%
\end{subequations}
Note that $I(a) = \Im \partial_a J(a)$, while $I(0) = 4 \ln W$ encodes the integral of interest.
The integral $J(a)$ can be evaluated via contour integration. Consider the closed contour $\mathcal{C}$ defined by the rectangular path $-R \to R \to R+i \to -R+i \to -R$. Using the periodicity $\cosh 2\pi(x+i) = \cosh 2\pi x$, $\tanh 2\pi(x+i) = \tanh 2\pi x$, and 
\begin{equation*}
    \ln(a - ix) = \ln\Gamma(a-i(x+i)) - \ln\Gamma(a-ix)
\end{equation*}
we deduce that, in the limit of large $R$, $J(a)$ can be written
\begin{equation}
    J(a) = \oint_\mathcal{C} dz \, \ln\Gamma(a - iz) \frac{\tanh 2\pi z }{ 1-2 \cosh 2 \pi z }.
\end{equation}
The contribution from the short sides of the rectangular contour vanish since the integrand vanishes there as $R\to \infty$.
The integrand possesses simple poles along the imaginary axis at $z \in P = \{ i/6, i/4, 3i/4, 5i/6 \}$ encircled by $\mathcal{C}$. The function $\ln\Gamma(a-iz) \eqcolon \Lambda(a-iz)$ is nonsingular.
By the residue theorem, $J(a)$ evaluates to
\begin{equation*}
    J(a) = i \left[ \Lambda(a+\tfrac14)+\Lambda(a+\tfrac34) - \Lambda(a+\tfrac16) - \Lambda(a+\tfrac56) \right].
\end{equation*}
Using the definition of the digamma function $\psi(z) = \partial_z \ln\Gamma(z)$, we arrive at the final result 
\begin{equation}
    I(a) = \psi(a+\tfrac14)+\psi(a+\tfrac34) - \psi(a+\tfrac16) - \psi(a+\tfrac56),
\end{equation}
which satisfies $I(0)=\ln(27/4)$. Thus, we have the exact result $\ln W = \ln(27/4)/4$, as used in the main text.


\section{Proof of Lemma~\ref{lem:outcome-distribution-general}}
\label{app:outcome-dist-proof}

\measDistribution*

\begin{proof}
    Since $O_i^2 = I$, the probability of obtaining a specific outcome bitstring $\vec{s}$ in the state $\ket{\psi}$ is
    \begin{equation}
        \Pr(\vec{s}) = 
        \braket{ \psi | \prod_{i \in \Lambda} \frac{1}{2}\left[ I + (-1)^{s_i} O_i \right] | \psi } \eqcolon \braket{ \psi | \prod_{i \in \Lambda} \pi_{s_i} | \psi },
    \end{equation}
    where we introduced the projector $\pi_{s_i}$ onto the subspace associated to eigenvalue $(-1)^{s_i}$. The required probability is then found by summing over the probabilities of obtaining measurement outcome bitstrings consistent with the constraints:
    \begin{equation}
        \Pr\left(\bigwedge_{c \in C} \bigoplus_{i \in c} s_i = f_c\right) = \sum_{\vec{s}} \Pr(\vec{s}) \prod_{c \in C}\delta\left( \bigoplus_{i \in c} s_i = f_c \right) 
        \label{eqn:probability-sum-masked}
    \end{equation}
    To evaluate this expression, we write the Kronecker delta associated to constraint $c \in C$ in terms of a sum over an auxiliary variable $r_c \in \F{2}$
    \begin{equation}
        \delta\left( \bigoplus_{i \in c} s_i = f_c \right) = \frac{1}{2} \sum_{r_c \in \F{2}} (-1)^{r_c(f_c \oplus \bigoplus_{i \in c} s_i)}.
    \end{equation}
    Substitute this expression into Eq.~\eqref{eqn:probability-sum-masked}, expand $\Pr(\vec{s})$ in terms of the projectors $\pi_{s_i}$, and redistribute the summation over the variables $s_i$ to give
    \begin{equation}
        \Pr = \frac{1}{2^{\abs{C}}} \sum_{ \{ r_c \} } (-1)^{\bigoplus_{c \in C} r_c f_c} \braket{\psi | \prod_{i \in \Lambda} \sum_{s_i} (-1)^{s_i \bigoplus_{c \ni i} r_c}  \pi_{s_i} | \psi}
        \label{eqn:probability-sum-redistribute}
    \end{equation}
    To perform the summation over the variables $s_i$, make use of the identity
    \begin{equation}
        \frac{1}{2} \sum_{s_i  \in \F{2}} (-1)^{\lambda s_i} [ I + (-1)^{s_i} O_i ] = O_i^\lambda
        \label{eqn:Z2-orthogonality}
    \end{equation}
    for $\lambda \in \F{2}$. If a site does not participate in any constraints, the sum evaluates to $O_i^0 = I$ and is therefore trivial. Substituting the result~\eqref{eqn:Z2-orthogonality} into Eq.~\eqref{eqn:probability-sum-redistribute}, we arrive at the general expression~\eqref{eqn:general-outcome-dist}. To obtain the simplified expression~\eqref{eqn:general-outcome-dist-nonoverlap}, we use the simplification that, when the sets $I_k$ do not overlap, each site $i$ belongs to at most one constraint. Thus, 
    \begin{equation}
        \braket{ \psi | \prod_{i \in \Lambda} O_i^{\bigoplus_{c \ni i} r_c} | \psi } = 
        \braket{ \psi | \prod_{c \in C} \prod_{i \in c} O_i^{r_c} | \psi } 
    \end{equation}
    Substituting into Eq.~\eqref{eqn:Z2-orthogonality}, we can then perform the summation over the variables $r_c$ giving
    \begin{align}
        \Pr &= \frac{1}{2^{\abs{C}}} 
        \braket{ \psi | \prod_{c \in C} \sum_{r_c} (-1)^{r_c f_c}  \prod_{i \in c} O_i^{r_c} | \psi } \\
        &= \braket{ \psi | \prod_{c \in C} \frac12 \left( I + (-1)^{f_c} \prod_{i \in c} O_i \right) | \psi }
    \end{align}
    which is equivalent to Eq.~\eqref{eqn:general-outcome-dist-nonoverlap} up to a relabeling of the elements of $C$.
\end{proof}

\bibliographystyle{quantumsans}
\bibliography{biblio}

\end{document}